%% file: main.tex
\documentclass[mnsc]{informs3arxiv} 

\OneAndAHalfSpacedXI 

\usepackage{natbib}
 \bibpunct[, ]{(}{)}{,}{a}{}{,}%
 \def\bibfont{\small}%
\usepackage{tocloft}
\usepackage{lscape}
\newcommand{\new}[1]{{\color{black}{#1}\color{black}}}

\TheoremsNumberedThrough     
\ECRepeatTheorems

\EquationsNumberedThrough    

\MANUSCRIPTNO{MS-0001-1922.65}

\newif\ifanon
\anonfalse
\input{notation.tex}
\hypersetup{ hidelinks }

\begin{document}

\RUNAUTHOR{Liao and Kroer}

\RUNTITLE{Statistical Inference in Market Equilibrium}

\TITLE{
    Statistical Inference and A/B Testing in Fisher Markets and Paced Auctions
}

\ARTICLEAUTHORS{%
\AUTHOR{Luofeng Liao}
\AFF{Department of Industrial Engineering and Operations Research, Columbia University, \EMAIL{\texttt{ll3530@columbia.edu}}}
\AUTHOR{Christian Kroer}
\AFF{Department of Industrial Engineering and Operations Research, Columbia University, \EMAIL{\texttt{ck2945@columbia.edu}}}
} 

\input{abs.tex}




%


\maketitle
\input{intro.tex}

\input{theory.tex}

\ACKNOWLEDGMENT{This research was supported by the Office of Naval Research awards N00014-22-1-2530 and N00014-23-1-2374, and the National Science Foundation awards IIS-2147361 and IIS-2238960. 
This work has been presented at
Conference on Neural Information Processing Systems (2022), 
International Conference on Machine Learning (2023),
INFORMS annual meeting (2023)
Stanford Rising Star Workshop (2024), 
Workshop on Marketplace Innovation (2024), 
and the Ads Experimentation team at Meta. 
We thank attendees of all these talks for many helpful discussions.}

\bibliographystyle{informs2014} 
\bibliography{refs.bib} 
\ECSwitch



\begin{APPENDICES}
    \renewcommand{\theHsection}{A\arabic{section}}

\input{fisher_market_app.tex}
\input{fppe_app.tex}

\input{experiment_app.tex}
\end{APPENDICES}
%
%
%



\end{document}

%% file: notation.tex
\usepackage{mathtools}       
\usepackage{amsmath}
\usepackage{accents}
\usepackage{amsmath,amssymb}
\usepackage{multirow}
\usepackage{nameref}
\usepackage{hyperref}
\usepackage{cleveref}
\usepackage[shortlabels]{enumitem}
\setlist[itemize]{topsep=0pt,}
\usepackage{mathrsfs}
\usepackage{bbm}
\usepackage{makecell}

\usepackage{dsfont} 
\usepackage{xcolor} 

\usepackage{tabularx}
\usepackage{booktabs}
\usepackage[labelfont=bf,format=plain,justification=raggedright,singlelinecheck=false]{caption}

\newcommand{\ubar}[1]{\underaccent{\bar}{#1}}

\newcommand{\pmr}{pacing multiplier\,}
\newcommand{\pmrs}{pacing multipliers\,}

\def\##1\#{\begin{align}#1\end{align}}
\def\$#1\${\begin{align*}#1\end{align*}}

\def\given{\,|\,}

\def\tr{\mathop{\text{tr}}\kern.2ex}

\def \gam {{ \gamma }}
\def \indi {{ \mathds{1} }}

\def\R{{\mathbb R}}
\def\Rp{{\mathbb R _+}}
\def\Rn{{\mathbb R^n}}
\def\Rnp{{\mathbb R^n_+}}

\def\Rnmpp{{\mathbb R^{n-1}_{++}}}

\def\Rnpp{{\mathbb R^n_{++}}}
\def\Q{{\mathbb Q}}

\def\P{{\mathbb P}}
\def\E{{\mathbb E}}
\def\B{{\mathbb B}}

\def\cC{{\mathcal{C}}}
\def\cB{{\mathcal{B}}}

\def\cF{{\mathcal{F}}}
\def\cH{{\mathcal{H}}}

\def\cG{{\mathcal{G}}}

\def\cK{{\mathcal{K}}}

\def\cX{{\mathcal{X}}}

\def\cD{{\mathcal{D}}}
\def\cP{{\mathcal{P}}}

\def\cN{{\mathcal{N}}}

\def\rg{{\rangle}}
\def\lg{{\langle}}

\newlist{enumconditions}{enumerate}{1} 
\setlist[enumconditions]{label = \thelemma.\alph*}
\crefname{enumconditionsi}{Cond.}{Conditions}

\newlist{enumlmresult}{enumerate}{1} 
\setlist[enumlmresult]{label = \thelemma.\arabic*}
\crefname{enumlmresulti}{Part}{Parts}

\newlist{enumdef}{enumerate}{1} 
\setlist[enumdef]{label = \thedefn.\arabic*}
\crefname{enumdefi}{Cond.}{Cond.}

\crefname{assumption}{Assumption}{Assumptions}
\renewcommand*{\theAssumption}{\arabic{assumption}}

\newlist{enumthmresult}{enumerate}{1} 
\setlist[enumthmresult]{label = \thetheorem.\arabic*}
\crefname{enumthmresulti}{Part}{Parts}

\newlist{enumassumption}{enumerate}{1} 
\setlist[enumassumption]{label = \theAssumption.\arabic*}
\crefname{enumassumptioni}{Condition}{Conditions}

\crefname{definition}{Def.}{Defs.}
\crefname{equation}{Eq.}{Eqs.}
\crefname{claim}{Claim}{Claims}
\crefname{lemma}{Lemma}{Lemmas}
\crefname{appendix}{App.}{Apps}
\crefname{theorem}{Theorem}{Theorems}
\crefname{section}{Section}{Sections}
\crefname{corollary}{Corollary}{Corollaries}
\crefname{example}{Example}{Examples}
\crefname{figure}{Figure}{Figures}
\crefname{table}{Table}{Tables}

\def \cd {{ \cdot }}
\def \deltasti {{\delta^*_i}}
\def \deltast {{\delta^*}}
\def \deltagami {{\delta^\gamma_i}}
\def \deltagam {\delta^\gamma}

\def \I {{\mathrm{{I}}}}
\def \II {{\mathrm{{II}}}}
\def \III {{\mathrm{{III}}}}
\def \Ic {{I_<}}
\def \Icc {{I_<}}

\def \roott {{\sqrt t}}

\def \LAR  {{ \mathsf{{LAR}}}}
\def \REV {{ \mathsf{{REV}}}}
\def \REVst {{ \mathsf{{REV}}  }^* }
\def \REVgam {{ \mathsf{{REV}} }^ \gamma }

\def \NSW {{  \mathsf{{NSW}}}}

\def \FPPE {{  \textsf{{FPPE}}}}
\def \oFPPE {{  \widehat{ \textsf{{FPPE}}}}}
\def \LFM {{  \mathsf{{LFM}}}}
\def \oLFM {{  \widehat{ \mathsf{{LFM}} } }}
\def \LNSW {{ \mathsf{{NSW}}}}
\def \NSWgam {{  \mathsf{{NSW}}}^\gamma}
\def \NSWst {{  \mathsf{{NSW}}}^*}
\def \lin {\text{lin}}
\def \pay {{ \small  \textsf{{pay}}}}

\def \CTR {{ \small \textsc{{CTR}}}}
\def \click {{ \small \textsc{{click}}}}
\def \AE {{ \small \textsc{{Adexchange}}}}
\def \RG{{ \small \textsc{{Region}}}}
\def \TG{{ \small \textsc{{Usertag}}}}

\def \CR {{  \mathsf{{CR}}}}
\def \CI {{  \mathsf{{CI}}}}

\def \cBst {{ \mathcal{B}^*}}
\def \cBgam {{ \mathcal{B}^\gamma}}
\def \cBgamS {{ \mathcal{B}^\gamma_S}}
\def \cBgamC {{ \mathcal{B}^\gamma_C}}
\def \cBgamX {{ \mathcal{B}^\gamma_X}}

\def \cBstS {{ \mathcal{B}^*_S}}
\def \cBstC {{ \mathcal{B}^*_C}}
\def \cBstX {{ \mathcal{B}^*_X}}

\def \b {\beta}

\def \Hst {{H^*}}

\def \ust { u^* }
\def \ugam { u^\gamma }
\def \ugami { u^\gamma_i }

\def \usti {{ u^*_i }}

\def \mubar {\bar \mu}
\def \mubargam {{\bar \mu^\gamma}}

\def \mutau {\mu^\tau}

\def \must {{\mu^*}}
\def \mubarst {{\bar \mu^*}}
\def \musti {\mu^*_i}
\def \mubarsti {{\bar \mu^*_i}}

\def \ubarbetai {{\ubar\beta_i}}

\def \vi {v_i}
\def \vh {v_h}
\def \vk {v_k}
\def \vithetau {{ v_i(\theta^\tau) }}
\def \vithe {{ v_i(\theta) }}
\def \vitau {{ v_i^\tau }}
\def \vbar {{ \bar v }}
\def \vbarsq {{ \bar v ^2 }}

\def \vbarit {{ \bar v_i^t }}

\def \fbar {{\bar f}}

\def \var {{ \operatorname {Var} }}
\def \cov {{ \operatorname {Cov} }}

\def \nui {{ \nu_i}}
\def \nubar {{ \bar\nu}}

\def \sumtau {{\sum_{\tau=1}^{t}}}

\def \sumiton {{ \sum_{i=1}^n }}
\def \sumi {{ \sum_{i} }}

\def \ptau {{ p ^ \tau }}
\def \pgamtau {{ p ^ {\tau} }}
\def \pgam {{ p ^ {\gamma} }}

\def \xtaui {{ x^\tau_i }}
\def \xgam { x^{\gamma} }

\def \xgamtaui {{ x^{\gamma,\tau}_i }}

\def \ttinf {{t\to \infty}}
\def \thetau {{ \theta^\tau }}

\def \t {\theta}

\def \eps{{ \epsilon }}

\def \sighatsq {{ \hat \sigma ^2}}

\def \pst {p^*}
\def \betabar {{\bar \beta}}

\def \betadia {{\beta^\diamond}}

\def \betast {\beta^*}

\def \betasti {{\beta^*_i}}
\def \betastX {{\beta^*_X}}

\def \betai {\beta_i}
\def \betak {\beta_k}
\def \betah {\beta_h}
\def \betagam {\beta^\gamma}

\def \betagami {\beta^\gamma_i}

\def \xitau {{x_i^\tau}}
\def \xst {x^*}
\def \xsti {{x^*_i}}

\def \inv {^{-1}}
\def \pinv {^{\dagger}}
\def \sq {^{2}}
\def \tp {^{\top}}
\def \st {^{*}}
\def \ot {^{\otimes 2}}
\def \dom {{ \operatorname*{Dom}\,}}

\def \Diag {{ \operatorname*{diag}}}
\def \Vec {{ \operatorname*{Vec}}}
\newcommand{\defeq}{=}

\newcommand*\diff{\mathop{}\!\mathrm{d}}

\def \dt {{\diff \theta}}

\def \Ndiff {{N_{\text{diff}}}}
\def \Thetatie {{\Theta_{\text{tie}}}}
\def \toas {{ \,\overset{\mathrm{{a.s.\;}}}{\longrightarrow} \,}}
\def \toprob {{ \,\overset{{p}}{\to}\, }}
\def \toepi {{ \,\overset{\mathrm{epi}}{\longrightarrow}\, }}
\def \tod {{ \,\overset{{\,d\,\,}}{\rightarrow}\, }}
\def \invtroot {t^{-1/2}}

\DeclareMathOperator*{\esssup}{ess\,sup}
\def \relint {{ \operatorname*{relint}}}

\def \bidgap {{  \mathsf{{bidgap}}}}

%% file: abs.tex
\ABSTRACT{%
\new{
We initiate the study of statistical inference and A/B testing for two market equilibrium models: linear Fisher market (LFM) equilibrium and first-price pacing equilibrium (FPPE).
LFM arises from fair resource allocation systems such as allocation of food to food banks. }
For LFM, we assume that the observed data is captured by the classical finite-dimensional Fisher market equilibrium, and its steady-state behavior is modeled by a continuous limit Fisher market. 
\new{The FPPE model arises from internet advertising where advertisers are constrained by budgets and advertising opportunities are sold via first-price auctions.  }
We propose a statistical framework for the FPPE model, in which a continuous limit FPPE models the steady-state behavior of the auction platform, and a finite FPPE provides the data to estimate primitives of the limit FPPE. Both LFM and FPPE have an Eisenberg-Gale convex program characterization, the pillar upon which we derive our statistical theory.
We start by deriving basic convergence results for the finite market to the limit market.
We then derive asymptotic distributions, and construct confidence intervals. Furthermore, we establish the asymptotic local minimax optimality of estimation based on finite markets. We then show that the theory can be used for conducting statistically valid A/B testing on auction platforms. Synthetic and semi-synthetic experiments verify the validity and practicality of our theory.
}%

%% file: intro.tex
\section{Introduction}

Statistical inference is a crucial tool for measuring and improving a variety of real-world systems with multiple agents, including large-scale systems such as internet advertising platforms and resource allocation systems. 
However, statistical interference is a crucial issue in such systems. Past work has often focused on interference such as networks effects, which may arise due to user interactions on social media platforms.
In this paper, we focus on a different type of interference: interference effects arising from competition between agents on a platform. To be concrete, consider the case of A/B testing for internet advertising: budgets are prevalent among advertisers on such platforms, and these budgets mean that the actions of one advertiser can affect the actions of another. Often, in such systems, randomization is performed e.g. at the user level and then budget-splitting is used to clone advertisers into the A and B treatment. However, budget interactions may cause \emph{all} users in e.g. the A or B treatment to be related to each other, and thus it is not at all clear that one can apply standard statistical methods that treat each user as an independent sample.
Instead, a theory of \emph{equilibrium interference} is needed, and we need to understand how statistical interference can be performed when such interference is present.
We study statistical inference and A/B testing in two closely-related equilibrium models: 
First, we study one of the most classical competitive equilibrium models: the \emph{linear Fisher market} (LFM) equilibrium. 
Second, we study the \emph{first-price pacing equilibrium} (FPPE)~\citep{conitzer2022pacing}, which is a model that captures the budget-management tools often employed on internet advertising platforms.

In a Fisher market, there is a set of $n$ budget-constrained buyers and $m$ goods.  A market equilibrium (ME) is an allocation of the goods and a corresponding set of prices on the goods such that the market \emph{clears}, meaning that demand equals supply.
In a linear Fisher market, a buyers' utility is linear in their allocation.
Beyond being a classical model of price formation, the Fisher market equilibrium arises in resource allocation systems via the the competitive equilibrium
from equal incomes (CEEI) mechanism~\citep{varian1974equity,budish2011combinatorial}. 
In CEEI, each individual is given an endowment
of faux currency and reports her valuation for the goods; then, a market equilibrium is computed, and
the goods are allocated accordingly. The resulting allocation has many desirable properties such as Pareto optimality, envy-freeness and proportionality. 
Below we list examples of allocation systems where a Fisher market equilibrium naturally arises.

\begin{example}[Allocation of resources] 
Scarce resource allocation is 
prevalent in real life.
In systems that assign blood donation to hospitals and blood banks~\citep{mcelfresh2020matching}, 
or donated foods to charities in different neighborhoods~\citep{aleksandrov2015online,sinclair2021fairness},
scarce compute resources to users~\citep{ghodsi2011dominant,parkes2015beyond,kash2014no,devanur2018new}, 
course seats to students~\citep{othman2010finding,budish2016course},
the CEEI mechanism is already in use or serves as a fair and efficient alternative.
For systems that implement CEEI,
we may be interested in quantifying the variability of the amount of resources (blood or food donation) received by the participants (hospitals or charities) of these systems as well as the variability of fairness and efficiency metrics of interest in the long run.
Enabling statistical inference in such systems enables better tools for both evaluating and improving these systems.
\end{example}

\begin{example}[Fair notification allocation] In certain social media mobile apps, 
users are notified of events such as other users liking or commenting on their posts. 
Notifications are important for increasing user engagement, but too many notifications can be disruptive for users. 
Moreover, in practice, different types of notification are managed by distinct teams, competing for the chances to push their notifications to users. 
\citet{kroer2023fair}
propose to use Fisher markets to fairly control allocation of notifications.
They treat notification types as buyers, and users as items in a Fisher market.
Platforms are often interested in measuring outcome properties of such notification systems.
In \cref{sec:exp_instagram}
we will present a simulation study of our uncertainty quantification methods applied to the notification allocation problem.
\end{example}


    

The second type of equilibrium model we study is the FPPE model, which arises in internet advertising.
First, we review how impressions are sold in internet advertising, where first or second-price auction generalizations are used.
When a user shows up on a platform, an auction is run in order to determine which ads to show, before the page is returned to the user. Such an auction must run extremely fast. This is typically achieved by having each advertiser specify the following ahead of time: their target audience, their willingness-to-pay for an impression (or values per click, which are then multiplied by platform-supplied \emph{click-through-rate} estimates), and a budget.
Then, the bidding for individual impressions is managed by a proxy bidder controlled by the platform.  
As a concrete example, to create an ad campaign on Meta Ads Manager, advertisers need to specify the following parameters: (1) the conversion location (say website, apps, Messenger and so on), (2) optimization and delivery (target your ads to users with specific behavior patterns, such as those who are more likely to view the ad or click the ad link), (3) audience (age, gender, demographics, interests and behaviors), and (4) how much money do you want to spend (budget). 
Given the above parameters reported by the advertiser, the (algorithmic) proxy bidder supplied by the platform is then responsible for bidding in individual auctions to maximize advertiser utility, while respecting the budget constraint.

\new{An important role of these proxy bidders is to ensure smooth budget expenditure.
Pacing is a budget management method that modifies the advertiser's bids by applying a shading factor, known as a (multiplicative) \emph{pacing
multiplier}, to the advertiser's bid.}
Tuning the pacing multiplier changes the spending rate: the larger the pacing multiplier, the more aggressive the bids. 
The goal of the proxy bidder is to choose this pacing multiplier such that the advertiser exactly exhausts their budget (or alternatively use a multiplier of one in the case where their budget is not exhausted by using unmodified bids).
In this paper we focus on pacing-based budget management systems. 

First-price pacing equilibrium~\citep{conitzer2022pacing} is a market-equilibrium-like model that captures the steady-state outcome of a system where all buyers employ a proxy bidder that uses multiplicative pacing.
\citet{conitzer2022pacing} showed that an FPPE always exists and is unique. 
Moreover, as a pacing configuration method, FPPE enjoys nice properties such as being revenue-maximizing among all budget-feasible pacing strategies, 
shill-proof (the platform does not benefit from adding fake bids under first-price auction mechanism), and
revenue-monotone (revenue weakly increases when adding bidders, items or budget).
The FPPE model specifically captures the setting where each auction is a \emph{first-price} auction. First and second-price auctions are both prevalent in practice, but equilibrium models for second-price auctions are much less tractable (in fact, even finding one is computationally hard~\citep{chen2023complexity}). To that end, we focus on the first-price auction setting in this paper; the second-price setting is interesting, but we expect that it will be much harder to give satisfying statistical inference results for it.

Quantifying uncertainty in pacing systems is an important task on online advertising platforms. Basic statistical tasks, such as the prediction of the bidding behavior of advertisers or the revenue of the whole platform, require a statistical theory to model the intricacies of the bidding process. A/B testing, a method that seeks to understand the effect of rolling out a new feature, also requires a rigorous theoretical treatment to handle equilibrium effects.
\new{
    What a platform would typically like to do, is to treat each item in A and B as a separate unit and measure e.g., the price of the item as an independent observation. 
    However, interference occurs due to the optimizing behavior of each buyer, where they end up buying a bundle in their demand set (on a platform that runs ad auctions, this would typically occur through the proxy bidder that performs pacing). This optimization combined with the budget constraints causes interference both between different buyers and different items. Consequently, one has to treat the entire market as the observed ``unit.''
}
To the best of our knowledge, our results are the first to provide a statistical theory that captures such competitive interference effects that one would expect on an internet advertising platform.

Although LFM and FPPE have seemingly very different use cases, they each have an \emph{Eisenberg-Gale} convex program characterization~\citep{eisenberg1959consensus,eisenberg1961aggregation}.
This is the unifying theme that allows us to study these two models using similar tools.
In particular, this allows us to reduce inference about market equilibrium to inference about stochastic programs, where many classical tools from mathematical programming~\citep{shapiro2021lectures} and empirical processes theory~\citep{vaart1996weak} can be applied. 

\subsection{Contributions}

\paragraph{Statistical models for resource allocation systems and first-price pacing auction platforms. }
We leverage the infinite-dimensional Fisher market model of \citet{gao2022infinite} in order to propose a statistical model for resource allocation systems and FPPE platforms.
In this model, we observe market equilibria formed with a finite number of items that are i.i.d.\ draws from some distribution, and aim to make inferences about several primitives of the limit market, such as revenue, Nash social welfare (a fair metric of efficiency), and other quantities of interest; see \cref{tbl:qoi}. 
Importantly, this lays the theoretical foundation for A/B testing in resource allocation systems and auction markets,
which is a difficult statistical problem because buyers interfere with each other through the supply and the budget constraints. 
With the presence of equilibrium effects, traditional statistical approaches which rely on the i.i.d.\  assumption or SUTVA (stable unit treatment value assumption, \citet{imbens2015causal}) fail. 
The key lever we use to approach this problem is a convex program characterization of the equilibria, called the Eisenberg-Gale (EG) program. With the EG program, the inference problem reduces to an $M$-estimation problem~\citep{shapiro2021lectures,van2000asymptotic} on a constrained non-smooth convex optimization problem.

\begin{table}[htbp]
    \centering
    \begin{tabular}{@{}ccc@{}}
        \toprule
         \textbf{FPPE} & \textbf{LFM}   \\
        \midrule
        \makecell{pacing multipliers $\betast$,\\ revenue $\REVst$} & \makecell{ inverse bang-per-bucks $\betast$, \\ utilities $\ust$, \\ Nash social welfare $\NSWst$  } \\
        \bottomrule
    \end{tabular}
    \caption{Quantities of interest in LFM and FPPE.}
    \label{tbl:qoi}
\end{table}

\paragraph{Convergence and inference results for LFM and FPPE.}
We show that the finite market, which represents the observed data, is a good estimator for the limit market
by showing a hierarchy of results: strong consistency, convergence rates, and asymptotic normality.
We also establish that the observed market is an optimal estimator of the limit market in the asymptotic local minimax sense \citep{van2000asymptotic,le2000asymptotics,duchi2021asymptotic}. Finally, we provide consistent variance estimators, whose consistency is proved by a uniform law-of-large-numbers over certain function classes.
A shared challenge for developing statistical theory for both LFM and FPPE
is nonsmoothness: The objective function in the EG convex program (for a sampled finite market) is non-differentiable \new{on $\Rn$} almost surely, since it involves the max operator.
{
    \new{
    Even so, we will see that in our limit market model, the expectation operator, which becomes an integral, smooths out the non-differentiability issues under relatively mild conditions.
    We explore sufficient conditions for differentiability in \cref{sec:analytical_properties_of_dual_obj}.
    }
} 
For FPPE, there is a more prominent issue: the parameter-on-boundary issue, which means that the optimal population solution might be on the boundary of the constraint set.
Here we briefly discuss how we handle the two issues when deriving asymptotic distribution results for FPPE, which is one of the more difficult results in the paper.
First, asymptotic distribution results for $M$-estimation are known to hold under certain regularity conditions on stochastic programs~\citep[Theorem 3.3]{shapiro1989asymptotic}.
One such condition concerns the differentiability of the population objective.
We provide low-level sufficient conditions for differentiability, and show they have natural interpretations from an economic perspective~(\cref{sec:analytical_properties_of_dual_obj}).
Another important condition to verify is stochastic equicontinuity~(\cref{it:stoc_equic}), which we establish by leveraging empirical process theory~\citep{vaart1996weak,kosorok2008introduction}.

\paragraph{Statistically reliable A/B testing in resource allocation systems and FPPE platforms.}
Applying our theory, we \new{formalize and analyze} a budget-split A/B testing design for item-side randomization that resembles real-world A/B testing methodology in markets with budgets.
In the budget-split design, treatment and control markets are formed independently, and buyer's budgets are split proportionally between them, while items are randomly assigned to treatment or control markets.
Then, based on the equilibrium outcomes in the A and B markets, we construct estimators and confidence intervals that enable statistical inference. 

\new{
\subsection{Overview and Motivation of our Limit Market Model}
As described above, the fundamental quantity of interest in our model is the limit market, whether for LFM or FPPE. We now provide some discussion and motivation for our use of this limit market.
\begin{figure}
    \centering
    \includegraphics[scale = .25]{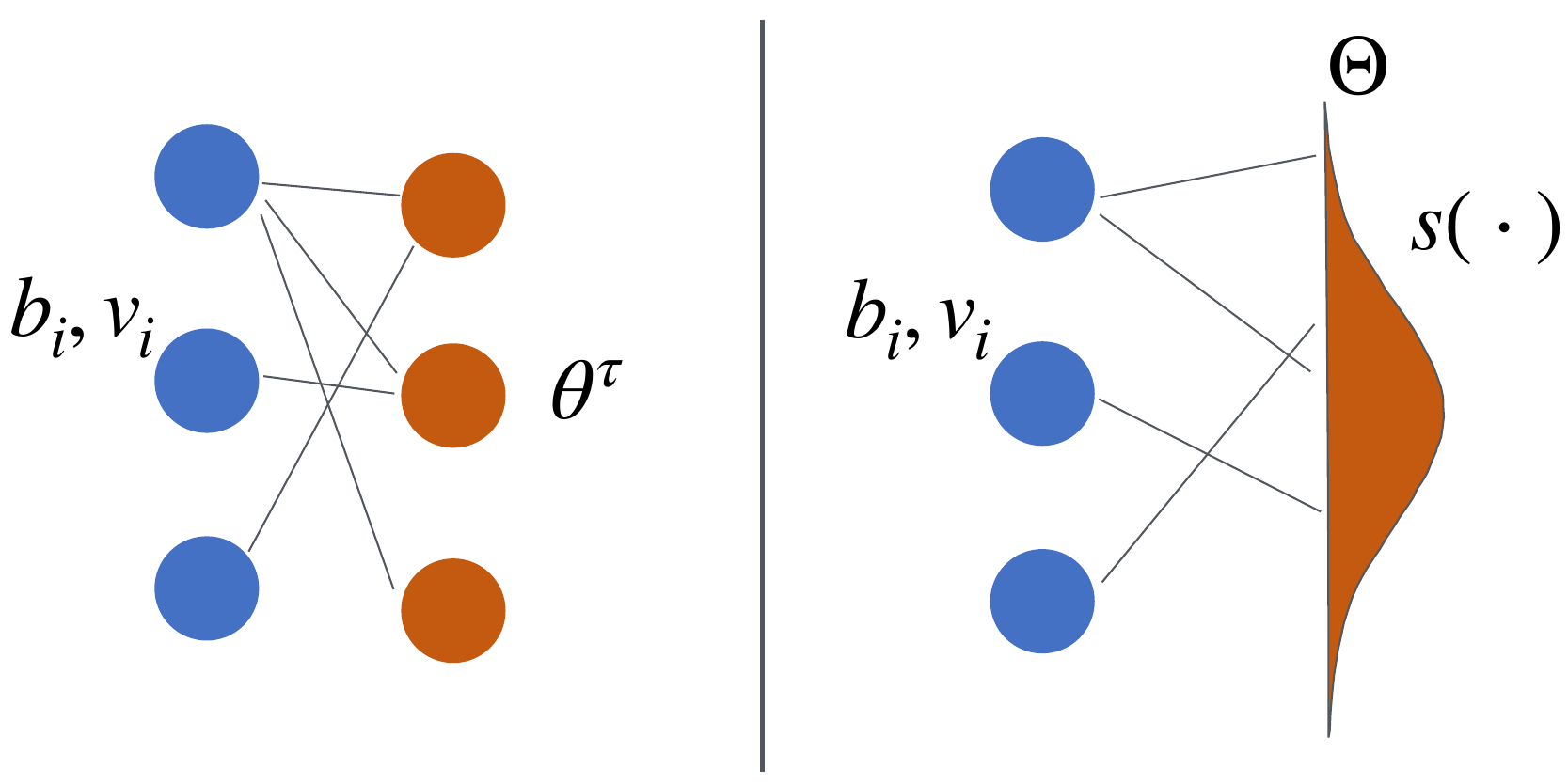}
    \caption{
        \new{Items (orange) are assigned to buyers (blue). Left: a finite market. Right: a limit market.}}
    \label{fig:fppe}
\end{figure}
Suppose a practitioner faces the following prediction task: given data arising from a finite market, what can we say about future finite markets that are created under similar conditions? 
In the case of the allocation systems we study in the paper,
this can be interpreted as follows: suppose we observe the behavior of the system for one week (represented by the finite market), and the allocation system is operating in equilibrium. Can we now make statistical inferences about next week's market, assuming that the generated items follow a similar pattern? Similarly, A/B testing a system under a budget-split design requires us to make statistical statements about each arm.


To address questions such as these, we study the statistical behavior of finite markets sampled from a limit market where the item set is continuous, and the supply becomes a distribution on the item set; see \cref{fig:fppe} for an illustration. 
Limit markets are formally introduced in \cref{sec:lfm and fppe definitions}. 

The limit market serves as a conceptual model through which we derive \textit{prediction intervals for future finite markets}.
We focus on deriving confidence intervals for the limit market, i.e., an interval $I$ constructed from the data so that the probability $\P(\text{quantities in the limit market lie in } I)$ is high.
One can then easily construct prediction intervals for a future market.
Take FPPE as an example. Let $\betagam$ be the observed multiplier, and $\betast$ be the limit one. Suppose $\beta^{new}$ is the multiplier in a future market, independent of the observed market. We will show that the observed multiplier concentrates around the limit multiplier: $\sqrt t (\betagam - \betast) \tod N(0, \Sigma_\beta)$ for some matrix $\Sigma_\b$. Then 
\begin{align}
    \sqrt t (\betagam - \beta^{new})  = \sqrt t ((\betagam - \betast) - (\beta^{new} - \betast)) \tod N(0, 2\Sigma_\beta)
\end{align}
by independence between $\betagam$ and $\beta^{new}$.
Given an estimator of $\Sigma_\beta$, a prediction interval for $\beta^{new}$ can thus be constructed.

Of course, if the goal is to make predictions about future finite markets, one might wonder if there is a way to give a formal model of predicting future finite markets directly from observed finite markets, without going through the limit market.
For some quantities it is easy to impose a simple statistical model and develop a prediction theory while ignoring the equilibrium structure in the data; for other quantities, it is more difficult to even put a statistical model on them.
In either case, ignoring the equilibrium structure casts doubt on the validity of such approaches.

}
\subsection{Related Works} \label{sec:related_work}

\textbf{A/B testing in two-sided markets.}
Empirical studies by 
\citet{blake2014marketplace,fradkin2019simulation} demonstrate bias in experiments due to marketplace interference.
\citet{basse2016randomization} study the bias and variance of treatment effects under two randomization schemes for auction experiments.
\citet{bojinov2019time} study the estimation of causal quantities in time series experiments.
Some recent state-of-the-art designs are the multiple randomization designs \citep{liu2021trustworthy,johari2022experimental,bajari2021multiple} and the switch-back designs 
\citep{sneider2018experiment,hu2022switchback,li2022interference,bojinov2022design,NEURIPS2020_abd98725}.
The surveys by 
\citet{kohavi2017surprising,bojinov2022online} contain detailed accounts of A/B testing in internet markets.
See \citet{larsen2022statistical} for an extensive survey on statistical challenges in A/B testing.
Compared to these papers, our paper is the first to focus on A/B testing with equilibrium effects.

\new{
The budget-split design we study is similar to one studied by \citet{liu2021trustworthy}: buyer's budgets are split proportionally, creating two sets of identical buyers, and then items are assigned to either group of buyers at random. However, we differ in the theory to analyze such experiments.
The main difference between our paper and theirs is that they did not consider equilibrium effects, while we do.
In particular, \citet{liu2021trustworthy} assume that pacing and other strategic effects are fixed across treatments. In turn, this means that any strategic behavior or budget-optimizing behavior ignores the A/B assignment.
On a related note, we consider randomness of both the item (impression) arrivals and the treatment assignment, while Liu et al.\ only consider the latter.  
Second, we provide an exact variance formula for many quantities in the ad auction systems, such as pacing multipliers and revenues, while Liu et al.\ do not. 


}
\textbf{Pacing equilibrium.}
Pacing is a prevalent budget-management methods on ad auction platforms.
In the first-price setting,
\citet{borgs2007dynamics} study first price auctions with budget constraints in a perturbed model, whose limit prices converge to those of an FPPE.
Building on the work of \citet{borgs2007dynamics}, \citet{conitzer2022pacing} introduce the FPPE model and discover several properties of FPPE such as shill-proofness and monotonicity in buyers, budgets and goods.
There it is also established that FPPE is closely related to the quasilinear Fisher market equilibrium~\citep{chen2007note,cole2017convex}.
\citet{gao2022infinite} propose an infinite-dimensional variant of the quasilinear Fisher market, which lays the probability foundation of the current paper.
\citet{gao2021online,liao2022dualaveraging} study online computation of the infinite-dimensional Fisher market equilibrium.
In the second-price setting,
\citet{balseiro2015repeated} investigate budget-management in second-price auctions through a fluid mean-field approximation;
\citet{balseiro2019learning} study adaptive pacing strategy from buyers' perspective in a stochastic continuous setting;
\citet{balseiro2021budget} study several budget smoothing methods including multiplicative pacing in a stochastic context;
\citet{conitzer2022multiplicative} study second price pacing equilibrium, and shows that the equilibria exist under fractional allocations.

\textbf{$M$-estimation when the parameter is on the boundary}
There is a long literature on the statistical properties of 
$M$-estimators when the parameter is on the boundary~\citep{geyer1994asymptotics,shapiro1990differential,shapiro1988sensitivity,shapiro1989asymptotic,shapiro1991asymptotic,shapiro1993asymptotic,shapiro2000asymptotics,andrews1999estimation,andrews2001testing,knight1999epi,knight2001limiting,knight2006asymptotic,knight2010asymptotic,dupacova1988asymptotic,dupavcova1991non,self1987asymptotic}.
Some recent works on the statistical inference theory for constrained $M$-estimation include \citet{li2022proximal,hong2020numerical,hsieh2022inference}.
Our work leverages \citet{shapiro1989asymptotic}, which develops a general set of conditions
for asymptotic normality of constrained $M$-estimators when the objective function is nonsmooth.
Working under the specific model of FPPE, we build on and go beyond these 
contributions by deriving sufficient condition for asymptotic normality in FPPE, establishing local asymptotic minimax theory and developing valid inferential procedures.

\textbf{Statistical learning and inference with equilibrium effects}
Online learning approaches, which are related to statistical learning, have been investigated for other equilibrium models, such as
general exchange economy~\citep{guo2021online,liu2022welfare} and
matching markets~\citep{cen2022regret,dai2021learning,liu2021bandit,jagadeesan2021learning,min2022learn}.
Our work is also related to the rich literature of inference under interference~\citep{hudgens2008toward,aronow2017estimating,athey2018exact,leung2020treatment,hu2022average,li2022random}.
In the economic literature, researchers have studied how to estimate auction market primitives from bid data; see \cite{athey2007nonparametric} for a survey.

Closely related to our work is a recent preprint by \citet{munro2021treatment}. They consider a potential outcomes framework where the outcome of an agent depends on the treatments of all agents, but only through the equilibrium price. The equilibrium price is attained by a market clearance condition. Although both their work and our work consider a limit market equilibrium (in their case a mean-field equilibrium), there are significant differences. First, \citet{munro2021treatment} send the number of agents to infinity while we consider the asymptotics where the number of items grows. Second, \citet{munro2021treatment} present a very general market equilibrium framework that requires abstract regularity conditions (which do not hold in our setting), while we focus on equilibria arising from resource allocation systems and auction pacing systems, and consequently we are able to present low-level conditions that facilitate statistical inference. Third, their model works with a \emph{single} market where buyers are randomly exposed to treatment or control. 
\new{
    Consequently, \citet{munro2021treatment} focuses on estimating direct effects and spillover effects when buyers in a single market are randomly assigned to either treatment or control. 
    In contrast, our work is focused on interventions that affect the entire valuation distribution, and we study separated markets when we apply our theory to A/B testing.
}
For this reason, an A/B testing framework such as the one we develop is necessary. \citet{Wager2021,sahoo2022policy} also consider a mean-field game modeling approach and perform policy learning with a gradient descent method.

\ifanon
\else
This paper builds upon two preliminary conference papers~\citep{liao2023fisher,liao2023stat}. The present paper gives a more unified presentation, and provides some additional results on the statistical theory of Fisher markets and FPPE.  
Perhaps most importantly, this paper conducts two semisynthetic experiments based on an ad auction dataset~\citep{liao2014ipinyou} and an Instagram notification dataset~\citep{instagramdata}, demonstrating the practicality of the proposed theory.
This paper also provides strong consistency and convergence rate for FPPE, minimax optimality results for LFM, and a novel closed-form expression for the Hessian matrix of the population EG objective~(\cref{eq:def_pop_eg}) using results from differential geometry~\citep{kim1990cube}.
Building on the preliminary versions of the present paper, \citet{liao2024bootstrap} extends the FPPE statistical theory to the cases where degenerate buyers are present and develops bootstrap inference methods, and \citet{liao2024interference} studies interference among FPPEs.
\fi

\subsection{Notations}
Let $e_i$ be the $i$-th basis vector in $\Rn$.
Furthermore, we let $A\pinv $ be the Moore-Penrose pseudo inverse of a matrix $A$. 
Let $\diff \t$ denote the Lebesgue measure in $\R^D$.
For a measurable space $(\Theta, \diff \theta)$, we let $L^p$ (and $L^p_+$, resp.) denote the set of (nonnegative, resp.) $L^p$ functions on $\Theta$ w.r.t\ the integrating measure $\diff \theta $ for any $p\in [1, \infty]$ (including $p=\infty$). 
We treat all functions that agree on all but a measure-zero set as the same.
For a sequence of random variables $\{X_n\}$, we say $X_n = O_p(1)$ if $\lim_{K\to \infty} \limsup_n \P(|X_n| > K) = 0$. We say $X_n = o_p(1)$ if $X_n$ converges to zero in probability.
We say $X_n = O_p(a_n)$ (resp.\ $o_p(a_n)$) if $X_n/a_n = O_p(1)$ (resp.\ $o_p(1)$).
The subscript $i$ is for indexing buyers and superscript $\tau$ is for items.

%% file: theory.tex
\section{Linear Fisher Market and First-Price Pacing Equilibrium}
\label{sec:lfm and fppe definitions}
\new{In this section we introduce the Fisher market equilibrium and the first-price pacing equilibrium.} We start by presenting components that are common to both models, and then introduce each equilibrium concept.
In both LFM and FPPE, we have a set of $n$ buyers and a set of items, and the goal is to find market-clearing prices for the items.
The items are represented by a set $\Theta \in \R^D$, a compact set with $\int \indi (\Theta) \diff \t > 0$.
Clearly the measure \new{space} $(\Theta,\dt)$ is atomless.

Both LFM and FPPE require the following elements; see \cref{fig:fppe} for an illustration.

(1)
The \emph{budget} $b_i$ of buyer $i$. Let $b = (b_1,\dots, b_n)$. 
\new{The smallest budget is denoted by $\ubar b = \min_i b_i$.}
(2)
The \emph{valuation} for buyer $i$ is a function $v_i \in L^1_+$. Buyer $i$ has valuation $v_i(\theta)$ (per unit supply) of item $\theta\in \Theta$. Let $v: \Theta \to \Rn$, $v(\theta) = [v_1(\theta),\dots, v_n(\theta)] \tp$. We assume $\vbar = \max _i \sup_\theta  v_i(\theta)< \infty $.
(3) 
The \emph{supplies} of items are given by a function  $ s \in L^\infty_+$, i.e., item $\theta\in \Theta$ has $s(\theta)$ units of supply. 

    Without loss of generality, we assume a unit total supply $\int_\Theta s \diff \theta = 1$, which makes $s$ a probability measure.
    Let $\P$ denote the probability measure induced by $s$, i.e.,
    $\P (A) = \int_A s\dt$  for a measurable set $A$. 
    Given $g:\Theta \to \R$, we let $\E[g] = \int g(\theta)s(\theta)\diff \theta$ and $\var[g] = \E[g\sq] - (\E[g])\sq$.
    \new{
        Also let $\nu_i = \int v_i s\diff \t = \E[ v_i]$ be the monopolistic utility of buyer $i$, and $\nubar = \max_i \nu_i$.
    } 
    Given  $t$ i.i.d.\ draws $\{ \theta^1,\dots, \theta^t\}$ 
    from $s$, let $P_t g(\cd) = \frac1t \sumtau g(\thetau)$.
    Let $\vitau = v_i(\thetau)$.

Equilibria in both LFM and FPPE are characterized by a particular type of convex program known as an \emph{Eisenberg-Gale} (EG) convex program. For statistical inference purposes, we will focus on the duals of these EG programs, which is a convex optimization problem over the space of \emph{pacing multipliers} $\beta \in \Rnp$ (these pacing multipliers turn out to represent the price-per-utility of buyers in equilibrium).
In both cases, the dual EG objective separates into per-item convex terms
\begin{align}
    \label{eq:def:F}
    F (\t, \b) = f(\t,\b)-  \sumiton b_i \log \beta_i \;,\;\; f(\t,\b)=\max_{i\in[n]} \beta_i v_i(\theta)  \;.
\end{align}
and the population and sample EG objectives are
\begin{align}
    \label{eq:def_pop_eg}
    H(\beta) = \E[F(\t,\b)] \; , \;\; H_t(\beta) = P_t F(\cd, \beta) \; .
\end{align}
The reason we focus on the duals is that they can be cast as sample average approximations of the limit convex programs. This interpretation is not possible for the original primal EG programs.

\subsection{Linear Fisher Markets (LFM)}
\label{sec:lfm_review}
The LFM model has two primary uses. Its original intent is as a model of competition, and price formation in a competitive market.
An additional, and practically important, use of LFM is as a tool for fair and efficient resource allocation (with the items being the resources).
If every individual in a resource allocation problem is given one unit of faux currency, then the resulting LFM equilibrium allocation is known to be both Pareto efficient and satisfy the fairness desiderata of envy-freeness and proportionality~\citep{nisan2007algorithmic}. This fair allocation approach is known as competitive equilibrium from equal incomes (CEEI)~\citep{varian1974equity}.

We now describe the competitive equilibrium concept. 
Imagine there is a central policymaker that sets 
prices $p(\cdot)$ for the items $\Theta$. Upon observing the prices, buyer~$i$ 
maximizes their utility subject to their budget. 
Their \emph{demand set} is the set of bundles that are optimal under the prices:
\[ D_i (p) := \argmax_{x_i \in L^\infty_+(\Theta)} \bigg\{ \int v_i x_i s \diff \t :  \int p x_i s\diff \t \leq b_i\bigg\} \; .\] 
Of course, due to the supply constraints, if prices are too low, there will be a supply shortage. On the other hand, if prices are too high, a surplus occurs. A competitive equilibrium is a set of prices and bundles such that all items are sold exactly at their supply (or have price zero).
We call such an equilibrium the \emph{limit LFM equilibrium} for the supply function $s$~\citep{gao2022infinite}.
\begin{definition}[Limit LFM]
    \label{def:LMF}
The limit equilibrium, denoted $\LFM (b,v,s, \Theta)$, is an allocation-price tuple $(x,p(\cd))$ such that the following holds.
\begin{enumerate}
    \item (Supply feasibility and market clearance) $\sum_i x_i \leq 1$ and $\int p (1 - \sum_i x_i) s \diff \t= 0$. 
    \item (Buyer optimality) $x_i \in D_i (p)$ all $i$.
\end{enumerate}
\end{definition}

\citet{gao2022infinite} show that an equilibrium of a limit LFM must exist, and that when the measure space $(\Theta,\dt)$ is atomless, a pure equilibrium allocation \footnote{An allocation $x$ is pure if $x_i(\t) \in \{0,1\}$.} must exist.
Given an equilibrium $(\xst, \pst)$, let 
\begin{align*}
    \usti = \int \vi s \xsti \diff \t \;, \quad \betasti = b_i / \usti  \;,\quad \NSWst = \sumiton b_i \log(\usti )
\end{align*}
be buyer $i$' utility, her inverse bang-per-buck, and the (log) Nash social welfare of the whole market.
The inverse bang-per-buck $\beta_i^*$ can also be seen as the price-per-utility of buyer $i$. 
In a general LFM, the equilibrium allocation may not be unique, but the equilibrium quantities $\pst,\betast,\ust$ are unique. 
In order to facilitate statistical inference, we will impose certain  
differentiability conditions, which turn out to imply uniqueness and purity of the equilibrium allocation $\xst$. 

Next we introduce the \emph{finite} LFM, which models the data we observe in a market.
The finite LFM equilibrium is nothing but a limit LFM equilibrium where the item set $\Theta$ is the finite set of observed items $\gamma$. 
Let $\gam = \{ \theta^1,\dots, \theta^t\}$ be $t$ i.i.d. samples from the supply distribution $s$, each with supply $\sigma \in \Rp$.

\begin{definition}[Finite LFM] \label{def:observed_market}

    The finite observed LFM, denoted $\oLFM(b,v,\sigma, \gamma)$, is any allocation-price tuple $({x}, p) \in \R^{t\times n}_+ \times \Rnp $ such that the following hold:
    \begin{enumerate}
        \item (Supply feasibility and market clearance) $\sumi\xtaui \leq 1 $ and  $\sum_\tau \ptau (1- \sumi \xitau)  = 0$.
    
        \item (Buyer optimality) $x_i \in D_i (p) = \argmax_{x_i} \{ \sum_\tau \xitau \vitau : \sigma \sum_\tau \xitau \ptau \leq b_i , \xitau \geq 0\}$, the demand set given the prices. 
    \end{enumerate}
\end{definition}

Let  $(\xgam, \pgam) \in \oLFM(b,v,\sigma = 1/t, \gamma)$ \footnote{We use $\in$ since the equilibrium allocation may not be unique; equilibrium prices are unique.}, 
where
$\xgam = (\xitau )_{i,\tau}$ and price $\pgam = [p^1,\dots, p^t]$. 
We study this form of finite LFM due a scale-invariant property of LFM, \new{and this ensures that for all market sizes, the ``buyer size'' is comparable to the ``item size'' of the market} (see \cref{sec:lfm_scale_invariance}).
Buyer $i$'s utility is $\ugam_i=\sigma \sumtau \vitau \xtaui = \frac1t \sumtau \vitau \xtaui  $, and the inverse bang-per-buck is $\betagami = b_i / \ugami$. The (log) Nash social welfare is $\NSWgam = \sumi b_i \log (\ugami)$.

There exists natural bounds on $\betast$ in a limit LFM. Recall $\nu_i =\int v_i s\dt $ is the expected value of buyer~$i$. 
By \citet{gao2022infinite}, 
we know that 
$b_i / \nu_i \leq \betasti \leq (\sum_{i'} b_{i'}) / (\min_{i'} \nu_{i'})$.
 We define 
\begin{align}
    \label{eq:def_C_LFM}
    C_\LFM = \prod_{i=1}^{n}
    \bigg[ \frac{b_i}{2\nu_i},  \frac{ 2 \sum_{i'} b_{i'}}{\min_{i'} \nu_{i'}}\bigg] \subset \Rnp
\end{align}
to be the region whose interior $\betast$ must lie in.

It is well-known~\citep{eisenberg1959consensus,cole2017convex,gao2022infinite}
that the equilibrium inverse bang-per-buck in a limit (resp.\ finite) LFM uniquely solves
the population (resp.\ sample) dual EG program
\begin{align}
    \label{eq:pop_deg_lfm}
    \betast = \argmin_{\beta \in \Rnp}
     H(\beta) \; , \;
     \betagam = \argmin_{\b\in \Rnp} H_t(\b) \; .
\end{align}

\new{
We review other properties of LFM, such as scale-invariance and mechanism design properties, including fairness and efficiency, in \cref{sec:lfm_scale_invariance}.
}

\subsection{First-Price Pacing Equilibrium (FPPE)}
\label{sec:fppe_review}
The FPPE setting~\citep{conitzer2022pacing} models an economy that typically occurs on internet advertising platforms: the buyers (advertisers in the internet advertising setting) are subject to budget constraints, and must participate in a set of first-price auctions, each of which sells a single item.
Each buyer is assigned a \emph{pacing multiplier} $\beta_i \in [0,1]$ by the platform to scale down their bids in the auctions, and submits bids of the form $\beta_i v_i(\theta)$ for each item $\theta$.  \new{From the platform's perspective, the goal of choosing $\beta_i$ is to ensure that there is \emph{no unnecessary pacing}: A buyer's budget constraint must be satisfied, but if $\beta_i < 1$ then the buyer exhausts their budget exactly.}
In the FPPE model, all auctions occur simultaneously, and thus the buyers choose a single $\beta_i$ that determines their bid in all auctions.
The utility of a buyer in FPPE is quasilinear: it is the sum of their value received from items plus their leftover budget (this is equivalent for decision-making purposes to the utility being the value received from items minus payments).

\begin{definition}[Limit FPPE, \citet{gao2022infinite}]
    \label{def:limit_fppe}
    A limit FPPE, denoted $\FPPE(b,v,s, \Theta)$, is the unique tuple $(\beta, p(\cd)) \in [0,1]^n \times L^1_+ (\Theta)$ such that there exist $x_i : \Theta \to [0,1]$, $i\in[n]$ satisfying
    \begin{enumerate}[series = tobecont,itemjoin = \quad]
        \item (First-price) 
        Prices are determined by first-price auctions:
        for all items $\theta \in \Theta$, $p(\theta) = \max_i \betai v_i(\theta)$. 
        Only the highest bidders win:
        for all $i$ and~$\theta$, $x_i(\theta) > 0$ implies $\betai \vithe =\max_k \beta_k v_k (\theta)$ 
        \label{it:def:first_price}
        \item (Feasibility, market clearing)  
        Let $\pay_i = \int x_i(\theta) p(\theta) s(\theta)\diff \theta $ be the expenditure of buyer $i$.
        Buyers satisfy budgets:
        for all~$i$, $\pay_i  \leq b_i$. 
        There is no overselling: 
        for all $\theta$, $\sumiton x_i (\theta) \leq 1$.  
        \label{it:def:supply_and_budget}
        All items are fully allocated:
        for all~$\theta$, $p(\theta) > 0$ implies $ \sumiton x_i(\theta) = 1$.
        \item (No unnecessary pacing) For all $i$, $\pay_i < b_i$ implies $\betai = 1$. 
        \label{it:def:rev_max}
    \end{enumerate}
\end{definition}

FPPE is a hindsight and static solution concept for internet ad auctions. 
Suppose the platform knows all the items that are going to show up on a platform.
Then FPPE describes how the platform could configure the $\betai$'s in a way that ensures that all buyers satisfy their budgets, while maintaining their expressed valuation ratios between items. 
In practice, the $\betai$'s are learned by an online algorithm that is run by the platform~\citep{balseiro2019learning,conitzer2022pacing}, \new{and FPPE captures the hindsight solution that these learning algorithms should converge to.}
FPPE has many nice properties, such as the fact that it is a competitive equilibrium, it is revenue-maximizing, revenue-monotone, shill-proof, has a unique set of prices, and so on~\citep{conitzer2022pacing}.
We refer readers to \citet{conitzer2022pacing,kroer2022market} for more context about the use of FPPE in internet ad auctions.

\citet{gao2022infinite} show that 
a limit FPPE always exists and is unique, and when the item space is atomless, a pure allocation exists.
Let $\betast$ and $\pst$ be the unique FPPE equilibrium multipliers and prices. 
Revenue in the limit FPPE is 
\begin{align}
    \REVst \defeq \int \pst(\theta) s(\theta)\diff \theta \; .
\end{align}
It is also easy to see that $\REVst = \sumi \pay_i$.
As with LFM, we will impose differentiability assumptions which imply uniqueness of $\xst$.
When $\xst$ is unique, we let $\deltasti \defeq b_i - \pay_i$ be the leftover budget.

In an FPPE, based on the pacing multiplier and the budget expenditure,
we can categorize buyers in terms of how they satisfy the no unnecessary pacing condition.
As we will see later, the statistical behavior of pacing multipliers varies by category.
\begin{itemize}
    \item Paced buyers ($\betasti < 1$). 
    We use $ \Ic = \{i : \betasti < 1\}$ to denote them.
    Due to the budget constraints, they are not able to bid their value in the auctions at equilibrium, and by the no unnecessary pacing condition in \cref{def:limit_fppe}, their budgets are fully exhausted, i.e. $\deltasti = 0$.
    
    \item Unpaced buyers ($\betasti = 1$).
    We use $I_= = \{i : \betasti = 1\}$ to denote them. 
    They can be further divided according to their budget expenditure.
    \begin{itemize}
        \item Buyers who have strictly positive leftover budgets ($\betasti=1, 0< \deltasti \leq b_i$). 
        This category also includes buyers who do not win any items ($\betasti = 1, \deltasti = b_i$).
         
        \item  Degenerate buyers ($\betasti = 1, \deltasti = 0 $); and edge-case in the FPPE model. 
        If these buyers were given an arbitrarily-small amount of additional budget then they would have positive leftover budget at equilibrium without changing the equilibrium.
        For the FPPE statistical theory developed in this paper, we assume absence of such buyers (\cref{as:scs}). 
        \ifanon
        \else
        In follow-up work to the conference version of the present paper, \citet{liao2024bootstrap} give some results for the case where degenerate buyers exist.
        \fi
    \end{itemize}
\end{itemize}

We let $\gam = \{ \theta^1,\dots, \theta^t\}$ be $t$ i.i.d.\ draws from $s$, each with supply $\sigma=1/t$. They represent the items observed in an auction market.
The definition of a finite FPPE is parallel to that of a limit FPPE, except that we change the supply function to be a discrete distribution supported on $\gamma$.

\begin{definition}[Finite FPPE, \citet{conitzer2022pacing}]    
    \label{def:finite_fppe}
    The finite observed FPPE, $\oFPPE(b,v, \sigma , \gamma)$, is the unique tuple
    $(\beta,p) \in [0,1]^n \times \R^t_+ $ 
    such that there exists $x_i^\tau \in [0,1]$ satisfying:
\begin{enumerate}
    \item 
    (First-price) For all $\tau$, $\ptau = \max_i \betai \vitau$. For all $i$ and $\tau$, $\xitau > 0$ implies $\betai \vitau =\max_k \beta_k v_k^\tau$. 
    \item
    (Supply and budget feasible)  For all $i$, $ \sigma \sum_\tau \xitau \ptau  \leq b_i$. For all $\tau$, $\sumi\xitau \leq 1$.  
    \item
    (Market clearing)  For all $\tau$, $\ptau > 0$ implies $ \sumi\xitau = 1 $.
    \item
    (No unnecessary pacing) For all $i$, $ \sigma \sum_\tau \xitau \ptau  < b_i$ implies $\betai = 1$.
\end{enumerate}
\end{definition}

Let $(\betagam, \pgam) = \oFPPE(b,v,\sigma = 1/t,\gamma)$. 
\new{
    As for LFM, decreasing supply per item ensures that for all market sizes, the ``buyer size'' is comparable to the ``item size'' of the market (see \cref{sec:fppe_scale_invariance}).
}
Given the equilibrium price $\pgam = [p^1,\dots, p^t]\tp$, the revenue in a finite FPPE is $\REVgam \defeq \sigma \sumtau \ptau = \frac1t \sumtau \ptau$.

It is well-known \citep{cole2017convex,conitzer2022pacing,gao2022infinite}
that $\beta$ in a limit (resp.\ finite) FPPE uniquely solves 
the population (resp.\ sample) dual EG program
\begin{align}
    \label{eq:pop_deg}
    \betast =  \argmin_{\beta \in (0, 1]^{n}}
     H(\beta) \;, \;
     \betagam = \argmin_{\b\in(0,1]^n} H_t(\b)  \;,
\end{align}
where the objectives $H$ and $H_t$ are the same as in \cref{eq:pop_deg_lfm}. The difference between the LFM and FPPE convex programs is that for FPPE we impose the constraint $\beta\in (0,1]^n$.

The EG program and certain quantities of the FPPE are related as follows.
\begin{lemma} \label{lm:fppe_relation}
    Suppose $H$ is twice continuously differentiable at the equilibrium pacing multiplier vector $\betast$.
    Then $\nabla H (\betast) = - \deltast$, and $\nabla\sq H(\betast) \betast = [b_1/\betast_1, \dots, b_n / \betast_n]\tp$.
\end{lemma}

\textit{Proof sketch}
    The first equality follows from the fact that 
    leftover budgets are the Lagrange multipliers corresponding to the constraint $\beta \leq 1 \in \Rn$.
    The second equality follows from the first-order homogeneity of $f(\t,\b)=\max_i \betai \vithe$ in \cref{eq:def:F}. Appendix~\ref{sec:techlemma_fppe} for details.

\subsection{Differentiability Assumption}
\new{
    As previously explained, our statistical theory will be founded on M-estimation theory.
    In $M$-estimation, twice differentiability is usually required in order to establish asymptotic normality, and we will similarly impose it on the EG objective (\cref{eq:def_pop_eg}) for our asymptotic normality results in LFM (\cref{thm:clt_beta_u}) and FPPE (\cref{thm:clt}), and the statement of minimax lower bounds (\cref{thm:nsw_aym_risk,thm:rev_localopt}) in later sections. 
    We will revisit differentiability and derive sufficient conditions on market primitives in \cref{sec:analytical_properties_of_dual_obj}, after we present the main statistical results for LFM and FPPE. 
}

\begin{assumption}[\textsf{\scriptsize{SMO}}]
    \label{as:smo}
Let $\betast$ denote the equilibrium inverse bang-per-buck in LFM, or the equilibrium pacing multiplier in FPPE.
Assume the map $\beta \mapsto \fbar(\b)= \E_s[\max_i \betai \vithe ]$ is twice continuously differentiable in a neighborhood of $\betast$. We let $\cH\defeq \nabla\sq H (\betast)$.
\end{assumption}

\input{fisher_market_results}
\input{fppe_results}

\section{The Differentiability of the EG Objective}
\label{sec:analytical_properties_of_dual_obj}
In this section, we provide lower-level conditions on the market's primitives and equilibrium such that \cref{as:smo} is implied.

We start with the differential structure of $f(\t,\b)=\max_i \betai\vithe$. 
The function $f(\t, \b)$ is a convex function of $\b$ and its subdifferential $\partial_\b f(\t, \b)$ is the convex hull of $\{ v_i e_i \in \Rnp: \text{index $i$ such that $\betai\vithe = \max_k \betak v_k(\t)$}  \}$, with $e_i$ being the base vector in $\Rn$.
When $\max_i \betai v_i(\t)$ is attained by a unique $i^*$, the function $f$ is differentiable. In that case, the $i$-th entry of $\nabla_\b f(\t,\b)$ is $v_i(\t)$ for $i=i^*$ and zero otherwise. 


\subsection{First-order Differentiability} \label{sec:first_diff}
The $\bidgap$ function in \cref{eq:def:bidgap} is useful for characterizing 
first-order differentiability
of $\fbar (\b) = \E_s[f(\t, \b)]$. 
When there is a tie for an item $\theta$, we have $\bidgap(\beta,\theta) = 0$. 
When there is no tie for an item $\theta$, the gap $\bidgap(\beta,\theta)$ is strictly positive.
The gap function characterizes smoothness of $f$:
$f(\cdot, \theta)$ is differentiable at $
\beta$
iff $\bidgap(\beta,\theta)$ is strictly positive.

\begin{theorem}[First-order differentiability]
    \label{thm:first_differentiability}
    The following are equivalent. 
    (i) The dual objective $H$ is differentiable at a point $\beta$. 
    (ii) The function $\t \mapsto \bidgap(\beta,\theta)$ is strictly positive $s$-almost surely:
    \begin{align} \label{eq:as:notie}
       \P( \{ \t: \bidgap(\beta,\theta) > 0 \} ) = 1
        \;.
     \end{align}
    (iii) 
    The set of items that incur ties under pacing profile $\b$ is $s$-measure zero: 
    $\P (\theta \in \Theta: \betai \vithe = \betak v_k(\theta)\text{ for some $i\neq k$}) = 0 $.

    When one and thus all of the above conditions hold for some $\b$, 
    the gradient $\nabla_\b f(\t, \b)$ is well-defined for $s$-almost every $\t$, and 
    $\nabla \fbar (\beta) = \E[\nabla f(\t, \b)]$.
    Proof and further technical remarks given in \cref{proof:sec:analytical_properties_of_dual_obj}.
\end{theorem}

\subsection{Second-order Differentiability} \label{sec:second_diff}

Given the neat characterization of differentiability of the dual objective via the 
gap function $\bidgap(\beta,\theta)$, 
it is natural to explore higher-order smoothness, which is needed for some of our asymptotic normality results.
\new{On the negative side, Appendix \cref{eg:cont_value}  
gives an example where \cref{eq:as:notie} holds at a point $\b$ and $\fbar$ is differentiable in a neighborhood of $\b$, yet $\fbar$ is not twice differentiable at $\b$.}
We now provide sufficient conditions that imply twice differentiability of $H$.

\begin{theorem}[Second-order differentiability, informal]
    \label{thm:second_order_informal}
    If any of the following conditions hold, then $\fbar$ and thus $H$ are $C^2$ at a point $\beta$. (i) A stronger form of \cref{eq:as:notie} holds: $\E[\bidgap(\beta,\t)\inv] < \infty$. (ii) The angular component of the random vector $v=(v_1,\dots,v_n):\Theta \to \Rnp$ is smoothly distributed. 
    (iii) $\Theta = [0,1]$, $s$ is the Lebesgue measure, the valuations $v_i(\cdot)$'s are linear functions, integrating to~1, with distinct intercepts, and $\b$ is the equilibrium inverse bang-per-buck in LFM.
    See Appendix~\ref{sec:analytical_formal} for formal statements.
\end{theorem}

\new{
    We briefly comment on the inverse bid gap integrability condition. 
    By \cref{thm:first_differentiability} we already know a necessary and sufficient condition for first-order differentiability is that for all items (up to a measure-zero set) there is a positive bid gap. 
    The integrability condition essentially guarantees that for most items, the bid gap is sufficiently  positive.
    We prove sufficiency for twice differentiability by first showing that the bid gap is a Lipschitz constant for the gradient of the EG objective, and then apply the dominated convergence theorem.
}

We show in Appendix~\ref{sec:closedformhessian} that when $H$ is twice differentiable, the Hessian matrix of $H$ has a closed-form expression.

\section{Experiments}

\subsection{Synthetic Experiment}

\subsubsection{Hessian Estimation}
\label{sec:exp_hessain}
Recall that a key component in the variance estimator is the Hessian matrix, which 
we estimate by the finite-difference method in \cref{eq:def:hessian_estimate}.
Finite difference estimation requires the smoothing parameter $\varepsilon_t$. The smoothing $\varepsilon_t$ is used to (1) estimate the active constraints and (2) construct the numerical difference estimator $\hat \cH$. 
\cref{thm:variance_estimation} suggests a choice of 
$\varepsilon_t = t^{-d}$ for some $0 < d < \tfrac12$. 
In \cref{app:hessian_estimation}, we investigate the effect of $d$ numerically. Here we give high-level take-aways.
We find that $d$ represents a bias-variance trade-off. For small $d$, the variance of the estimated value $\hat \cH_{ii}$ is small and yet bias is large. For a large $d$ variance is large and yet the bias is small (the estimates are stationary around some point). Our experiments suggest using $d \in (0.32, 0.47)$.

\subsubsection{Visualization of the FPPE Distribution}
\label{sec:exp_fast_rate}
Next we look at how the FPPE distribution behaves in a simple setting.
We choose the FPPE instances as follows.
Consider a finite FPPE with $n=25$ buyers and $t = 1000$ items. Let $U_i$ be i.i.d.\ uniform random variables on $[0,1]$. 
Buyers' budgets are generated by $b_i = U_i + 1$ for $i=1,\dots, 5$ and $b_i = U_i$ for $i= 6, \dots, 25$. The extra budgets are to ensure we observe $\betasti = 1$ for the first few buyers.
The valuations $\{ v_1, \cdots, v_n\}$ i.i.d. uniform, exponential, or truncated standard normal distributions. 
Under each configuration we form 100 observed FPPEs, and plot the histogram of each $\sqrt{t}( \betagami - \betasti)$. 
The population EG \cref{eq:pop_deg} is a constrained stochastic program and can be solved with stochastic gradient based methods. The true value $\betast$ is computed by the dual averaging algorithm \citep{xiao2010dual}. The mean square error decays as $\E[\| \beta^{\text{da}, t} - \betast\|\sq] = O (t \inv) $ with $t$ being the number of iterations, and so if we choose $t$ large enough, we should still observe asymptotic normality for the quantities $\sqrt t(\beta^{\text{da}, t} - \betagam)$.

\textit{Results.}
\begin{figure}
    \centering
    \includegraphics[scale=.3]{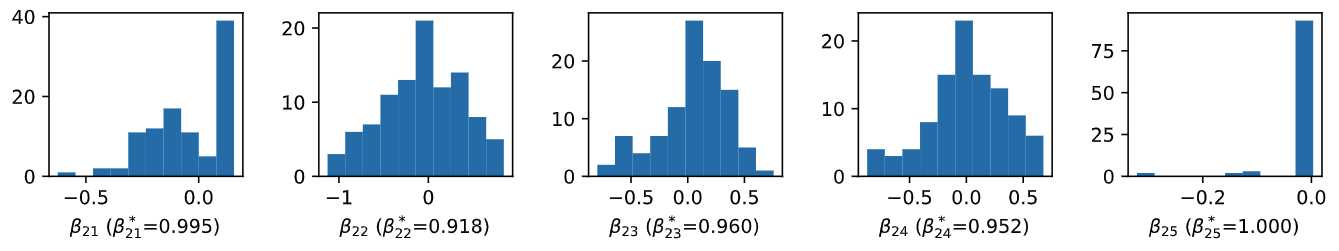}
    \caption{Finite sample distributions of $\sqrt t (\betagami - \betasti)$ of 5 buyers in an FPPE.
    We see that for buyer 25, its finite-sample pacing multiplier is exactly 1 for most of the time.
    For buyer 21, its limit pacing multiplier is very close to 1 and so its distribution is not normal for small samples. For buyers 22 -- 24, their finite sample distribution is close to normal distributions.
    The full figure is in \cref{fig:nonnormality_unif}.}
    \label{fig:nonnormality_normal_compressed}
\end{figure}
\cref{fig:nonnormality_normal_compressed} shows five out of 25 distributions for pacing multipliers.
Full plots for all three distributions are given in \cref{fig:nonnormality_unif,fig:nonnormality_expo,fig:nonnormality_normal}.
We see that (i) if $\betasti < 1$ then the finite sample distribution is close to a normal distribution, and (ii) if $\betasti = 1$ (or very close to $1$, such as $\beta_{14, 21}$ in the uniform value plots, $\beta_{20, 23}$ in exponential), the finite sample distribution puts most of the probability mass at 1.
For cases where $\betasti$ is close, but not very close, to 1, we need to further increase the number of items to observe normality.

\subsection{Semi-real Experiment: Nash Social Welfare Estimation in Instagram Notification System}
\label{sec:exp_instagram}

Notifications are important in enhancing the user experience
and user engagement in mobile apps.
Nevertheless, an excessive barrage of notifications can be disruptive for users. Typically, a mobile application has various notification types, overseen by separate teams, each with potentially conflicting objectives. And so it is necessary to regulate notifications and send only those of most value to users. 
\citet{instagramdata} propose to use Fisher market equilibrium-based methods to efficiently send notifications, where they treat the opportunity to send a user a notification as an item, and different types of notifications as buyers. 
In this section, we use the inference method developed in \cref{sec:lfm_inference} to quantify uncertainty in equilibrium-based notification allocation methods.

\textit{The data.}
The dataset released by \citet{instagramdata}
contains 
about $400,000$ generated notifications of four types for a subset of about $60,000$ Instagram users from September 14--23, 2022.
The four types of notifications (buyers) are 
likes, 
daily digest of stories,
feed suite organic campaign (notification about new
posts on the user's feed), 
and comments subscribed.
The value $v_i(\theta)$ of a notification type~$i$ to a user~$\theta$ at a specific time is predicted by the platform's algorithm and available in the dataset as a numerical value in $[0,1]$.
The budgets of notification types are also given.
For a user-notification type pair, we average over the whole time window and use the average to represent $v_i(\theta)$, resulting in a user-notification type matrix.
However, even after aggregation over time, there are lots of missing values, i.e., many users do not have every notification type generate a potential notification.

\textit{Value imputation and simulation by the Gaussian copula.}
We assume the values in the notification system admit the following representation. 
There exists unique monotone functions $f_i:\R \to [0,1]$, such that $(v_1, \dots, v_4) = (f_1(Z_1), \dots, f_4 (Z_4))$, where $Z = [Z_1, \dots, Z_4]$ follows a multivariate Gaussian distribution with standard normal marginals. Such an assumption is equivalent to assuming the value distribution possesses a Gaussian copula (\citealt[Lemma 1]{zhao2020missing} and \citealt[Lemma 1]{liu2009nonparanormal}).
Given this representation, we propose a two-step simulation method. In the first step, we learn the monotone functions by matching the quantiles of values with the quantiles of a standard normal. 
We use isotonic regression to learn the monotone functions.
Second, given the learned functions $\hat f_i$ and inverses  $\hat f_i \inv : [0,1] \to \R$, we transform $v_i$ to $\hat f_i \inv(v_i)$, and compute the covariance matrix of $\hat f_i \inv(v_i)$, denoted $\hat \Sigma$.
Even though some values are missing, the covariance $\Sigma$ can still be estimated by $\hat \Sigma$ if values are missing completely at random. 
\footnote{
    \new{
        The validity of the copula imputation method relies on the missing completely at random assumption (MCAR) \citep{zhao2020missing}, i.e., values and missingness are independent. 
        Unfortunately, we cannot determine whether this is true unless accessing the missing data. 
        MCAR will be assumed in this experiment.
    }
}
Now to simulate a new item for buyers,
we draw $Z \sim N(0, \hat \Sigma)$, and return $[\hat f_1 (Z_1),\dots, \hat f_4 (Z_4)] \tp $ as the value. 
There are multiple advantages to this method. First, the dependence structure of the available dataset is preserved. Second, the generated values are within the range $[0,1]$ as the original data are. Third, the marginal distribution of values are also preserved in the simulated data. 

As a final step to mimic realistic data, since some users may turn off notifications of certain types, the values of those notifications will be zero. We simulate this by setting certain values to zero 
according to the sparsity pattern in the original dataset.
See \cref{fig:fair_noti_data} for a comparison between original dataset and simulated data.

\begin{figure}
    \centering
    \begin{minipage}{.45\textwidth}
      \centering
      \includegraphics[width=\linewidth]{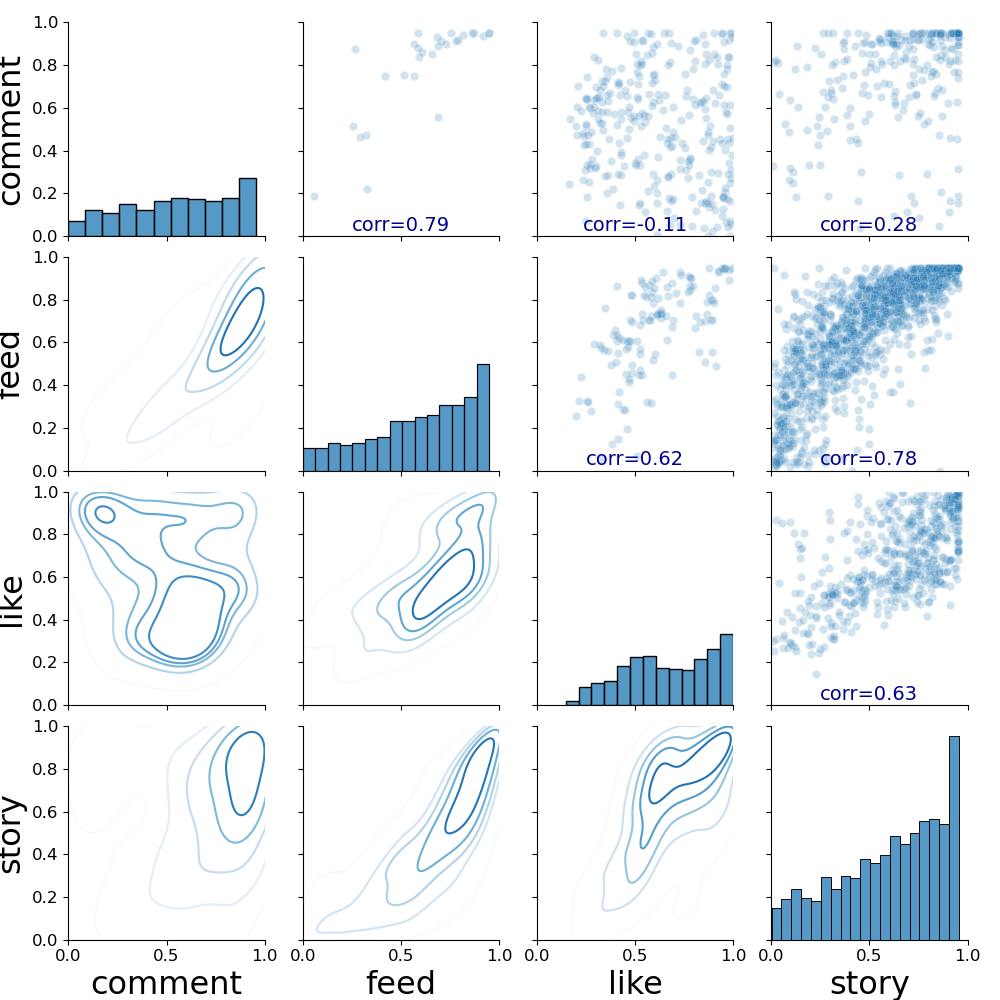}
      \label{fig:test1}
    \end{minipage}%
    \hspace{1cm}
    \begin{minipage}{.45\textwidth}
      \centering
      \includegraphics[width=\linewidth]{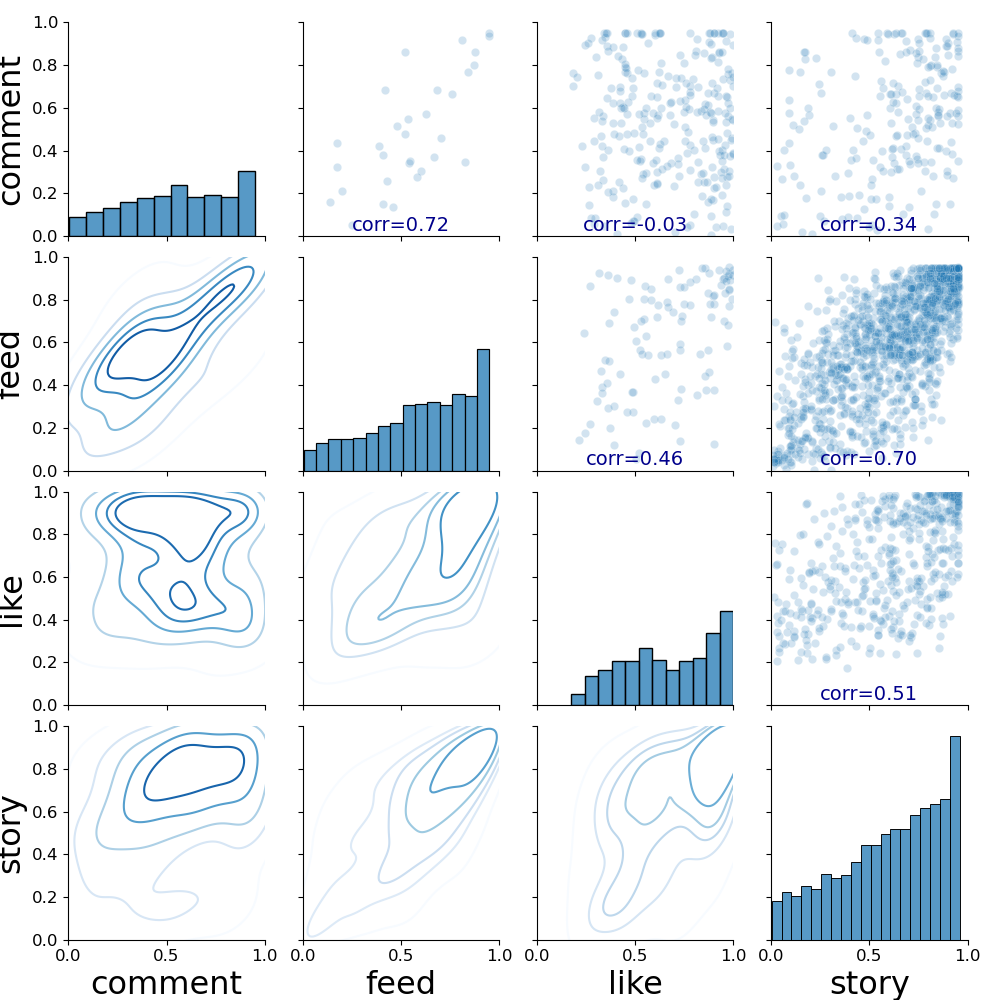}
      \label{fig:test2}
    \end{minipage}
    \caption{\new{The fair notification allocation data. Left: original data with missing values. Right: simulated data.}}
    \label{fig:fair_noti_data}
    \end{figure}

\textit{Setup and results.}
We apply the confidence interval in \cref{thm:ci_lnsw} and study the coverage properties.
The nominal coverage rate is set to 95\%. 
First, we do see that even for a small sample size of 100, the nominal coverage rate is achieved. And as we increase the size of markets $t$, the coverage maintains at around 95\% and the width of the CI shrinks roughly at the rate $1/\sqrt t$. 
\begin{table}
    \centering
    \begin{tabular}{l|r}
        \toprule
        items & (coverage rate, width of CI) \\
        \midrule
        100 & (0.94, 0.88) \\
        200 & (0.95, 0.63) \\
        400 & (0.93, 0.43) \\
        600 & (0.97, 0.35) \\
        \bottomrule
    \end{tabular}
    \caption{Coverage rate of log Nash social welfare in fair notification allocation. To help interpret the CI width, the log Nash social welfare in the limit market is around -16.}
    \label{tbl:nsw}
\end{table}

\subsection{Semi-real Experiment: A/B Testing of Revenue in First-Price Auction Platforms}
\label{sec:exp_ipinyou}
In this section we apply our revenue estimation method to a real-world dataset, the iPinYou dataset \citep{liao2014ipinyou}.
The iPinYou dataset \citep{liao2014ipinyou} contains raw log data of the bid, impression, click, and conversion history on the iPinYou platform
in the weeks of 
March 11--17, 
June 8--15
and October 19--27.
We use the impression and click data of 5 advertisers on June 8, 2013, containing a total of 1.8 million impressions and 1,200 clicks.
As in the main text, let $i \in \{1,2,3,4,5\}$ index advertisers (buyers) and let $\tau$ index impressions/users (items in FPPE terminology).
The five advertiser are labeled by number and their categories are given: 1459 (Chinese e-commerce), 3358 (software), 3386 (international e-commerce) and 3476 (tire).
From the raw log data, the following dataset can be extracted.
The response variable is a binary variable $ \click_i^\tau \in \{0,1\}$ that indicates whether the user clicked the ad or not. 
The relevant predictors include 
a categorical variable $\AE$ of three levels that records from which ad-exchange the impression 
was generated, a categorical variable $\RG$ of 35 levels indicating provinces of user IPs, and finally 44 boolean variables, each a $\TG$, indicating whether a user belongs to certain user groups  defined based on 
demographic, geographic and other information.
We select the top-10 most frequent user tags and denote them by $\TG_1,\dots,\TG_{10} \in \{0,1\}$. 
Both $\AE$ and $\TG$ are masked, and we do not know their real-world meaning.

\textit{Simulate advertisers with logistic regression.}
The raw data contains only five advertisers. In order to simulate new realistic advertiser, we fit a logistic regression and then perturb the fitted coefficients to generate more advertisers.
We posit the following logistic regression model for click-through rates (CTRs).
For a user $\tau$ that saw the ad of advertiser $i$, the click process is governed by
\begin{align*}
   & \CTR_i^\tau = \P(\click_i^\tau = 1 \given \thetau) = 
    \frac{1}{1+\exp (w_i\tp \theta^\tau)},
   \\
   &\theta^\tau = [1, \AE_2, \AE_3, \RG_2,\dots, \RG_{35}, \TG_1, \dots, \TG_{10}]\in \{1\} \times \{0,1\}^{46}
\end{align*}
where the weight vectors $w_i \in \R^{47}$ are the coefficients to be estimated from the data.
Note that $\AE_1$ and $\RG_1$ are absorbed in the intercept.
By running 5 logistic regressions, we obtain regression coefficients $w_1, w_2, \dots, w_5$.
To visualize the fitted regression, in \cref{fig:real} we show the estimated click-through rate distributions of the five advertisers. The diagonal plots are the histogram of CTRs, and the off-diagonal panels are the pair-wise scatter plots of CTRs.
To generate more advertisers, we take a convex combination of the coefficients $w_i$'s, add uniform noise, and obtain a new parameter, say $w'$. Given an item, the CTR of the newly generated advertisers will be $\frac{1}{1 + \exp(\theta \tp w')}$. 
The value distribution is the historical distribution of the simulated advertisers' predicted CTRs of the 1.8 million impressions.

\begin{figure}[ht!]
    \centering
    \includegraphics[scale = 0.6]{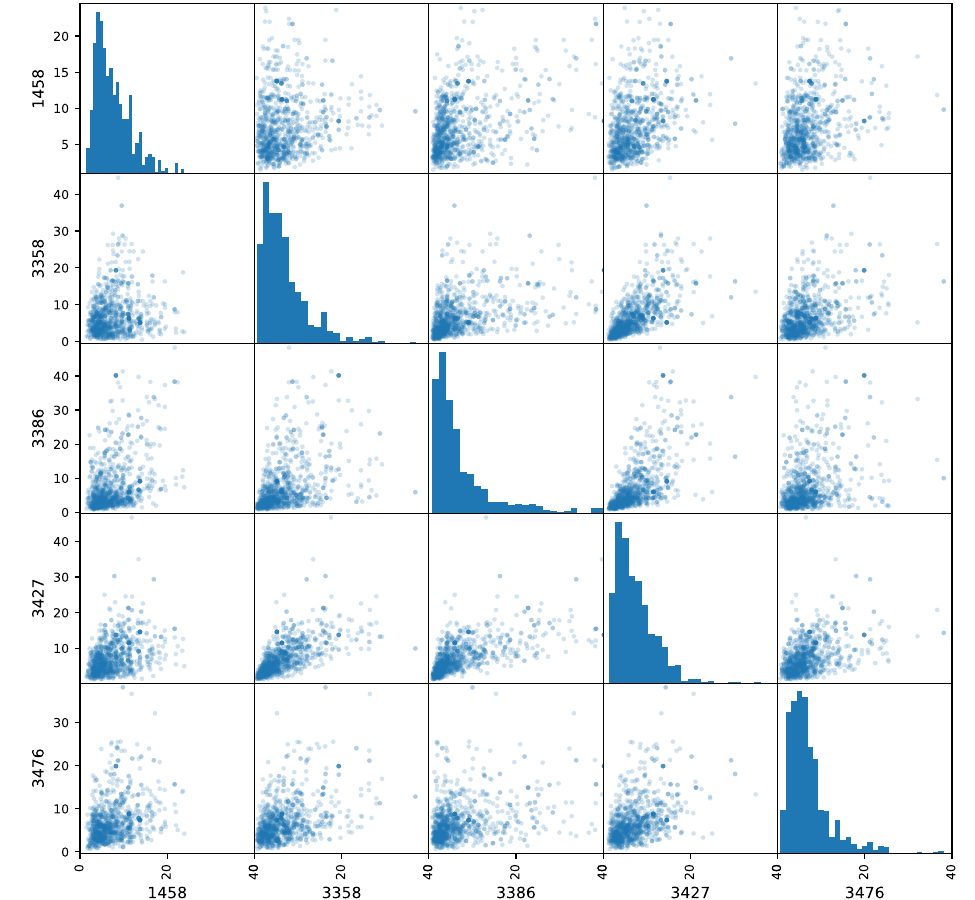}
    \caption{Click-through rate (in basis points, i.e. 0.01\%) distributions from logistic regression. }
    \label{fig:real}
\end{figure}

\subsubsection{Revenue Coverage with Hessian Estimation}
\label{sec:exp_rev_coverage}
In this section we aim to produce confidence interval of the revenue with 
the CI constructed from \cref{eq:plugin variance}, with a focus on the effect of Hessian estimation on the coverage.
Firstly, the sum equals $n$ times the average price-per-utility of advertisers, a measure of efficiency of the system. Secondly, since most quantities in FPPE, such as revenue and social welfare, are smooth functions of \pmrs, being able to perform inference about a linear combination of $\beta$'s indicates the ability to infer first-order estimates of those quantities.

\textit{Setup.}
An experiment has parameters $(t,n, d, \alpha)$. Here 
$t$ is the number of items,
$n$ the number of advertisers, 
and $\alpha$ is the proportion of advertisers that are not budget-constrained (i.e., $\beta = 1$).
Parameter $d$ is a tuning parameter of the revenue variance estimator. It is the exponent of the finite-difference stepsize $\varepsilon_t$ in \cref{eq:def:hessian_estimate}, i.e,  $\varepsilon_{\cH, t}= t ^{-d}$.
To control $\alpha$ in the experiments, we select budgets as follows. Give infinite budgets to the first $\lfloor \alpha n \rfloor $ advertisers. Initialize the rest of the advertisers' budgets randomly, and keep decreasing their budgets until their \pmrs are strictly less than 1.
In the experiment $(t,n, d, \alpha)$, we first compute the \pmr in the limit market using dual averaging \citep{xiao2010dual,gao2021online,liao2022dualaveraging}.
Then we sample one FPPE by drawing values from the synthetic value distribution obtained previously.
Now given one FPPE, apply the formula in \cref{eq:rev_variance}
to construct CI and record 
coverage.
The reported coverage rate for an experiment with parameters $ (t,n, d, \alpha)$ is averaged over 100 FPPEs.

\textit{Results.}
Representative results are presented in \cref{tbl:rev_coverage_short}; we present the full table in \cref{tbl:rev_coverage_full}. As the number of item increases, we observe the empirical coverage rate achieving the nominal 90\% coverage rate, while the width of confidence interval is narrowing.
We also observe that the confidence interval is robust against the Hessian estimation and the proportion of unpaced buyers; for different choices of 
exponent in the differencing stepsize ($\varepsilon_{\cH, t}$ in \cref{eq:def:hessian_estimate}) and proportion of unpaced buyers ($\alpha$), the coverage performance remains similar.

\begin{figure}[ht]
    \centering
    \includegraphics[scale = 0.6]{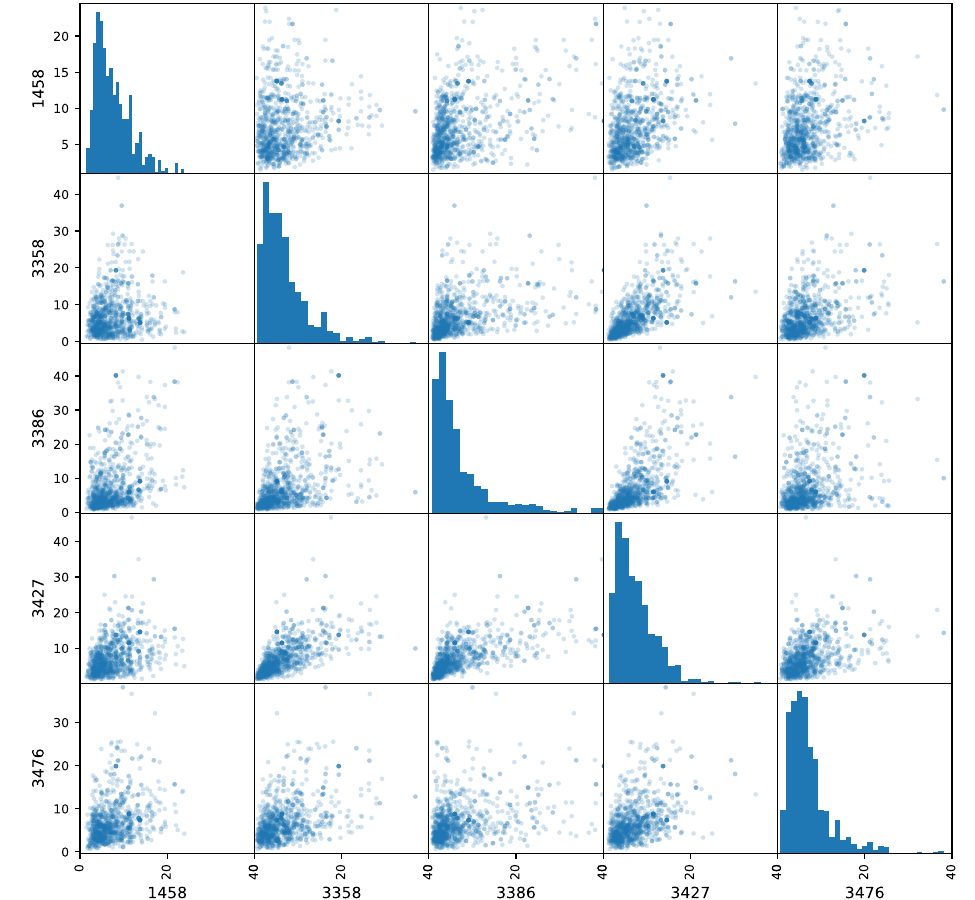}
    \caption{Click-through rate (in basis points, i.e. 0.01\%) distributions from logistic regression. }
    \label{fig:real}
\end{figure}

\input{tbl_rev_coverage_short.tex}

\subsubsection{Coverage Comparison against a Simple I.I.D.\ Price Model}
\label{sec:exp_rev_coverage}
\new{

In this section we compare the Hessian-based revenue variance estimator in \cref{eq:plugin variance} with $\varepsilon_{\cH,t} = t^{-0.4}$, the Hessian-free version in \cref{eq:simplified variance estimator}, and a baseline estimator based on an i.i.d.\ price model.

The baseline variance estimator is $\hat \sigma \sq _{\text{naive}} = \frac1t \sumtau (\ptau - \bar p) \sq $, where $\bar p = \frac1t \sumtau \ptau$. The limit of this estimator is  
$
    \sigma_{\text{naive}} \sq = \var(p^*(\t)) 
$.
The statistical model behind this method would be to assume that $\{ p^1,\dots, p^t\}$ are iid draws from some distribution, while totally ignoring the auction mechanism and the bidding behavior of the buyers. 
This is in contrast to the correct formula $\var(\tilde p ^*)$ in \cref{thm:simplified_rev_var} when a bid-gap condition holds.
We note that this assumption is not correct, since the pacing multipliers change as a function of which items are sampled. Nonetheless, it has the following intuitive appeal: if we assume that the pacing multipliers \emph{are fixed}, then this estimator does correctly estimate the price variance, regardless of market size (though note that revenue variance is not the same as price variance). Moreover, our results do show that the pacing multipliers, while not fixed, do concentrate around the limit pacing multipliers as the market grows.

We compare the confidence interval of the revenue with 
the CI constructed from \cref{eq:rev_variance} against the naive CI described above.
The reported coverage rate is over 500 FPPEs.

\textit{Results.}
Representative results are presented in \cref{tbl:rev_coverage_new_vs_naive_short}.
We present the full table in \cref{tbl:rev_coverage_new_vs_naive_full}.
We see that (1) the Hessian-free method is  significantly better in the market where all buyers are budget exhausted ($I_= = \emptyset$, $\alpha = 0$ in the table). The Hessian-free method  is able to determine that all the budgets will be exhausted, and thus the revenue equals the sum of budgets deterministically, while the naive method ignores this simple fact and puts an unnecessary interval on revenue.
Note that the three methods will coincide when $\alpha = 1$ in large samples. 
(2) Although overall the three methods achieve the nominal coverage rate, the Hessian-based method is theoretically valid under mild assumptions on the FPPE, while the Hessian-free method is valid under an additional bid-gap assumption but is computationally cheaper since it does not need Hessian estimation.

\input{tbl_rev_coverage_short_comparison.tex}

}

\subsubsection{Treatment Effect Coverage}
\label{sec:experiment_abtest}
\textit{Setup.}
In this experiment, we fix the differencing stepsize in Hessian estimation to be 
$\varepsilon_{\cH,t} = t^{-0.4}$ and the proportion of unpaced buyers to be 30\%, which 
is a realistic number for real-world auction platforms.

An experiment has parameters $(t,n, \pi)$, where $t$ is the number of 
items, $n$ the number of users, and $\pi$ the treatment probability (see \cref{sec:ab testing}). To model treatment application, we use the shift of value distribution. Choose two sets of logistic regression parameters, $\{ w_i(0) \}_i$
and $\{w_i (1)\}$. Then if a user $\t$ is applied treatment $\omega \in \{0,1\}$, then its value to buyer $i$ will be $ 1 / (1 + \exp(w_i(\omega) \tp \theta))$, $i\in[n]$.
In an experiment, 
the limit revenues in the two limit markets $\FPPE(b,v(0),s,\Theta)$ and $\FPPE(b,v(1), s,\Theta)$, and the limit treatment effect will be calculated first.
Then we perform the a/b test experiments 100 times, and construct 100 CIs for the treatment effect. 
We report coverage rates and widths of CIs.

\textit{Results.} Representative results are presented in \cref{tbl:abtest_short}; in \cref{tbl:abtest_full} we present the full results.
First, the overall coverage rates across different market setups and treatment probabilities $\pi$ are around the nominal 90\%, and the width of confidence interval shrinks as sample size grows.
Second, when the number of items is sufficiently large (say over 400), we observe 
a $U$-shape relationship between the width of CI and treatment probability $\pi$; the CI widths are wider when $\pi$ is close to the extreme points (say 0.1 and 0.9) than when $\pi$ stays away from the extreme points.
This is explained by the treatment effect variance formula in \cref{thm:clt_ab_testing}. Holding the two variances fixed, the treatment effect variance tends to infinity if we send $\pi \to 0$ or $1$.

\input{tbl_abtest_short.tex}

\section{Conclusion}

We introduced a theory of statistical inference for Fisher markets, resource allocation systems that deploy the CEEI mechanism, and first-price auction platforms. 
We showed that quantities observed in the finite market equilibrium observed from these systems are good estimators of their corresponding limit market values. 
We presented convergence rate results, asymptotic distribution characterizations, local minimax optimality results, and constructed confidence interval tools.
Finally, we showed how to use these tools to develop a theory of statistical inference in A/B testing under competition effects.

A few open questions remain. 
In practice, the item arrival process exhibits nonstationarity and seasonality. A statistical theory for LFM and FPPE that incorporates temporal dependence is desirable.
It would also be desirable to design a notion of online confidence intervals for the limit market, since, in practice, items typically arrive on platforms sequentially. 
For FPPE we assumed the absence of degenerate buyers in \cref{as:scs}; lifting this assumption would be interesting. 
Finally, we restricted our attention to the first-price setting for FPPE. In practice, second-price auctions are also widespread. A theory of statistical inference for second-price auctions is also desirable, though we expect it to be significantly weaker, due to computational complexity barriers, as well as non-uniqueness issues.

\new{
Beyond linear Fisher markets and FPPE, it would be interesting to investigate whether our proof techniques apply to other equilibrium models captured by mathematical programs.
The Eisenberg-Gale convex program is known to work for several utility classes beyond linear utilities (as we study) including Leontief, constant-elasticity-of-substitution (CES), and Cobb-Douglas utilities~\citep{nisan2007algorithmic}.
Later work has also showed that convex program variants exist for e.g. spending-restricted and utility-restricted versions of Fisher markets~\citep{cole2017convex}.
In all these cases, it is possible to derive dual programs similar to the one we leverage, and they have a sum over prices in the objective~\citep{cole2017convex}, which may lend itself to our stochastic approximation approach (SAA). 
Several different convex programs are known for the Arrow-Debreu exchange model as well, see e.g.~\citet{devanur2016rational}. 

A/B testing using one market is interesting.
When both treatments are applied to the same market, the market equilibrium can still be described, and inferences about it can be made, using our theory. 
    However, under this design, there are two layers of interference: 
    (1) treatments 1 and 0 interfere with each other now that they are in the same market, and (2) the interference induced by market equilibrium. 
    In this new setting, there are many fundamental questions that would need to be answered, such as what is the correct notion of treatment effect, and whether a given treatment effect can be estimated from observation of a single market. 
    We believe this is a very interesting problem, and many of our tools can probably be used to analyze this problem in an FPPE framework. But the first step would be to formulate the types of treatment effects one would even work with, and we think this is a whole new paper worth of material to work out.

    Buyer-side treatments, such as modifications to the advertiser UI for specifying auction parameters like target audiences and return-on-investment, are also of interest. However, since ad campaigns are typically configured only once, the budget-split design is not a natural fit for studying these treatments unless all buyers are required to use both the treatment and control UIs twice.
}

%% file: fisher_market_results.tex
\section{Statistical Results for Linear Fisher Markets}
We now turn to investigating the statistical convergence properties of finite LFMs to the limit LFM.
Suppose we sample an LFM ($\oLFM(b,v,1/t, \gamma)$), where $\gamma$ consists of $t$ i.i.d.\ samples from $s$.
We will study how such finite LFMs are distributed around the limit LFM ($\LFM(b,v,s,\Theta)$) as $t$ grows.
We focus on convergence of the following quantities: individual utilities, (log) Nash social welfare (NSW), and the pacing multiplier vector $\b$ (which characterizes the equilibrium, as shown in \cref{eq:pop_deg_lfm}).
\cref{sec:lfm_basic} presents strong consistency and convergence rate results.
\cref{sec:lfm_asym} presents asymptotic distributions for the quantities of interest, 
and a local minimax theory based on \citet{le2000asymptotics}, showing that the finite LFM provides an optimal estimate for the limit LFM in a local asymptotic sense.
\cref{sec:lfm_inference} discusses estimation of asymptotic variance of NSW.

\subsection{Basic Convergence Properties}
\label{sec:lfm_basic}
In this section we show that we can treat observed quantities in the finite LFM as consistent estimators of their counterparts in the limit LFM.
Below we state the consistency results; the formal versions can be found in \cref{sec:consistency}.
We say an estimator sequence $\{ \hat a_t\}$ is strongly consistent for $a$ if $\P(\lim_{t \to \infty} \hat a_t = a) = 1$.

\begin{theorem}[Strong Consistency]
    The NSW, \new{approximate equilibrium} pacing multipliers, and utility vectors in the finite LFM are strongly consistent estimators of their counterparts in the limit LFM. 
\end{theorem}

Next, we refine the consistency results and 
provide finite sample guarantees. We start by focusing on Nash social welfare
and the set of approximate market equilibria. The convergence of utilities and 
pacing multipliers will then be derived from the latter result.

\begin{theorem} \label{thm:lnsw_concentration}
    For any failure probability $0< \eta < 1$, let $t \geq 2 {\vbarsq {\log(4n/\eta)}}$. Then with probability greater than $1-\eta$, we have $
        |\LNSW^\gam - \LNSW^* | 
        \leq O(1) {\vbar \big(\sqrt{n\log ((n+\vbar)t)} +  \sqrt{\log(1/\eta)} \big)}{ t^{-1/2}},
        $
    where $O(1)$ hides only constants.
    Proof in Appendix~\ref{sec:proof:thm:lnsw_concentration}.

\end{theorem}

In more intuitive wording, \cref{thm:lnsw_concentration} establishes a high-probability convergence rate $|\LNSW^\gam - \LNSW^* | = \tilde O _p(\vbar \sqrt n t^{-1/2})$. The proof proceeds by first establishing a pointwise concentration inequality and then applies a discretization argument.

To state the result for the pacing multipliers $\b$, 
we define approximate market equilibria (which we define in terms of approximately optimal pacing multiplier vectors $\beta$). 
Let 
\begin{align}
    \cB^\gam(\epsilon)  \defeq \{ \beta \in \Rnpp: H_t(\beta) \leq \inf_\beta H_t(\beta) + \epsilon \} 
    , \; 
    \cB^*(\epsilon)  \defeq \{ \beta \in \Rnpp: H(\beta) \leq \inf_\beta H(\beta) + \epsilon \}
    \label{eq:def:cBstar}
    \;.
\end{align}
be the sets of $\epsilon$-approximate solutions to the sample and the population EG programs, respectively. 
The next theorem shows that the set of $\epsilon/2$-approximate solutions to the sample EG program is contained in the set of $\epsilon$-approximate solutions to the population EG program with high probability.

\begin{theorem}[Convergence of Approximate Market Equilibrium] \label{thm:high_prob_containment}
    Let $\eps > 0$ be a tolerance parameter and $\alpha \in (0,1)$ be a failure probability. Then for any $0\leq \delta \leq \eps/2$, to ensure 
        $ \P \big(
        C_\LFM\cap \cBgam(\delta) 
        \subset 
        C_\LFM\cap \cBst(\epsilon)
        \big) \geq 1 -2 \alpha$
    it suffices to set
    \begin{align}
        \label{eq:t_condition}
        t \geq O(1)\vbarsq \min\bigg\{  \frac{1}{ 
            \ubar{b}\epsilon} , \frac{1}{\epsilon\sq} \bigg\}\bigg(n \log\Big(\tfrac{16(2n+\vbar)}{\eps - \delta}\Big) + \log\tfrac{1}{\alpha}\bigg)
            \;,
    \end{align}
    where 
    the set
    $C_\LFM$, defined in 
    \cref{eq:def_C_LFM}, is the natural region in which $\betast$ must lie, and $O(1)$ hides absolute constants.
    Proof in Appendix~\ref{sec:proof:thm:high_prob_containment}.

\end{theorem}

By construction of $C_\LFM$ we know $\betast \in C_\LFM$ holds, and so $C_\LFM\cap \cBst(\epsilon)$ is not empty. 
In the appendix, \cref{lm:value_concentration} shows that for $t$ sufficiently large, $\betagam \in C_\LFM$ with high probability, in which case the set 
$C_\LFM\cap \cBgam(\delta)$ is not empty. 
Setting $\delta = 0$ in \cref{thm:high_prob_containment} we obtain the corollary below.

\begin{corollary}
    \label{cor:H_concentration}
    Let $t$ satisfy \cref{eq:t_condition}. Then with probability $\geq 1 - 2 \alpha$ it holds $H(\betagam) \leq H(\betast) + \epsilon$.
\end{corollary}

More importantly, it establishes the fast statistical rate 
$    H(\betagam ) - H(\betast) = \new{\tilde{O}_p ( n \vbar\sq / \ubar b \cdot t\inv )}
$
for $t$ sufficiently large, where we use $\tilde{O}_p$ to ignore logarithmic factors.
In words: the limit LFM objective value of the \emph{finite} solution $\betagam$ converges to the optimal limit LFM objective value at a $1/t$ rate.

By strong convexity of the dual objective, the containment result can be translated to high-probability convergence of the pacing multipliers and the utility vector.

\begin{corollary}
    \label{cor:beta_u_concentration}
    Let $t$ satisfy \cref{eq:t_condition}. Then with probability at least $1 - 2 \alpha$ we have  
    $\| \betagam - \betast \|_2 \leq \sqrt{\frac{8\epsilon }{ \ubar{b}}}$ and $\| \ugam - \ust \|_2 \leq \frac{4}{\ubar{b}}\sqrt{8\epsilon / \ubar{b}}$.
\end{corollary}

We compare the above corollary with Theorem~9 from~\citet{gao2022infinite} which establishes a convergence rate of the stochastic approximation estimator based on the dual averaging algorithm~\citep{xiao2010dual}. 
In particular, they show that the average of the iterates, denoted $\beta_{\mathrm{DA}}$, enjoys a convergence rate of 
$\| \beta_{\mathrm{DA}} - \betast\|_2^2 = \tilde O_p \big(\frac{\vbar\sq }{ \ubar{b}\sq} \frac{1}{t}\big)$,
where $t$ is the number of sampled items.
The rate achieved in \cref{cor:beta_u_concentration} is 
$\|\betagam - \betast\|_2 \sq = \tilde O_p \big(\frac{n \vbarsq}{\ubar{b}\sq} \frac{1}{t} \big)$. 
We see that our rate is worse off by a factor of $n$. 
We conjecture that it can be removed by using the more involved localization arguments \citep{Bartlett2005}.
\footnote{
    \new{
        A recent work by \citet{liu2024metric} shows the dimension dependence can be improved using a stability argument.
    }
}
\new{On the positive side, our estimates are produced by the strategic behavior of the agents without any extra computation, and can be observed directly from $\betagam$. In contrast, the computation of the dual averaging estimator requires knowledge of the values $v_i(\theta)$.}

\subsection{Asymptotics of Linear Fisher Market}

\label{sec:lfm_asym}

In this section we derive asymptotic normality results for Nash
social welfare, utilities and pacing multipliers.
As we will see, a central limit theorem for Nash social welfare holds under basically no additional assumptions. 
However, the CLTs of pacing multipliers and utilities will require twice continuous differentiability of the population dual objective $H$ at optimality, with a nonsingular Hessian matrix.
We present CLT results under such a premise; \cref{thm:second_order_informal} gives three quite general settings under which these conditions hold. 

\begin{theorem}[Asymptotic Normality of Nash Social Welfare]\label{thm:normality}
    It holds that
    \begin{align}  \label{it:thm:normality:1}
        \sqrt{t}(\NSW^\gam - \NSW^*) \tod \cN 
        (0, \sigma\sq_{\NSW})
        \;,
    \end{align}    
    where $ \sigma\sq_{\NSW} 
    = \var(\pst)
    =\int(p^*)\sq s \dt  -(\int p^* s \dt )\sq $.
    Proof in Appendix~\ref{sec:proof:thm:normality}.
\end{theorem}

As stated previously, our asymptotic results for $\beta$ and $u$ require that $H$ is twice-continuously differentiable at $\betast$.
When this differentiability holds the set of items that incur ties is $s$-measure zero (see \cref{thm:first_differentiability}), and thus the equilibrium allocation $\xst$ in the limit LFM is unique and must be pure. 
Now we define a map $\must: \Theta \to \Rnp$, which represents the utility all buyers obtain from the item $\t$ at equilibrium. Formally,
\begin{align}
    \must(\t) = [\xst_1 (\t) v_1 (\t), \dots, \xst_n(\t) v_n(\t)]\tp.
\end{align}
Since $\xst$ is pure, only one entry of $\must(\t)$ is nonzero. 
Clearly $\int \must s \diff \t = [\ust_1,\dots, \ust_n]\tp$.

\begin{theorem}[Asymptotic Normality of Pacing Multipliers and Utilities] \label{thm:clt_beta_u}
    Let \cref{as:smo} hold with non-singular Hessian matrix $\cH=\nabla \sq H(\betast)$. 
Then 
\begin{align}
    \sqrt{t} (\beta^\gam - \beta^*) \tod \allowbreak\cN \big(0, 
        \Sigma_\beta
        \big), \quad 
\sqrt{t}(u^\gam - u^*) \tod \allowbreak\cN\big
        (0, 
        \Sigma_u
        \big) ,
\end{align}     
    where 
         $\Sigma_\beta = \cH\inv\allowbreak 
         \cov(\must)
        \cH\inv $ 
    and
        $\Sigma_u = 
        \Diag(-b_i/(\betasti)\sq )\allowbreak 
        \cH\inv\allowbreak 
        \cov(\must) \allowbreak 
        \cH\inv\allowbreak 
        \Diag(-b_i/(\betasti)\sq ) .$ 
    Proof in Appendix~\ref{sec:proof:thm:normality}.

\end{theorem}
\cref{thm:normality} can also be derived from \cref{thm:clt_beta_u} using the delta method, since $\NSWst = \sumiton b_i \log (\usti) = \sumi b_i \log (b_i / \betasti)$ is a smooth function of $\betast$.

We will show that the asymptotic variances in \cref{thm:clt_beta_u} are the best achievable, in an asymptotic local minimax sense. To make this precise, we need to introduce ``supply neighborhoods'' obtained through perturbing the original supply~$s$.

\subsubsection{Perturbed Supply}
\label{sec:perturb_supply}
First we introduce notation to parametrize neighborhoods of the supply $s$.
Let $g\in G_d = \{g:\Theta \to \R^d: \E[g]=0, \E[\|g\|\sq]<\infty\}$ be a direction along which we wish to perturb the supply $s$.
Given a vector $\alpha \in \R^d $ signifying the magnitude of perturbation,
we want to scale the original supply of item $\theta$ by  $\exp(\alpha\tp g(\theta))$ and then obtain a perturbed supply distribution by appropriate normalization.
To do this we define the perturbed supply by  
\footnote{
    In \citet{duchi2021asymptotic} they allow more general classes of perturbations, we specialize their results for our purposes.
}
\begin{align}
    \label{eq:perturbed_is_roughly_expo}
    s_{\alpha,g} (\theta) \defeq C \inv [{1 + \alpha\tp g(\theta)}] s(\theta) 
\end{align}
with a normalizing constant $C = 1 + \int \alpha\tp g(\theta) s(\theta)\diff\theta$.  
As $\alpha\to 0$, the perturbed supply $s_{\alpha,g}$ effectively approximates
$
    s_{\alpha,g} (\theta) \propto \exp(\alpha\tp g(\theta)) s(\theta) 
.$

We let $\betast_{\alpha,g}$, $\ust_{\alpha,g}$ and $\NSWst_{\alpha,g}$ be the limit inverse bang-per-buck, price and revenue in $\LFM(b,v,s_{\alpha,g}, \Theta)$.
Clearly $\betast = \betast_{0,g}$ for any $g$ and similarly for $\ust_{0,g}$ and $\NSW_{0, g}$.

\subsubsection{Asymptotic Local Minimax Optimality}

Given the asymptotic normality of observed LFM, it is desirable to understand the best possible statistical procedure for estimating the limit LFM.
One way to discuss the optimality is to measure the difficulty of estimating the limit LFM when the supply distribution varies over small neighborhoods of the true supply $s$, asymptotically. When an estimator achieves the best worst-case risk over these small neighborhoods, we say it is asymptotically locally minimax optimal. For general references, see \citet{vaart1996weak,le2000asymptotics}. More recently \citet[Sec.\ 3.2]{duchi2021asymptotic} develop asymptotic local minimax theory for constrained convex optimization, and we rely on their results.

Let $L:\Rn \to \R$ be any symmetric quasi-convex loss \footnote{A function is quasi-convex if its sublevel sets are convex.}. 
In asymptotic local minimax theory we are interested in the local asymptotic risk: given a sequence of estimators $\{ \hat \beta_t: \Theta^t \to \Rn \}_t$,
\begin{align*}
\LAR_\beta (\{\hat \beta_t\}_t )  \defeq 
    {\sup_{ g\in G_d, d\in \mathbb{N}}
    \lim_{c\to \infty}
    \liminf_{t \to \infty}
    \sup_{\|\alpha\|_2\leq \frac{c}{\sqrt t}}
    \E_{s_{\alpha,g}^{\otimes t}}[L(\sqrt{t} (\hat \beta_t - \betast_{\alpha,g} ))] } \;.
\end{align*}
If we ignore the limits and consider a fixed $t$, then $\LAR_\beta$ roughly measures the 
worst-case risk for the estimators $\{\hat \beta_t\}$.
Note that $\alpha$ is a $d$-vector, and thus the shrinking norm-balls depend on $d$, and the expectation is taken w.r.t.\ the $t$-fold product of the perturbed supply.

Similarly, define the risk for utility $u$ (resp. Nash social welfare $\NSW$) 
given an estimator sequence $\{ \hat u _t\}$ (resp. $\{  \widehat \NSW_t\}$): 
\begin{align*}
    &\LAR_u(\{\hat u_t\}_t ) =
    {\sup_{ g\in G_d, d\in \mathbb{N}}
    \lim_{c\to \infty}
    \liminf_{t \to \infty}
    \sup_{\|\alpha\|_2\leq \frac{c}{\sqrt t}}
    \E_{s_{\alpha,g}^{\otimes t}}[L(\sqrt{t} (\hat u_t - u^*_{\alpha,g} ))] } \;,
    \\
    & \LAR_\NSW(\{\widehat \NSW_t\}_t ) =
    {\sup_{ g\in G_d, d\in \mathbb{N}}
    \lim_{c\to \infty}
    \liminf_{t \to \infty}
    \sup_{\|\alpha\|_2\leq \frac{c}{\sqrt t}}
    \E_{s_{\alpha,g}^{\otimes t}}[L(\sqrt{t} (\widehat \NSW_t - \NSW^*_{\alpha,g} ))] } \;. 
\end{align*}

\begin{theorem} \label{thm:nsw_aym_risk}
    Let \cref{as:smo} hold.
    Then 
    \begin{align*}
        & \inf_{\hat \beta_t} \LAR_\b (\{ \hat \beta_t \})
        \geq 
        \E[L(\cN(0, \Sigma_\b))] \;,
        \\
        & \inf_{\hat u_t} \LAR_u ( \{ \hat u_t \})
        \geq 
        \E[L(\cN(0, \Sigma_u))]
        \;,
        \\
        & \inf_{\widehat \NSW_t} \LAR_\NSW ( \{ \widehat \NSW_t \} )
        \geq 
        \E[L(\cN(0, \sigma^2_\NSW))]
        \;,
    \end{align*}
    where $\sigma\sq_\NSW$ is defined in \cref{thm:normality}, and $\Sigma_\b, \Sigma_u$ in \cref{thm:clt_beta_u}.
    Proof in \cref{sec:proof:thm:nsw_aym_risk}
\end{theorem}
\subsection{Variance Estimation and Inference}\label{sec:lfm_inference}
In this section we show how to construct confidence intervals for Nash social welfare. 
We will show how to construct confidence intervals for pacing multipliers and utilities in the FPPE section. The procedure is similar for LFM, and thus we omit it here.

First, regarding inference, it is interesting to note that 
the observed NSW ($\LNSW^\gam$) is a negatively-biased estimate of the limit NSW ($\LNSW^*$), i.e., $ \E[\LNSW^\gam] - \LNSW^* \leq 0$.\footnote{Note 
$ \E[\LNSW^\gam] - \LNSW^* = \E[\min_\beta H_t(\beta)] - H(\betast) \leq \min_\beta \E[H_t(\beta)]  - H(\betast)  = 0.$}
Moreover, it can be shown that, when the items are i.i.d.\ $\E[\min H_t] \leq \E[\min H_{t+1}]$  by a simple argument from Proposition~16 from~\citet{shapiro2003monte}. Monotonicity tells us that increasing the market size produces, on average, less biased estimates of the limit NSW.

To construct confidence intervals for Nash social welfare, one needs to estimate the asymptotic variance. 
Let $\ptau$ be the price of item $\thetau$ in the finite market, and $\bar p = \frac1t \sumtau \ptau$.
The variance estimator is then
\begin{align}
    \hat\sigma_{\NSW}\sq \defeq \frac{1}{t} \sumtau\big( \ptau - \bar p\big)\sq 
    \; .
\end{align}
We emphasize that in the computation of the variance estimator $ \hat\sigma_{\NSW}\sq$ one does not need knowledge of the valuations $\{\vithetau\}_{i,\tau}$.
All that is needed is the equilibrium prices $\pgam = (p^{1},\dots, p^{t})$ for the items. 
Given the variance estimator, we 
construct the confidence interval 
$ \big[ \LNSW^\gam \pm z_{\alpha/2}  \frac{\hat \sigma_{\NSW}}{\sqrt t}\big]$,
where $z_\alpha$ is the $\alpha$-th quantile of a standard normal.
The next theorem establishes validity of the variance estimator.
\begin{theorem} \label{thm:ci_lnsw}
    It holds that $\hat \sigma ^2_{\NSW} \toprob \sigma_{\NSW}\sq$.
    Given $0< \alpha < 1$, it holds that
   $ 
    \lim_{t\to\infty} \P\big( \NSW^* \in [ \NSW^\gam \pm z_{\alpha/2}  \hat \sigma_{\NSW}/\sqrt t \,] \big) 
        =
        1-\alpha 
$.  Proof in Appendix~\ref{sec:proof_var_est}.

\end{theorem}

%% file: fppe_results.tex
\section{Statistical Results for FPPE}

Next we study statistical inference questions for the FPPE model.
Since FPPE is characterized by an EG-style program similar to that of LFM, many of the results for FPPE are similar to those for LFM. However, an important difference is that the FPPE model has constraints on the pacing multipliers, which makes the asymptotic theory more involved.
As for LFM, we assume that we observe a finite auction 
market $\oFPPE(b,v,1/t, \gamma)$ with $\gamma$ being $t$ i.i.d. draws from $s(\cdot)$, and we use it to estimate quantities from the limit market $\FPPE(b,v,s,\Theta)$.
In FPPE we mainly focus on the revenue of the limit market, and for the same reason as in LFM, since FPPE is characterized by an EG program with pacing multiplier as the variables, we also present results for $\b$.
Similarly to the LFM case, one can use results for $\b$ to derive estimators for buyer utilities.

\subsection{Basic Convergence Properties}
\label{sec:fppe_basic}
Since FPPE has a similar convex program characterization as LFM, strong consistency and convergence rate results can be derived using similar ideas.

\begin{theorem}\label{thm:fppe_as_convergence}
    We have $\betagam \toas \betast$,
    and $\REVgam \toas \REVst$. Proof in Appendix~\ref{sec:proof:thm:fppe_as_convergence}.
\end{theorem}

We complement the strong consistency result with the following rate results. They are derived using a discretization argument similar to the one for LFM.
\begin{theorem}
    \label{thm:rev_convergence}
    It holds that
    $ \| \betagam - \betast\|_2
    = \tilde O_p(\frac{\sqrt{n} (\vbar + 2\nubar n + 1)}{\ubar{b} \sqrt t})
    $
    and
    $ | \REV^\gam - \REV^* | = 
        \tilde{O}_p
        (\frac{\vbar \sqrt{n} (\vbar + 2\nubar n + 1 ) }{ \ubar{b} }\frac{1}{\sqrt{t}}).$
    Proof in Appendix~\ref{sec:proof:thm:rev_convergence}.
\end{theorem}

The above bounds hold for a broad class of limit FPPE models and may be loose for a particular model. 
In \cref{sec:fppe_asym}, we show that for buyers $i \in \{ i:\betasti = 1\}$, their pacing multipliers converge at a rate faster than $ \tilde O_p(1/\sqrt t)$.

\subsection{Asymptotics of FPPE}
\label{sec:fppe_asym}

As in LFM, our statistical inference results require the limit market to behave smoothly around the optimal pacing multipliers $\betast$.
To that end, we will make \cref{as:smo}, as in LFM.
Similar to LFM, under \cref{as:smo}, the equilibrium allocation $\xst$ is unique and must be pure. 
Again we can define 
\begin{align}
    \label{eq:def:must}
     \must (\t) = [x^*_1(\t) v_1(\t), \dots, x^*_n(\t) v_n (\t)]\tp\; .
\end{align}
Under \cref{as:smo}, the equilibrium allocation $\xst$ is unique, so $\must(\cdot)$ is also unique. 
Moreover, $\must(\t) = \nabla f(\t, \betast)$, 
and $ (\betast) \tp \must(\t) = \pst(\t)$ (the set of nondifferentiable points has measure zero, and thus we can ignore such points). 
Also let $\mubarst$
denote the utility from items:
\begin{align}
    \label{eq:def:mubarst}
    \mubarst = \int \must s \dt \in \Rnp
\end{align}
Note that in the FPPE model, buyers' utility consists of two parts: utility from items and leftover budgets.
In FPPE, the pacing multipliers relate budgets and utilities via \citep{conitzer2022pacing} 
\begin{align} \label{eq:beta_u_relation_fppe}
    \mubarsti + \deltasti = b_i/ \betasti \;.
\end{align}

In the unconstrained case, classical $M$-estimation theory says that, under regularity conditions, an $M$-estimator is asymptotically normal with covariance matrix $\cH\inv  \var(\text{gradient}) \cH\inv$ \citep[Chap.\ 5]{van2000asymptotic}. However, in the case of FPPE, which is characterized by a constrained convex problem, the Hessian matrix needs to be adjusted to take into account the geometry of the constraint set $B=(0,1]^n$ at the optimum $\betast$.
We let $\cP \defeq \Diag( \indi( \betasti < 1) )$ be an ``indicator matrix'' of buyers whose $\betasti < 1$, and define the projected Hessian \citep[Section 12.5]{nocedal2006numerical}
\begin{align}
    \label{def:cHB}
    \cH_B \defeq \cP \cH \cP \;.
\end{align}
It will be shown that the asymptotic variance of $\betagam$ is $\cH_B\pinv \var(\text{gradient}) \cH_B\pinv$ and the ``gradient" is exactly $\must$.

\begin{assumption}[\textsf{\scriptsize{SCS}}] 
    \label{as:scs} Strict complementary slackness holds:
     $\betasti = 1$ implies  $\delta^*_i  > 0$.
\end{assumption}

\cref{as:scs} can be viewed as a non-degeneracy condition from a convex programming perspective, since $\delta_i$ corresponds to a Lagrange multiplier on $\beta_i \leq 1$.
From a market perpective, \cref{as:scs} requires that if a buyer's bids are not paced ($\betasti = 1$), then their leftover budget $\delta^* _i$ must be strictly positive.
This can again be seen as a market-based non-degeneracy condition: if $\deltasti=0$ then the budget constraint of buyer $i$ is binding, yet $\betasti=1$ would imply that they have no use for additional budget.
If \cref{as:scs} fails, one could slightly increase the budgets of buyers for which \cref{as:scs} fails, i.e., those who do not pace yet have exactly zero leftover budget, and obtain a market instance with the same equilibrium, but where \cref{as:scs} holds.

From a technical viewpoint, \cref{as:scs} is a stronger form of first-order optimality.  
Note $\nabla H(\betast) = -\deltast$ (cf.\ \cref{lm:fppe_relation}). 
The usual first-order optimality condition is 
\begin{align}
    - \nabla H(\betast) \in \cN_B(\betast) \;,
    \label{eq:usual_foc}
\end{align}
where $\cN_B(\beta) = \prod_{i=1}^n J_i(\beta)$ is the normal cone with $J_i (\beta)= [0,\infty)$
if $\betai = 1$ and $J_i (\beta)= \{0\}$ if $\betai < 1$ for $\beta \in \Rnpp$.
Then \cref{eq:usual_foc} translates to the condition that $\betasti = 1$ implies $\deltasti \geq 0$. 
On the other hand, when written in terms of the normal cone, 
\cref{as:scs} is equivalent to $$- \nabla H(\betast) \in \relint ( \cN_B(\betast) ) \; ,
\footnote{
    The relative interior of a set is $\relint(S)\defeq\{x \in S:$ there exists $\epsilon>0$ such that $N_\epsilon(x) \cap \operatorname{aff}(S) \subseteq S \}$ where $\operatorname{aff}(S)$ is the affine hull of $S$, and $N_\epsilon(x)$ is a ball of radius $\epsilon$ centered on $x$. 
}$$
equivalent to the condition $\betasti = 1$ implies $\deltasti > 0$.
Given that $ \relint ( \cN_B(\betast) ) \subset \cN_B(\betast) $, \cref{as:scs} is obviously a stronger form of first-order condition. Conditions like \cref{as:scs} are commonly seen in the study of statistical properties of constrained $M$-estimators (\citealt[Assumption B]{duchi2021asymptotic} and \citealt{shapiro1989asymptotic}). In the proof of \cref{thm:clt}, \cref{as:scs} forces the critical cone to reduce to a hyperplane and thus ensures asymptotic normality of the estimates.
Without \cref{as:scs}, the asymptotic distribution of $\betagam$ could be non-normal.

\subsubsection{Asymptotic Normality}

We now show that the observed pacing multipliers $\betagam$ and the observed revenue $\REV^\gamma$ are asymptotically normal.
Define the \emph{influence functions} 
\begin{equation} \label{eq:def_inf_func}
    \begin{aligned}
    D_\beta(\theta) & \defeq - (\cH_B)^\dagger (\must(\theta) - \mubarst) \; ,
    \\
    D_\REV(\theta) & \defeq \pst(\theta) - \REV^* + (\mubarst) \tp D_\beta(\theta).
    \end{aligned}
\end{equation}
Recall $\must$ is defined in \cref{eq:def:must}, $\mubarst$ in \cref{eq:def:mubarst}, $\cH_B$ in \cref{def:cHB}.
And note $\E[D_\beta] = 0$ and $\E[D_\REV] = 0$.

\begin{theorem} 
    \label{thm:clt}
    If
    \cref{as:smo,as:scs} hold, then 
    \begin{align}
        &\sqrt{t}(\betagam - \betast) =  \frac1\roott \sumtau D_\beta (\thetau) + o_p(1) \; ,
        \\
        &\sqrt{t} (\REV^\gam - \REV^*) = \frac1\roott \sumtau D_\REV(\thetau) + o_p(1)  \;.
        \label{eq:rev_clt}
    \end{align}
    Consequently, $ \sqrt{t}(\betagam - \betast )$ and $\sqrt {t} (\REV^\gamma - \REV^*) $
    are asymptotically normal with means zero and variances $\Sigma_\beta\defeq \E[D_\beta D_\beta\tp ] =  (\cH_B)\pinv \var(\must) (\cH_B)\pinv$ and $\sigma^2_\REV \defeq \E[D_\REV (\theta)^2]$. 
    Proof in Appendix~\ref{sec:proof:thm:clt}.
\end{theorem}

The functions $D_\beta$ and $D_\REV$ are called the {influence functions} of the estimates $\betagam$ and $\REV^\gam$ because they measure the change in the estimates caused by adding a new item to the market (asymptotically).

\new{
\cref{thm:clt} can be simplified if $I_= = [k]$, i.e., the first $k$ buyers are the ones with $\beta_i=1$.
Let $\betast_<$ and $\betagam_<$ denote the subvectors corresponding to $I_<$, and define $\betast_=$ and $\betagam_=$ similarly.
Let
$\cH_{<}$ denote the lower right $(n-k)\times(n-k)$ block of $\cH $ corresponding to $I_<$.
\cref{thm:clt} then gives
$\sqrt t (\betagam_< - \betast_<) \tod \cN(0, \Sigma_{\beta, <})$ and $\sqrt t (\betagam_= - \betast_=) = o_p (1)$ for some positive semi-definite matrix $\Sigma_{\beta, <}$ \footnote{If $I_=$ (resp. $I_<$) is empty, we disregard the statement for $\betagam_=$ (resp. $\betagam_<$).}.
To see this, 
note the pseudo-inverse of the projected Hessian $(\cH_B )\pinv = \Diag(0_{k\times k} , (\cH_{<})\inv )$. 
Consequently, 
$ \sqrt{t} (\betagami - \betasti)$ is of order $o_p(1)$ for $i \in I_=$, and thus converging faster than the usual $O_p(1)$ rate.
In fact, one can show $\P(\betagami = 1) \to 1$ for all $i \in I_=$; see \cref{lm:fast_beta_convergence}.
The fast rate phenomenon is empirically investigated in \cref{sec:exp_fast_rate}.

A practical implication of \cref{thm:clt} is the identification of budget constrained buyers in the limit market.
By \cref{as:scs} we have $I_= = \{i: \betasti = 1\}= \{ i: \deltasti > 0\}$, i.e., $I_=$ is the set of buyers who are not budget constrained, and $\Ic = \{i:\betasti < 1\} = \{ i: \deltasti = 0\} $ is the set of buyers that exhaust their budget. \footnote{Without \cref{as:scs} we only have $ \{i:\betasti < 1\} \subset \{i:\deltasti = 0\}$ and $\{i:\deltasti > 0 \}\subset \{i:\betasti = 1\}$ by complementary slackness.}
The fast rate $o_p(t^{-1/2})$ implies that the platform can identify which buyers are budget constrained with high confidence.

\begin{corollary} \label{cor:set_I_estimate}
    Let 
    \cref{as:smo,as:scs} hold.
    Let $\hat I_= = \{i: \betagami \geq 1- \epsilon_t\}$ and $\hat I_< = [n] \setminus \hat I_=$ for some sequence $\epsilon_t$ such that $0\leq \epsilon_t = o(1)$ and $\sqrt t \epsilon_t \to  c \in (0,\infty ]$. Then 
    $ \P(\hat I_= = I_= \text{ and } \hat I_< = I_<) \to 1$.
\end{corollary}

Finally, \cref{thm:clt} implies fast revenue convergence if $I_= = \emptyset$.
In this case, the influence function $D_\REV(\t) = 0$ for all $\t$ because $\cH \betast = \mubarst$ if $I_= = \emptyset$ (\cref{lm:fppe_relation}), and $(\betast)\tp \must(\t) = \pst(\t)$. 
Then \cref{eq:rev_clt} gives
$\sqrt t (\REVgam - \REVst) = o_p(1)$.
Intuitively, 
if $\betagami < 1$ for all $i$, then all buyers' budgets are exhausted in the observed FPPE, and so we have
$
    \REVgam = \sumiton b_i$. 
By the convergence $\betagam \toprob \betast$ and that $\betasti < 1$ for all $i$, we know that for large $t$ with high probability, $\betagami <1 $ for all $i$, and thus
$\REVgam =\sumiton b_i =\REVst$.
In that case, it must be that the asymptotic variance of revenue equals zero.
}

In an FPPE, individual utilities and Nash social welfare can be similarly defined. 
By applying the delta method and \cref{thm:clt} we can derive asymptotic distributions for these quantities, using the fact that they are smooth functions of $\beta$. See Appendix~\ref{sec:full:cor:clt_u_and_nsw} for more details.


\subsubsection{Asymptotic Local Minimax Optimality}

Given the asymptotic distributions for $\beta$ and $\REV$, we will show that the observed FPPE estimates are optimal in an asymptotic local minimax sense.
Recall in \cref{sec:perturb_supply} we have defined the perturbed supply family $G_d$ of dimension $d$ with perturbation $(\alpha,g)$.

\paragraph{Asymptotic local minimax optimality for $\beta$.}
We first focus on estimation of pacing multipliers.\
For a given perturbation ${(\alpha,g)}$, we let $\betast_{\alpha,g}$, $\pst_{\alpha,g}$ and $\REVst_{\alpha,g}$ be the limit FPPE pacing multiplier, price and revenue under supply distribution $s_{\alpha,g}$.
Clearly $\betast = \betast_{0,g}$ for any $g$ and similarly for $\pst_{\alpha,g}$ and $\REV_{\alpha, g}$.
Let $L:\Rn \to \R$ be any symmetric quasi-convex loss. 
\footnote{A function is quasi-convex if its sublevel sets are convex.}
In asymptotic local minimax theory we are interested in the local asymptotic risk: given a sequence of estimators $\{ \hat \beta_t: \Theta^t \to \Rn \}_t$, 

\begin{equation*}
\LAR_\beta (\{\hat \beta_t\}_t )  \defeq 
    {\sup_{ g\in G_d, d\in \mathbb{N}}
    \lim_{c\to \infty}
    \liminf_{t \to \infty}
    \sup_{\|\alpha\|_2\leq \frac{c}{\sqrt t}}
    \E_{s_{\alpha,g}^{\otimes t}}[L(\sqrt{t} (\hat \beta_t - \betast_{\alpha,g} ))] } \;.
\end{equation*}
As an immediate application of Theorem 1 from \citet{duchi2021asymptotic}, it holds that 
\begin{align*}
    \inf_{ \{\hat \beta_t\}_t }  \LAR_\beta (\{\hat \beta_t\}_t )  \geq \E[L(\cN( 0, (\cH_B)\pinv \var(\must) (\cH_B)\pinv))] \; . 
\end{align*}
where the expectation is taken w.r.t.\ a normal specified above.
Moreover, the lower bound is achieved by the observed FPPE pacing multipliers $\betagam$ according to the normality result in \cref{thm:clt}.

\paragraph{Asymptotic local minimax optimality for revenue estimation.}
For pacing multipliers, the result is a direct application of the perturbation result from \citet{duchi2021asymptotic}. The result for revenue estimation is more involved.
Given a symmetric quasi-convex loss $L:\R\to\R$, we define the local asymptotic risk for any procedure $\{\hat r _t : \Theta^t \to \R\}$ that aims to estimate the revenue: 
\begin{align*} \hspace{-.2cm}
\LAR_\REV ( \{\hat r_t \}) \defeq
    \sup_{ g\in G_d, d\in \mathbb{N}}
    \lim_{c\to \infty}
    \liminf_{t \to \infty}
    \sup_{\|\alpha\|_2\leq \frac{c}{\sqrt t}}
    \E_{s_{\alpha,g} ^{\otimes t}}[L(\sqrt{t} (\hat r_t  - \REVst_{\alpha,g} ))] \;. 
\end{align*} 
\begin{theorem}[Asymptotic local minimaxity for revenue]
    \label{thm:rev_localopt}
   If \cref{as:smo,as:scs} hold, then 
       $\inf_{ \{\hat r_t\}} \LAR_\REV ( \{\hat r_t \}) \geq \E[L(\cN(0, \sigma\sq_\REV))] \;.
$
\end{theorem}

Proof in Appendix~\ref{sec:proof:rev_local_asym_risk}. In the proof we calculate the derivative 
of revenue w.r.t.\ $\alpha$, which in turn uses 
a perturbation result for constrained convex programs from \citet{shapiro1989asymptotic}.
Again, the lower bound is achieved by the observed FPPE revenue $\REVgam$ according to the normality result in \cref{thm:clt}.
Similar optimality statements can be made for $u$ and $\NSW$ by finding the corresponding derivative expressions.

\subsection{Variance Estimation and Inference}
\label{sec:fppe_inference}

To perform inference on $\betast$, we construct estimators for the influence functions \cref{eq:def_inf_func}. In turn, this requires estimators for 
the projected Hessian $\cH_B$ (\cref{def:cHB}) and the variance of the utility map $\must$ (\cref{eq:def:must}).
First, given a sequence of smoothing parameters $\varepsilon_{\cP,t} = o(1)$ and $\varepsilon_{\cP,t} \sqrt t \to c \in (0,\infty]$, we estimate the projection matrix $\cP$ by 
$
    \hat \cP \defeq \Diag(\indi(\betagami < 1 - \varepsilon_{\cP,t} )) \;.
$
For another sequence $\varepsilon_{\cH,t} = o(1)$ with $\varepsilon_{\cH,t} \sqrt t \to \infty$, we introduce a numerical difference estimator $\hat \cH$ for the Hessian matrix $\cH$, whose $(i,j)$-th entry is
\begin{align}
    \label{eq:def:hessian_estimate}
    \hat{\cH}_{i j}\defeq[H_t(\betagam_{++})-H_t(\betagam_{+-})
     -H_t(\betagam_{-+}) 
    +H_t(\betagam_{--})] / (4 \varepsilon_{\cH,t}^2 )
\end{align}  
with 
$\betagam_{\pm \pm} \defeq  \betagam \pm e_i \varepsilon_{\cH,t} \pm e_j \varepsilon_{\cH,t}$,
and $H_t$ is defined in \cref{eq:def_pop_eg}. Finally, $\hat \cH_B = \hat \cP \hat \cH \hat \cP$ is the estimator of $\cH_B$. 
Next, recall $\xgam = (\xitau)_{i,\tau}$ and $(\vitau)_{i,\tau}$ are the allocation and values in the finite FPPE.
Mimicking $\must$ in \cref{eq:def:must}, $\mubarst$ in \cref{eq:def:mubarst} and $\cov(\must)$, define the finite sample analogues
\begin{align}
   \mu^\tau = [x_1^{\tau} v_1^\tau, \dots, x_n^{\tau} v_n^\tau]\tp
   , \quad 
   \bar \mu= \frac1t \sumtau \mutau
   , \quad  
   \hat \Omega = \frac1t \sumtau (\mutau - \mubar) (\mutau - \mubar)\tp 
   \;.
\end{align}

With all the new notations, we define the estimators of influence functions in \cref{eq:def_inf_func} 
\begin{align*}
    \hat D_\beta^\tau  \defeq - (\hat \cH_B)\pinv (\mutau- \mubar) \;,
    \quad
    \hat D _\REV ^\tau \defeq \ptau - \REV^\gamma + (\mubar) \tp\hat D_\beta^\tau \;.
\end{align*}
Given that the asymptotic variances of $\betagam$ and $\REVgam$ are 
$\E[D_\beta D_\beta \tp]$ and $\E[D_\REV\sq]$, respectively, plug-in estimators for the (co)variance are naturally
\begin{align}
    \hat \Sigma_\beta \defeq \frac{1}{t} \sumtau \hat D_\beta^\tau  (\hat D_\beta^\tau )\tp = (\hat \cH_B)\pinv  \hat \Omega (\hat \cH_B)\pinv
    \; , \;\;
    \hat \sigma \sq_\REV \defeq \frac{1}{t} \sumtau (\hat D_\REV ^\tau)\sq \; .
    \label{eq:plugin variance}
\end{align}

\begin{theorem}
    \label{thm:variance_estimation}
    Let the conditions of \cref{thm:clt} 
    and the required rate conditions on $\varepsilon_{\cP,t}, \varepsilon_{\cH,t}$ hold. 
    Then $\hat \Sigma_\beta \toprob \Sigma_\beta$ and $\hat \sigma_\REV \sq \toprob \sigma_\REV \sq$. 
    Proof in Appendix~\ref{sec:proof:thm:variance_estimation}.
\end{theorem}

For any $\alpha \in (0,1)$, the $(1-\alpha)$-confidence regions for $\betast$ and $\REVst$ are 
\begin{align}
    \CR_\beta\defeq \betagam + (\chi_{n, \alpha} / \sqrt{t})\hat \Sigma_\beta^{1/2} \B \; ,\;\;
    \CI_\REV \defeq[\;\REVgam\pm z_{\alpha/2}  \hat \sigma_{\REV}/\sqrt t \;] \; .
   \label{eq:rev_variance}
\end{align} 
where $\chi_{n,\alpha}$ is the $(1-\alpha)$-th quantile of a chi-square distribution with degree $n$, $\B$ is the unit ball in $\Rn$, and $z_{\alpha/2}$ is the $(1-\frac{\alpha}{2})$-th quantile of a standard normal distribution.
The coverage rate of $\CI_\REV$ is empirically verified in \cref{sec:exp_rev_coverage}.

The rate condition on $\varepsilon_{\cH,t}$ suggests choosing $\varepsilon_{\cH,t} = t^{- d}$ for $0 < d < \frac12$. \cref{sec:exp_hessain} studies how the choice of $d$ affects the Hessian estimation numerically.

\subsubsection{Hessian-Free Variance Estimation under a Bid Gap Condition}
\new{
We now show that the Hessian estimation can be avoided if there is sufficient separation between the highest and second-highest bids. 
Define the gap between the highest and the second-highest bids under pacing multiplier $\beta$ by
\begin{align}
    \label{eq:def:bidgap}
    \bidgap(\beta,\theta) \defeq \max \{\betai \vithe  \} - \operatorname{secondmax}\{ \betai \vithe \} \;,\end{align} 
here $\operatorname{secondmax}$ is the second-highest entry; e.g., $\operatornamewithlimits{secondmax}([1,1,2]) = 1$.
The condition will require integrability of the inverse of the bid gap, i.e., $\E[\bidgap(\beta,\theta)\inv] < \infty$.
We introduce a new price function as follows.
If $I_= \neq \emptyset$,
we define 
\begin{align*}
    \tilde p^*(\t) =  p^*(\t) \indi\{ \max_{i \in I_=} \betasti \vithe = \max_{i \in [n]} \betasti \vithe \}.
\end{align*}
If $I_= = \emptyset$, then we define $\tilde p ^*(\t) = 0$. 
The new price function $\tilde p \st$ preserves the price of an item $\t$ if it is won by a buyer in $I_=$, and $\tilde p \st$ sets the price to $0$ otherwise.
For the simplified revenue variance estimator we need an estimator of $I_=$, which naturally is $\hat I_= = \{i: \betagami \geq  1 - \varepsilon_{\cP, t} \}$. 
Recall revenue $\REVgam = \frac1t \sumtau \ptau$ and price $\pgam = [p^1, \dots, p ^ t ] \tp $. 
Define $\tilde p ^\tau = \ptau \indi \{ \max_{i \in \hat I_=} \betagami \vitau = \max_{i \in [n]} \betagami \vitau \} $, $\bar {\tilde p} = \frac1t \sumtau \tilde p ^\tau$. 
Define the simplified variance estimators (neither of which requires Hessian estimation):
\begin{align}
    \label{eq:simplified variance estimator}
  \hat \sigma\sq_{\REV, \textup{sim}}= \frac1t \sumtau (\tilde p^\tau - \bar {\tilde p})\sq\;,\;
  \hat \Sigma _{\beta, \textup{sim}} = \hat \cP \Diag((\betagami)\sq b_i \inv) \hat \Omega \Diag((\betagami) \sq b_i \inv) \hat \cP 
  \;.
\end{align}

\begin{theorem}[Hessian-Free Inference]
    \label{thm:simplified_rev_var}

    Let the conditions of \cref{thm:clt} 
    and the required rate condition on $\varepsilon_{\cP,t}$ hold. 
    If additionally $\E[\bidgap (\betast, \t)\inv] < \infty$, then 
    $\sigma\sq_\REV = \var (\tilde p \st)$ and 
    $\Sigma_\beta = \cP \Diag((\betasti)\sq b_i \inv) \cov(\must) \Diag((\betasti)\sq b_i\inv) \cP$.
    Moreover, the Hessian-free variance estimators are consistent, $\hat \Sigma _{\beta, \textup{sim}} \toprob \Sigma_\beta$ and $\hat \sigma\sq_{\REV, \textup{sim}} \toprob \sigma\sq_\REV$.
    Proof in \cref{sec:proof:simplified_var_est}.
\end{theorem}
}

\subsection{Application: A/B Testing in First-Price Auction Platforms}
\label{sec:ab testing}
Consider an auction market with $n$ buyers with a continuum of items $\Theta$ with supply function $s$.
Now suppose that we are interested in the effect of deploying some new technology (e.g. new machine learning models for estimating click-through rates in the ad auction setting).
To model treatment application we introduce the \emph{potential value functions}
$$ \hspace{-.2cm} v(0) \defeq (v_1(0,\cdot),\dots, v_n(0,\cdot)),\;v(1)\defeq (v_1(1,\cdot),\dots, v_n(1,\cdot)).$$
If item $\theta$ is exposed to treatment $w \in \{0,1\}$, then its value to buyer $i$ will be $v_i(w,\theta)$. 

Suppose we are interested in estimating the change in the auction market when treatment 1 is deployed to the entire item set $\Theta$. 
In this section we describe how to do this using A/B testing, specifically for estimating the treatment effect on revenue. 
Formally, we wish to look at the difference in revenues between the markets 
$$
    \FPPE(b, v(0), s) \text{ and }  \FPPE(b,v(1), s),
$$
where $\FPPE(b, v(0), s)$ is the market with treatment 1, and $\FPPE(b,v(1), s)$ is the one with treatment 0. 
The treatment effects on revenue is defined as $$\tau_{\REV} \defeq \REVst(1) - \REVst(0)\;,$$
where $\REVst(w)$ is revenue in the equilibrium $\FPPE(b,v(w), s)$.

\new{
The A/B test framework we put forward above is rather general.
This formalism is able to model treatments which ultimately affect the market via shifting the value distribution. We list several examples below.
\begin{itemize}
    \item 
    User interface (UI) changes (on the item side).
    Adjustments to ad aesthetics (e.g., font styles, ad placement, or link position) and user-side UI (e.g., button sizes, creating a holiday-themed user interface, running promotions during a major sports event) can be seen as modifying the experience that every user impression—our ``item''—undergoes. They affect each user's interaction with the content, thus shifting the value distribution of those items.
    \item Platform algorithm improvements (internal to the platform).
    Refinements in the value prediction algorithms, or the introduction of new ranking methods, and implementing ads quality filters all shift the distribution of item values.
    \item Addition (or deletion) of buyers. Say we add a buyer in B. We simply create a dummy buyer in A whose value function is zero for all items to represent that new buyer.
    \end{itemize}
}

We will refer to the experiment design as \emph{budget splitting with item randomization}.
The design works in two steps, and closely mirrors how A/B testing is conducted at large tech companies.

\textbf{Step 1. Budget splitting.}
    We create two markets, and every buyer is replicated in each market.
    For each buyer $i$ we allocate $\pi b_i$ of their budget to the market with treatment $w=1$, and the remaining budget, $(1-\pi)b_i$, to the market with treatment $w=0$.
    Each buyer's budget is managed separately in each market.

\textbf{Step 2. Item randomization.}
Let $(\theta^1,\theta^2,\dots)$ be i.i.d.\ draws from the supply distribution $s$.
For each sampled item, it is applied treatment $1$ with probability $\pi$ and treatment $0$ with probability $1-\pi$.  
The total A/B testing horizon is $t$. When the end of the horizon is reached, two observed FPPEs are formed. Each item has a supply of $\pi/t_1$ in the 1-treated market and $(1-\pi)/t_0$ in the 0-treated market. The $1/t_1$ is the scaling required for our CLTs and the $\pi$ factor ensures the budget-supply ratio agrees with the limit market; 
due to FPPE scale-invariance, we could equivalently rescale budgets.

Let $t_0$ be the number of $0$-treated items, and $t_1$ be the number of $1$-treated items. Conditional on the total number of items $t=t_1 + t_0$, the random variable $t_1$ is a binomial random variable with mean $\pi t$. Let $\gamma(0) = (\theta^{1,1},\dots, \theta^{1,t_1})$ be the set of $0$-treated items, and similarly $\gamma(1) = (\theta^{0,1},\dots, \theta^{0,t_0})$.
The total item set $\gamma = \gamma(0) \cup \gamma(1)$. 
The observables in the described A/B testing experiment are 
$$ \hspace{-.1cm}
\oFPPE \big(\pi b, v(1), \tfrac{\pi}{t_1}, \gamma(1)  \big),\; 
    \oFPPE \big( (1-\pi ) b,  v(0), \tfrac{1-\pi}{t_0}, \gamma(0)\big ), $$
both defined in \cref{def:finite_fppe}.
Let $\REVgam(w)$ denote the observed revenue in the $w$-treated market.
The estimator of the treatment effect on revenue is 
$$
     \hat \tau_{\REV} \defeq \REVgam(1) - \REVgam(0).$$

For fixed $(b, s)$, the variance $\sigma\sq_\REV$ in \cref{thm:clt} is a functional of the value functions. We will use $\sigma\sq_\REV(w)$ to represent the revenue variance in the equilibrium $\FPPE(b,v(w), s)$.
Each variance can be estimated using \cref{eq:plugin variance}.

\begin{theorem}[Revenue treatment effects asymptotic normality]
    \label{thm:clt_ab_testing}
    Suppose \cref{as:smo} and \cref{as:scs} hold in the limit markets $\FPPE(b,v(1),s)$ and $\FPPE(b,v(0),s)$. Then 
 $\sqrt{t} (\hat \tau_\REV - \tau_\REV) \tod\cN\allowbreak \big(0, \frac{ \sigma\sq_\REV(1)}{\pi} \allowbreak + \allowbreak\frac{ \sigma\sq_\REV(0)}{(1-\pi)} \big) .$
Proof in Appendix~\ref{sec:proof:thm:clt_ab_testing}.
\end{theorem}

Based on the theorem, an A/B testing procedure is the following.
Compute the revenue variance as \cref{eq:plugin variance} for each market, obtaining $ \hat \sigma\sq_\REV(1)$ and $ \hat \sigma\sq_\REV(0)$, and form the confidence interval
\begin{align}
    \hat \tau_\REV \pm z_{\alpha/2} \bigg( \frac{ \hat \sigma\sq_\REV(1)}{\pi} \allowbreak + \allowbreak\frac{ \hat \sigma\sq_\REV(0)}{(1-\pi)}  \bigg).
    \label{eq:ab_rev_ci}
\end{align}
If zero is on the left (resp.\ right) of the CI, we conclude that the new feature increases (resp.\ decreases) revenue with $(1-\alpha)\times 100\%$ confidence. 
If zero is in the interval, the effect is undecided.
See \cref{sec:experiment_abtest} for a semi-synthetic study verifying the validity of this procedure.

%% file: tbl_rev_coverage_short.tex
\begin{table}[ht]
\centering
\small
\caption{Coverage of revenue CI. $\alpha$ = proportion of $\beta_i = 1$, $d$ is the exponent in finite difference stepsize $\epsilon_t = t^{-d}$. \new{Numbers in parentheses represent
the lengths of CIs. Nominal coverage rate is 90\%.}}
\label{tbl:rev_coverage_short}
\begin{tabular}{llllllllllllll}
\toprule
         & buyers & \multicolumn{4}{l}{20} & \multicolumn{4}{l}{50} & \multicolumn{4}{l}{80} \\
         & $\alpha$ &                     0.05 &                     0.10 &                     0.20 &                     0.30 &                     0.05 &                     0.10 &                     0.20 &                     0.30 &                     0.05 &                     0.10 &                     0.20 &                     0.30 \\
\midrule $d$ & items &                          &                          &                          &                          &                          &                          &                          &                          &                          &                          &                          &                          \\
\midrule
\multirow{4}{*}{0.40} & 100 &  \makecell{0.79\\(1.72)} &  \makecell{0.79\\(1.82)} &  \makecell{0.93\\(1.87)} &   \makecell{0.9\\(1.81)} &  \makecell{0.87\\(1.84)} &  \makecell{0.81\\(1.91)} &  \makecell{0.89\\(1.82)} &  \makecell{0.88\\(2.00)} &  \makecell{0.81\\(1.89)} &  \makecell{0.89\\(1.89)} &  \makecell{0.97\\(1.97)} &   \makecell{0.9\\(1.95)} \\
         & 200 &  \makecell{0.88\\(1.33)} &  \makecell{0.88\\(1.36)} &  \makecell{0.87\\(1.35)} &   \makecell{0.9\\(1.34)} &  \makecell{0.88\\(1.32)} &  \makecell{0.93\\(1.36)} &  \makecell{0.89\\(1.37)} &  \makecell{0.94\\(1.40)} &  \makecell{0.87\\(1.37)} &  \makecell{0.93\\(1.37)} &  \makecell{0.88\\(1.42)} &  \makecell{0.86\\(1.42)} \\
         & 400 &  \makecell{0.84\\(0.93)} &  \makecell{0.88\\(0.99)} &  \makecell{0.93\\(0.99)} &  \makecell{0.91\\(0.98)} &   \makecell{0.9\\(0.95)} &  \makecell{0.94\\(0.98)} &  \makecell{0.92\\(0.98)} &  \makecell{0.84\\(1.00)} &  \makecell{0.88\\(0.97)} &  \makecell{0.85\\(0.98)} &  \makecell{0.86\\(1.01)} &  \makecell{0.85\\(1.01)} \\
         & 600 &  \makecell{0.89\\(0.76)} &  \makecell{0.88\\(0.79)} &   \makecell{0.9\\(0.81)} &  \makecell{0.89\\(0.80)} &   \makecell{0.8\\(0.77)} &  \makecell{0.87\\(0.80)} &  \makecell{0.81\\(0.80)} &  \makecell{0.92\\(0.83)} &  \makecell{0.86\\(0.80)} &  \makecell{0.83\\(0.80)} &  \makecell{0.97\\(0.83)} &  \makecell{0.89\\(0.83)} \\
\bottomrule
\end{tabular}
\end{table}

%% file: tbl_rev_coverage_short_comparison.tex


    
\begin{table}[ht!]
    \caption{
        \new{
            Coverage comparison between the Hessian-free CI in \cref{eq:simplified variance estimator}, the naive CI, and Hessian-based CI in \cref{eq:plugin variance}. $\alpha$ = proportion of $\beta_i = 1$. In each cell we present the coverage rate and the average CI widths in parentheses. Nominal rate = 90\%. The number of repetitions in each cell = 500.
        }
    }
    \label{tbl:rev_coverage_new_vs_naive_short}
    \centering
    
    \begin{tabular}{llll}
    \toprule
     & items & 100 & 800 \\
    buyers & $\alpha$ &  &  \\
    \midrule
    \multirow[t]{6}{*}{80} & 0.0 & \makecell{{0.88(0.04)}| 1.00(0.75)| 1.00(0.72)} & \makecell{{1.00(0.00)}| 1.00(0.28)| 1.00(0.27)} \\
     & 0.05 & \makecell{{0.93(1.05)}| 0.90(0.95)| 0.90(0.94)} & \makecell{{0.94(0.38)}| 0.90(0.34)| 0.90(0.34)} \\
     & 0.1 & \makecell{{0.95(1.15)}| 0.90(0.96)| 0.90(0.94)} & \makecell{{0.95(0.39)}| 0.90(0.34)| 0.90(0.34)} \\
     & 0.2 & \makecell{{0.92(1.28)}| 0.90(1.14)| 0.89(1.12)} & \makecell{{0.92(0.43)}| 0.91(0.41)| 0.91(0.41)} \\
     & 0.3 & \makecell{{0.90(1.27)}| 0.89(1.19)| 0.89(1.17)} & \makecell{{0.89(0.43)}| 0.88(0.43)| 0.88(0.43)} \\
     & 1.0 & \makecell{{0.90(1.33)}| 0.88(1.27)| 0.88(1.25)} & \makecell{{0.87(0.45)}| 0.87(0.45)| 0.87(0.45)} \\
    \cline{1-4}
    
    \end{tabular}
    \end{table}

%% file: tbl_abtest_short.tex
\begin{table}[ht]
\centering
\caption{Coverage of treatment effect. $\pi$ = treatment probability, the finite difference stepsize $\epsilon_t = t^{-0.4}$, proportion of unpaced buyers $\beta_i=1$ is 30\%.
\new{Numbers in parentheses represent
the lengths of CIs. Nominal coverage rate is 90\%.}}
\label{tbl:abtest_short}
\begin{tabular}{llllll}
\toprule
   & items &                      100 &                      200 &                      400 &                      600 \\
\midrule buyers & $\pi$ &                          &                          &                          &                          \\
\midrule
\multirow{5}{*}{50} & 0.1 &  \makecell{0.88\\(8.45)} &  \makecell{0.88\\(4.75)} &  \makecell{0.92\\(1.87)} &  \makecell{0.88\\(1.54)} \\
   & 0.3 &  \makecell{0.95\\(3.49)} &  \makecell{0.96\\(2.37)} &  \makecell{0.86\\(1.28)} &  \makecell{0.93\\(1.03)} \\
   & 0.5 &  \makecell{0.92\\(3.76)} &  \makecell{0.95\\(5.05)} &   \makecell{0.9\\(1.26)} &  \makecell{0.94\\(0.98)} \\
   & 0.7 &  \makecell{0.85\\(3.10)} &  \makecell{0.95\\(2.27)} &  \makecell{0.98\\(1.85)} &  \makecell{0.95\\(1.28)} \\
   & 0.9 &  \makecell{0.76\\(3.36)} &  \makecell{0.87\\(2.87)} &  \makecell{0.92\\(2.78)} &  \makecell{0.96\\(9.33)} \\
\bottomrule
\end{tabular}
\end{table}

%% file: fisher_market_app.tex
\ECHead{Proofs for Analytical Properties of the Dual Objective}

\section{Analytical Properties of the Dual Objective}
\label{sec:analytical_formal}

\subsection{Formal Statements}
Based on our differentiability characterization, it is natural to search for a stronger form of \cref{eq:as:notie} and hope that such a refinement could lead to second-order differentiability. 
\new{\cref{thm:second_order_informal} provided three sufficient conditions for second-order differentiability.
Condition (i) gives two such refinements of \cref{eq:as:notie}. }
Condition (ii) is motivated by the idea that the expectation operator tends to produce smooth functions. 
The exact smoothness requirement is presented in the appendix, which we show is easy to verify for 
several common distributions.
Finally, Condition (iii) considers the linear-valuations setting of~\citet{gao2022infinite}, where the authors provide tractable convex programs for computing the infinite-dimensional equilibrium. 
Here we give another interesting properties of this setup by showing that the dual objective is $C^2$. Generalization to piecewise linear value functions will also be discussed at the end.

The tied buyers for an item will be useful for later discussions.
Let $I(\beta, \theta) = \argmax_i{\beta_i v_i(\theta)}$ be the set of maximizing indices, which could be non-unique. 
We say there is \emph{no tie for item $\theta$ at $\beta$} if $I(\beta,\theta)$ is single-valued, in which case we use $i(\beta,\theta)$ to denote the unique maximizing index. 
Moreover, by Theorem 3.50 from~\citet{beck2017first}, the subgradient $\partial_\beta f(\theta, \b)$ is the convex hull of the set $\{v_ie_i, i\in I(\beta,\theta) \}$. When $I(\beta,\theta)$ is single-valued, the subgradient set is a singleton, and thus $f$ is differentiable.

\subsubsection*{Markets with sufficient bid gap}

A natural idea is to search for a stronger form of \cref{eq:as:notie} and hope that such a refinement could lead to second-order differentiability. 
In particular, this section is concerned with statement (i) of \cref{thm:second_order_informal}. First we show the condition based on the expectation.

\begin{theorem}
    \label{thm:int_implies_hessian}
    Suppose $H$ is differentiable in a neighborhood of $\betast$
    and 
    \begin{align}
        \E \bigg[\frac{1}{\bidgap(\betast,\theta)^{}}\bigg] = \int_\Theta \frac{1}{\bidgap(\betast,\theta)^{}} s(\theta) \diff \t < \infty
        \;,
        \tag{\small\textsc{INT}}
        \label{eq:as:UI}
    \end{align}
    then $H$ is twice differentiable at $\betast$. Furthermore, it holds $\nabla\sq \fbar (\betast) = 0 $ and $\nabla \sq H(\betast) =   \Diag(b_i/(\betasti)\sq)$.  

    Proof in \cref{proof:sec:analytical_properties_of_dual_obj}.
\end{theorem}


We compare the integrability condition in the above theorem with \cref{eq:as:notie}.
Both \cref{eq:as:UI} and \cref{eq:as:notie} can be interpreted as a form of robustness of the market equilibrium. The quantity $\bidgap(\beta,\theta)$ measures the advantage the winner of item $\theta$ has over other losing bidders. The larger $\bidgap(\beta,\theta)$ is, the more slack there is in terms of perturbing the pacing multiplier before affecting the allocation at $\theta$.
In contrast to \cref{eq:as:notie} which only imposes an item-wise requirement on the winning margin,
the above assumption requires the margin exists in a stronger sense. Concretely, such a moment condition on the margin function $\bidgap$ represents a balance between how small the margin could be and the size of item sets for which there is a small winning margin.

Second we consider the condition based on the essential supremum.
For any buyer $i$ and her winning set $\Theta_i^*$, there exists a positive constant $\eps_i > 0$ such that
        \begin{align} 
            \label{eq:as:win_margin}
            \betasti  v_i(\theta) \geq \max_{k\neq i} \beta^*_k v_k(\theta) + \eps_i \, , \forall \theta \in \Theta_i^*
            \tag{\small{GAP}} 
            \quad \Leftrightarrow \quad
            \esssup_{\theta \in\Theta} 1/\bidgap(\beta,\theta) < K < \infty  
        \end{align} 
It requires that the buyer wins the items without tying bids uniformly over the winning item set. 
The existence of a constant $K<\infty$ such that $1/\bidgap(\beta,\theta) < K$ for almost all items makes a stronger requirement than \cref{eq:as:UI}.
From a practical perspective, it is also evidently a very strong assumption: for example, it won't occur with many natural continuous valuation functions. Instead, the condition requires the valuation functions to be discontinuous at the points in $\Theta$ where the allocation changes.
Empirically, since $\betagam$ is a good approximation of $\betast$ for a market of sufficiently large size, 
\cref{eq:as:win_margin} can be approximately verified by replacing $\betast$ with $\betagam$.
As a trade-off, \cref{eq:as:UI} is a weaker condition than \cref{eq:as:win_margin} but is harder to verify in practical application.

Below we present two examples where \cref{eq:as:UI} holds.
\begin{example}[Discrete Values]
    Suppose the values are supported on a discrete set, i.e., $[v_1,\dots, v_n]\in  \{V_1,\dots, V_K  \} \subset \Rn$ a.s.\ 
    Suppose there is no tie for each item at $\betast$. Then \cref{eq:as:win_margin} and thus \cref{eq:as:UI} hold.
\end{example}

\begin{example}[Continuous Values]
    \label{eg:cont_value}
    Here we give a numeric example of market with two buyers where \cref{eq:as:UI} holds.
    Suppose the values are uniformly distributed over the sets $\{v \geq 0: v_2 \leq 1, v_2 \geq 2v_1 \}$ and $ \{v \geq 0: v_1 \leq 1, v_2 \leq  \frac12 v_1 \}$.
    See \cref{fig:cont_value} for an illustration. 
    By calculus, we can show the map $\fbar(\beta) = \E[\max\{v_1\beta_1, v_2 \beta_2 \}]$ is 
    \begin{align*}
       \fbar (\beta)
       = 
       \begin{cases}
        \big(\frac{5}{12} - \frac13 \frac{\beta_1}{\beta_2}\big) \beta_1 + \frac{2\beta_2}{3 \beta_1} \beta_2
        &\text{if $ \beta_2 \geq 2\beta_1$}
        \\
        \frac13 (\beta_1 + \beta_2) 
        &\text{if $\tfrac12  \beta_1 < \beta_2 < 2\beta_1$}
        \\
        \big(\frac{5}{12} - \frac13 \frac{ \beta_2}{\beta_1}\big) \beta_2 + \frac{2\beta_1 }{3 \beta_2}\beta_1 
        &\text{if $\beta_2 \leq \frac12 \beta_1$}
       \end{cases}
       \;.
    \end{align*}
   
    Next, we derive the Hessian of $\fbar$. 
    On $\{\tfrac12  \beta_1 < \beta_2 < 2\beta_1 \}$ we have $\nabla \sq \fbar = 0$.
    In the region $\{ \beta_2 > 2\beta_1\}$, the Hessian is
    \begin{align*}
        \nabla\sq \fbar(\beta)= \begin{bmatrix} \frac{2\,{\beta_{2}}^2}{3\,{\beta_{1}}^3} & -\frac{2\,\beta_{2}}{3\,{\beta_{1}}^2}\\ -\frac{2\,\beta_{2}}{3\,{\beta_{1}}^2} & \frac{2}{3\,\beta_{1}} \end{bmatrix} 
        \;.
    \end{align*}
    The Hessian on the region $\{ \beta_2 < \frac12 \beta_1 \}$ has a completely symmetric expression by switching $\beta_1$ and $\beta_2$.
    The Hessian can also be derived using formulas in \cref{sec:closedformhessian}.
    From here we can see the function $\fbar$ is $C^2$ except on the lines $\beta_2=2\beta_1$ and $\beta_2 = \beta_1/2$. 

    Consider $\b$ on $B = \{ \beta > 0: \frac12 \beta_1 < \beta_2 < 2 \beta_1 \}$. \cref{eq:as:UI} holds.
    And that $\nabla\sq \fbar (\b) = 0 $ on $B$, which agrees with \cref{thm:int_implies_hessian}.

    Consider $\beta$ in the interior of region $\{\beta > 0: \beta_2 > 2 \beta_1\}$. The function $\fbar$ is twice continuously differentiable but \cref{eq:as:UI} does not hold.

    Consider $\beta$ on the ray $\{\b > 0: \beta_2 = 2 \beta_1\}$. For these $\b$'s the set $\{\t: \bidgap(\b,\t) = 0 \}$ is measure zero and yet $\nabla \sq \fbar (\b)$ is not twice differentiable at these $\b$'s.
    This implies \cref{eq:as:notie} does not necessarily imply twice differentiability.
\end{example}
\begin{figure}[ht]
    \centering
    \includegraphics[scale=.4]{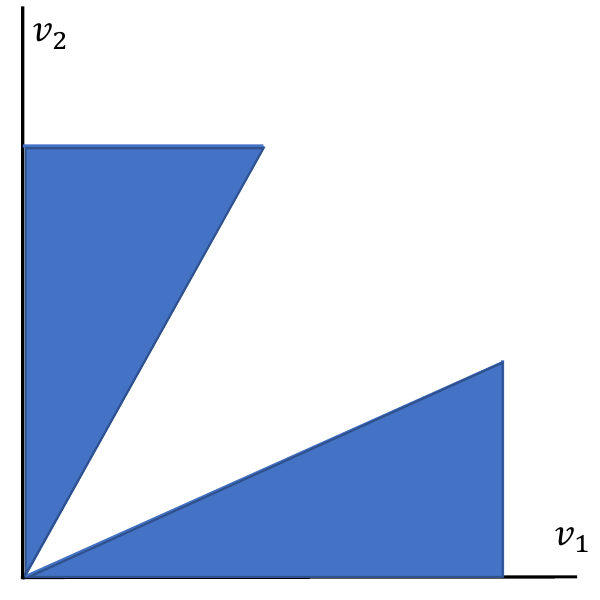}
    \caption{Value distribution in \cref{eg:cont_value}}
    \label{fig:cont_value}
\end{figure}
\subsubsection*{Markets with linear values}
Now we consider the condition (iii) of \cref{thm:second_order_informal}: linear valuations.
We adopt the setup in Section~4 from~\citet{gao2022infinite} where we impose an extra normalization on the values; the reasoning extends to the cases where at any point there are at most two lines intersecting at that point.  
Suppose the item space is $\Theta = [0,1]$ with supply $s(\theta) = 1$. The valuation of each buyer $i$ is linear and nonnegative: $v_i(\theta) = c_i \theta + d_i\geq 0$. Moreover, assume the valuations are normalized so that $\int_{[0,1]}v_i\diff \theta  = 1 \Leftrightarrow c_i/2 + d_i = 1$.
Assume the intercepts of $v_i$ are ordered such that $2 \geq d_{1}>\cdots>d_{n} \geq 0$.

We briefly review the structure of equilibrium allocation in this setting. By Lemma~5 from~\citet{gao2022infinite}, there is a unique partition $0 = a^*_0 < a_1^*<\cdots<a_n^*=1$ such that buyer $i$ receives $\Theta_i=\left[a_{i-1}^*, a_i^*\right]$. In words, the item set $[0,1]$ will be partitioned into $n$ segments and assigned to buyers $1$ to $n$ one by one starting from the leftmost segments. Intuitively, buyer 1 values items on the left of the interval more than those on the right, which explains the allocation structure.
Moreover, the equilibrium prices $p^*(\cdot)$ are convex piecewise linear with exactly $n$ linear pieces, corresponding
to intervals that are the pure equilibrium allocations to the buyers.

\begin{theorem}\label{thm:linear_value_implies_hessian}
    \new{In a market with linear values}, the dual objective $H$ is $C^2$ at $\betast$.
    Proof in \cref{proof:sec:analytical_properties_of_dual_obj}.
\end{theorem}

The above result also extends to most cases of piecewise linear (PWL) valuations discussed in Section~4.3 of~\citet{gao2022infinite}). 
In the PWL setup there is a partition of $[0,1]$, $A_{0}=0 \leq A_{1} \leq \cdots \leq A_{K-1} \leq A_{K}=1$, such that all $v_{i}(\theta)$'s are linear on $ [A_{k-1}, A_{k}]$. 
At the equilibrium of a market with PWL valuations, we call an item $\theta$ an \emph{allocation breakpoint} if there is a tie, i.e., $I(\betast, \theta)= \argmax _i \betai \vithe$ is multivalued.
Now suppose the following two conditions hold:
(i) none of the allocation breakpoints coincide with any of the valuation breakpoints $\{A_k\}$, 
and 
(ii) at any allocation breakpoint there are exactly two buyers in a tie.
Under these two conditions, one can show that in a small enough neighborhood of the optimal pacing multiplier $\betast$, 
the allocation breakpoints are differentiable functions of the pacing multiplier.
This in turn implies twice differentiability of the dual objective by repeating the argument in the proof of \cref{thm:linear_value_implies_hessian}. 
\new{Intuitively, this occurs because under the two conditions, the PWL case behaves like the linear case in a sufficiently-small neighborhood around the equilibrium.}
However, if either condition (i) or (ii) mentioned above breaks, the dual objective is not twice differentiable.


\subsubsection*{Markets with angularly smooth values}
We first use a change of variable, and let $z = v / \|v \|_2 $ be the projection of $v$ onto the unit sphere. It represents the angular component of the vector $v$.
Let $f(v, \b) = \max_i \betai\vithe$.
Using a change of variable $z = v/\|v\|_2$ and $r =\|v\|_2$ and homogeneity of $f$, 
the integral $\int f( v, \b) f_v \diff v$ can be written as 
\begin{align}
    \int _{S_n} \int_0^\infty f(rz, \b) r^{n-1} f_v(rz) \diff r \diff z 
    = 
    \int _{S_n} f(z, \b) \bigg( \int_0^\infty r^{n} f_v(rz) \diff r \bigg) \diff z  
\end{align}
where $\diff z$ is the surface measure on the unit ball $S_n$ in $\Rn$.
From this representation, it is not unreasonable to say
that if $\int_0^\infty  r^{n} f_v(rz) \diff r$ is a smooth function of the angular component $z$,
then $\fbar(\b)$ will also be a smooth function of $\b$.
Taking $v_i = r z_i$, $v_k = z_k/z_i$, and so $ \int_0^\infty  r^{n} f_v(rz) \diff r =\int_0^\infty v_i^n  f_v(v_iv_1, v_iv_2,\dots, v_i, \dots, v_iv_n) \diff v_i/ (z_i)^n$.
So equivalently 
we require smoothness of $v_{-i} \mapsto \int_0^\infty v_i^n f_v(v_iv_1, v_iv_2,\dots, v_i, \dots, v_iv_n) \diff v_i$.

Now we make this precise.
First we introduce some extra notations. For each $i\in[n]$, define the map $\sigma_i: \Rnp \to \Rnp$,
$$\sigma_i(v) = [v_1v_i,\,  \dots,\,v_{i-1}v_i,\, v_i,\, v_{i+1}v_i,\, \dots,\, v_nv_i ] \tp $$ for $i\in[n]$, which multiplies all except the $i$-th entry of $v$ by $v_i$. 

\begin{definition}[Angular regularity]
    Let $f: \Rnp \to \Rp$ be the probability density function (w.r.t.\ the Lebesgue measure) of a positive-valued random vector with finite first moment.
    We say the density $f$ is {angularly regular} if for all $ h_i(v_{-i} )\defeq \int_0^\infty f \big(\sigma_i(v)\big) v_i^n \diff v_i\,$, $ i \in [n]$, it holds
    (i) $h_i$ is continuous on $\R^{n-1}_{++}$, and
    (ii) all lower dimensional density functions of $h_i$ are continuous (treating $h_i$ as a scaled probability density function).  
\end{definition}

\begin{theorem}\label{thm:smooth_density_implies_hessian}
    Assume the random vector $[v_1,\dots, v_n]: \Theta \to \Rnp$ has a distribution absolutely continuous w.r.t.\ the Lebesgue measure on $\Rn$ with density function $f_v$.
    If $f_v$ is angularly regular, then $H$ is twice continuously differentiable on $\Rnpp$.

    Proof in \cref{proof:sec:analytical_properties_of_dual_obj}.
\end{theorem}

The above regularity conditions are easy to verify when the values are i.i.d.\ draws from a distribution. In that case, many smooth distributions supported on the positive reals fall under the umbrella of the described regularity. Below we examine three cases: the truncated Gaussian distribution, the exponential distribution and the uniform distribution.

When values are i.i.d.\ truncated standard Gaussians, the joint density $f(v) =c_1 \prod_{i=1}^n \exp(- v_i\sq / 2)$ and 
$h_i (v_{-i}) 
= c_1  \int_\Rp v_i^n \exp ({- \frac12 v_i\sq ( 1 + \sum_{k \neq i} v_k\sq)}) \diff v_i = c_2 (\sum_{k \neq i} v_k\sq)^{-n/2},$
which are regular. Here $c_i$, $i=1,2$, are appropriate constants.
Similarly, for the i.i.d.\ exponential case with the rate parameter equal to one, the density $f(v) = \prod_{i=1}^n \exp(- v_i)$ and $h_i (v_{-i}) =  (\sum_{k\neq i} v_k)^{-n}$ satisfy the required continuity conditions. 
Finally, suppose the values are i.i.d.\ uniforms on $[0,1]$. The joint density is $f(v) = \prod_{i=1}^n \indi\{0<v_i < 1\}$ and for example, if $i=1$, $h_1(v_{-1}) = (\min\{1, v_2\inv, \dots, v_n\inv \})^{n+1}/(n+1)$, which also satisfies the required continuity conditions.

\subsection{Proofs}
\label{proof:sec:analytical_properties_of_dual_obj}

\begin{proof}[Proof of \cref{thm:first_differentiability}]
    Recall $f(\theta, \b) = \max_i \beta_i \vithe$. Note $f$ is differentiable at $\beta$ if and only if $ \bidgap(\beta,\theta) > 0$. 
    Let $\Theta_{\mathrm{diff}}(\beta) \defeq \{ \theta: f(\theta, \b) \text{ is continuously differentiable at $\beta$} \}$. Then 
    \begin{align*}
        \Theta_{\mathrm{diff}}(\beta) = \bigg\{\theta: \frac{1}{\bidgap(\beta,\theta)} < \infty \bigg\} = \{ \theta: I(\beta,\theta) \text{ is single-valued}\} 
        \;.
    \end{align*}
    
    By Proposition 2.3 from~\citet{bertsekas1973stochastic} we know $\fbar(\beta) = \E[f(\t,\b)] = \int_\Theta f(\t,\b) s(\theta) \diff \t$ is differentiable at $\beta$ if and only if $\int \indi( \Theta_{\mathrm{diff}}(\beta)) s(\t)\dt = 1$.
    From here we obtain  \cref{thm:first_differentiability}.

\end{proof}
\begin{remark}

    Suppose \cref{eq:as:notie} holds in a neighborhood $N$ of $\betast$,
    i.e., $\frac{1}{\bidgap(\beta,\theta)}$ is finite a.s.\ for each $\beta \in N$, then by \cref{thm:first_differentiability} we know $H$ is differentiable on $N$.
    In fact, a stronger statement holds: $H$ is \textbf{continuously} differentiable on $N$.
    See Proposition 2.1 from \citet{shapiro1989asymptotic}.
\end{remark}
\begin{remark}[Comment on \cref{thm:first_differentiability}]
    We briefly discuss why differentiability is related to the gap in buyers' bids.
Recall $\fbar(\beta)= \E[\max_i \beta_i v_i(\theta)]$.
Let $\delta \in \Rnp$ be a direction with positive entries, and let $I(\beta, \theta) = \argmax_i{\beta_i v_i(\theta)}$ be the set of winners of item $\theta$ which could be multivalued. Consider the directional derivative of $\fbar$ at $\beta$ along the direction $\delta$:
\begin{align*}
    & \lim_{t \downarrow 0} 
     \E\bigg[\frac{\max_i (\beta_i + t \delta_i) v_i(\theta) - \max_i \beta_i v_i(\theta)}{t}\bigg] 
    \\
    & =  
    \E \bigg[\lim_{t \downarrow 0} \frac{\max_i (\beta_i + t \delta_i) v_i(\theta) - \max_i \beta_i v_i(\theta)}{t}\bigg] 
    \\
    & =
    \E\Big[\max_{i \in I(\beta,\theta)} v_i (\theta) \delta_i\Big] 
    \;,
\end{align*}
where the exchange of limit and expectation is justified by the dominated convergence theorem.
Similarly, the left limit is
\begin{align*}
    \lim_{t \uparrow 0} \E\bigg[\frac{\max_i (\beta_i + t \delta_i) v_i(\theta) - \max_i \beta_i v_i(\theta)}{t}\bigg] = \E\Big[\min_{i \in I(\beta,\theta)} v_i(\theta) \delta_i\Big] 
    \;.
\end{align*}
If there is a tie at $\beta$ with positive probability, i.e., the set $I(\beta,\theta)$ is multivalued for a non-zero measure set of items, then the left and right directional derivatives along the direction $\delta$ do not agree. Since differentiability at a point $\beta$ implies existence of directional derivatives, we conclude differentiability implies \cref{eq:as:notie}.

\end{remark}

\begin{proof}[Proof of \cref{thm:int_implies_hessian}]
    By assumption,
    there is a neighborhood of $\betast$, say $N$, on which $H$ is differentiable.
    By \cref{thm:first_differentiability},
    for $\b\in N$, $f$ is differentiable at $\b$ almost surely.
    Define $G: N \times \Theta \to \Rnp$, $G(\b,\t) = \nabla f(\t, \b)$.
    To compute the Hessian w.r.t.\ the first term $\fbar$, we look at the limit 
    \begin{align}
        \label{eq:G_ratio}
        \lim_{h \to 0} \E \bigg[\frac{G(\betast + h,\theta) - G(\betast, \theta)}{\|h\|}\bigg]
        \;.
    \end{align}

    Suppose we could exchange expectation and limit in \cref{eq:G_ratio}, then the above expression would become zero: for a fixed $\theta$, since \cref{eq:as:notie} holds at $\betast$, i.e, $\bidgap(\betast,\theta) > 0$, we apply \cref{lm:epsilon_as_lip} and obtain
    $
        \lim_{h \to 0}  {(G(\betast + h,\theta) - G(\betast, \theta))}/{\|h\|} = 0.
    $
    This implies that $H$ is twice differentiable at $\betast$ with Hessian $\nabla\sq H(\betast) = \nabla\sq \Psi(\betast)$.
    It is then natural to ask for sufficient conditions for exchanging limit and expectation.

    \begin{lemma}[$\bidgap(\beta,\theta)$ as   Lipschitz parameter of $G$] 
        \label{lm:epsilon_as_lip}
      Let $\b, \b' \in N$. 
        \begin{itemize}
            \item If $\|\b - \b'\|_\infty \leq \bidgap(\beta,\theta)/\vbar$ then 
            $G(\beta',\theta) = G(\beta,\theta)$.
            \item It holds ${ \| G(\beta ',\theta) - G(\beta, \theta) \|_2}{} \leq 6 \vbar\sq \cdot \frac{1}{ \bidgap(\beta,\theta)} \|\b' - \b\|_2 $. 
        \end{itemize}
\end{lemma}

    By \cref{lm:epsilon_as_lip}, we know the ratio 
    ${(G(\betast + h,\theta) - G(\betast, \theta))}/{\|h\|}$ is dominated by $6\vbar \bidgap(\betast,\theta)\inv$, which by \cref{eq:as:UI} is integrable. 
    By dominated convergence theorem, we can exchange limit and expectation, and the claim follows.
\end{proof}
    
\begin{proof}[Proof of \cref{lm:epsilon_as_lip}]

    Note that for any $\beta$ and $\theta$ with $\bidgap(\beta,\theta) > 0$, and any $\beta' = \beta + h$,
\begin{align} 
    \label{eq:lip_of_G}
    \frac{ \| G(\beta + h,\theta) - G(\beta, \theta) \|_2}{\|h\|_2} \leq 6 \vbar\sq \cdot \frac{1}{ \bidgap(\beta,\theta)}
    \;.    
\end{align}
To see this, we notice that on one hand, if $\| h \|_\infty\leq \epsilon / (3\vbar )$ where $\epsilon = \bidgap(\beta,\theta)$, then for $i = i(\beta,\theta)$ and all $\theta \in \Theta_i(\beta)$,
\begin{align*}
    \beta'_i v_i(\theta) 
    & = (\betai + h_i) v_i(\theta)
    \\
    &\geq \betai v_i(\theta) -  \epsilon /3 \tag{A}
    \\
    &\geq \beta_k v_k(\theta) + \epsilon -  \epsilon /3
    \tag{B}
    \\
    &\geq \beta'_k  v_k(\theta)  - \epsilon /3 + \epsilon -  \epsilon /3 \tag{C} 
    \;,
\end{align*}
where $(A)$ and $(C)$ use the fact $\|h\|_\infty \leq \epsilon / (3\vbar)$, and (B) uses the definition of $\epsilon$.
This implies $\argmax_i \beta'_i \vithe = \argmax_i \betai \vithe$ and thus
$G(\beta + h ,\theta) - G(\beta  ,\theta) = 0$.
On the other hand, if $\| h \|_\infty > \epsilon / (3\vbar)$, then $\| h\|_2\geq  \|h\|_\infty > \epsilon / (3\vbar)$. Using the bound $\| G \|_2 \leq \vbar$, we obtain \cref{eq:lip_of_G}.
This completes proof of \cref{lm:epsilon_as_lip}.
\end{proof}

\begin{proof}[Proof of \cref{thm:linear_value_implies_hessian}]
    Recall the normalization on values $c_i/2 + d_i = 1$.
    By Lemma 5 from~\citet{gao2022infinite}, we know that at the LFM equilibrium the there exists unique breakpoints 
    $ 0 = a_0^* < a_{1}^{*}<\cdots<a_{n}^{*}=1$, $a^*_i = (-\betasti d_i + \betast_{i+1}d_{i+1}) / (\betasti c_i - \betast_{i+1} c_{i+1}) $, 
    such that buyer $i$ receives the item set $[a_{i-1}^*, a_i^*] \subset \Theta = [0,1]$. 
    Moreover, it holds 
    \begin{align*}
        & \beta^*_1 d_1 > \beta^*_2 d_2 > \dots > \beta^*_n d_n \;,
        \\
        & \beta^*_1 c_1 < \beta^*_2 c_2 < \dots < \beta^*_n c_n \;.
    \end{align*}

    Now we consider a small enough neighborhood $N$ of $\betast$. For each $\beta\in N$, we define the 
    breakpoint $a_i^*(\beta)= (-\beta_i d_i + \beta_{i+1}d_{i+1}) / (\beta_i c_i - \beta_{i+1} c_{i+1}) $ by solving for $\theta$ through 
    $\beta_i(c_i \theta + d_i) = \beta_{i+1}(c_{i+1} \theta + d_{i+1})$ for $i\in[n-1]$, $a_0^*(\beta) = 0$, and $a_n^*(\beta) = 1$. 
    Let $\Theta_i(\beta) = \{\theta \in \Theta: v_i(\theta)\beta_i \geq v_k(\theta)\beta_k, \forall k\neq i \}$. 
    \begin{lemma}
        \label{lm:linearv}
        There is a neighborhood $N$ of $\betast$, so that 
        $\Theta_i(\beta) = [a_{i-1}^*(\beta), a_{i}^*(\beta) ]$ for $\beta \in N$. 
    \end{lemma}
    
    When exists, the gradient is always $ \nabla \fbar(\beta) 
     = \sumiton e_i \int \indi(\Theta_i(\b)) v_i(\t) s(\t)\dt$.
    For $\b \in N$, it further simplifies to
    \begin{align*}
        \nabla \fbar(\beta)  
        & = \sumiton e_i \int \indi([a_{i-1}^*(\beta) ,a_i^*(\beta)]) (c_i \theta + d_i) \diff \theta
        \\
        & = \sumiton e_i \bigg( \frac{c_i}{2}\big( [a_{i}^*(\beta)]\sq - [a_{i-1}^*(\beta)]\sq \big) + d_i \big(a_{i}^*(\beta) - a_{i-1}^*(\beta)\big) \bigg) 
        \;.
    \end{align*}
    From this we see that continuous differentiability of the breakpoints $a^*_i (\b)$ implies continuous differentiability of $\nabla  \fbar$.
    This finishes the proof of \cref{thm:linear_value_implies_hessian}.
\end{proof}

\begin{proof}[Proof of \cref{lm:linearv}]
    We construct such a neighborhood $N$. Define $$\delta = \min
    \bigg\{
        \frac12 \ubar{\Delta}_{\beta d} / \bar{\Delta}_d
        , \frac12 \ubar{\Delta}_{\beta c} / \bar{\Delta}_c 
        ,
        \frac {1}{4} \ubar{\Delta}_a \ubar{\Delta}_{\beta c} / \vbar 
        \bigg\}\; ,$$ where 
        $\ubar{\Delta}_a\defeq \min |a_i - a_{i-1}|$,
    $ \ubar{\Delta}_{\beta c} \defeq \min \{  \beta^*_{i-1}c_{i-1} -\beta^*_i c_i \} > 0$, 
    $ \bar{\Delta}_c \defeq \max_i\{c_{i-1} - c_{i}\} > 0$, and $ \ubar{\Delta}_{\beta d} > 0$ and $ \bar{\Delta}_d > 0$ are similarly defined. 
    Let $N = \{ \beta: \|\beta - \betast\|_\infty < \delta \}$.
    The neighborhood $N$ is constructed so that on $N$ it holds 
    \begin{align*}
        & \beta_1 d_1 > \beta_2 d_2 > \dots > \beta_n d_n \; ,
        \\
        & \beta_1 c_1 < \beta_2 c_2 < \dots < \beta_n c_n \; ,
        \\
        & 0 = a_0^*(\beta) < a_{1}^{*} (\beta) <\cdots<a_{n}^{*}(\beta)=1 \; ,
    \end{align*}
    where the first inequality follows from 
    $\delta \leq \frac12 \ubar{\Delta}_{\beta d} / \bar{\Delta}_d
    $, the second inequality from $\delta \leq \frac12 \ubar{\Delta}_{\beta c} / \bar{\Delta}_c$, and the third inequality follows from $\delta \leq \ubar{\Delta}_a \ubar{\Delta}_{\beta c} / (4\vbar)$, where $\vbar = \max_i \sup_{\theta \in [0,1]} c_i \theta + d_i$.
    For $\t\in[0,a_1^*(\b)]$,
    $\max_i \beta_i (c_i\t + d_i)$ is achieved by $i=1$.
    Similarly for $i = 2, \dots, n$.
    So we have shown $\Theta_i(\beta) = [a_{i-1}^*(\beta), a_{i}^*(\beta) ]$ for $\beta \in N$.
\end{proof}

\begin{proof}[Proof of \cref{thm:smooth_density_implies_hessian}]
    We need the following technical lemma on the continuous differentiability of integral functions.
    \begin{lemma} [Adapted from Lemma 2.5 from~\citet{wang1985distribution}]
        \label{lm:integral_is_cont_diff}
        Let $u = [u_1, \dots, u_n]\in \Rnpp$, and $$I(u) = \int_{0}^{u_1} \diff t_1 \dots \int_{0}^{u_n}  h (t_1, t_2,\dots, t_n)  \diff t_n \; ,$$ where $h$ is a continuous density function of a probabilistic distribution function on $\Rnp$ and such that all lower dimensional density functions are also continuous. Then the integral $I(u)$ is continuously differentiable.
    \end{lemma}
    \begin{remark}
        The difference between the above lemma and the original statement is that  the original theorem works with density $h$ and integral function $I(u)$ both defined on $\Rn$, while the adapted version works with density $h$ and integral $I(u)$ defined only on $\Rnpp$. 
    \end{remark}

    The gradient expression is $$\nabla \fbar(\beta)  = \sumiton e_i \int v_i \indi(V_i(\beta)) f_v(v)\diff v \; ,$$
    where the set $V_i(\beta) = \{ v \in \Rnp: v_i \beta_i \geq v_k \beta_k, k\neq i\}, i \in [n]$, is the values for which buyer $i$ wins. 
    For now, we focus on the first entry of the gradient, i.e., $\int v_1 \indi(V_1(\beta)) f_v(v)\diff v$.
    We write the integral more explicitly as follows. By Fubini's theorem,
    \begin{align}
        & \int v_1 \indi(V_1(\beta)) f_v(v)\diff v
        \\
        &  = \int_{0}^{\infty} \diff v_1 \int_0^{\frac{\beta_1 v_1}{\beta_2}} \diff v_2 \int_0^{\frac{\beta_1 v_1}{\beta_3}} \diff v_3 \dots 
        \int_0^{\frac{\beta_1 v_1}{\beta_n}}
        \underbrace{(v_1 f_v(v_1,\dots, v_n) )}_
        {=: A_1(v)}\diff v_n 
        \label{eq:integral_before_change_of_var}
        \;.
    \end{align}
    To apply the lemma we use a change of variable. Let $t = T(v) = [v_1,\frac{v_2}{v_1}, \dots, \frac{v_n}{v_1}]$ and $v = T\inv (t) = [t_1, t_2 t_1, \dots, t_n t_1]$. Then \cref{eq:integral_before_change_of_var} is equal to 
    \begin{align}
        \int_{0}^{\infty} \diff t_1 \int_0^{\frac{\beta_1}{\beta_2}} \diff t_2 \int_0^{\frac{\beta_1 }{\beta_3}} \diff t_3 \dots 
        \int_0^{\frac{\beta_1}{\beta_n}} \underbrace{\big(t_1^n f_v(t_1, t_2 t_1, \dots, t_n t_1)\big)}_
        {=: A_2(t)}
        \diff t_n 
        \;.
        \label{eq:after_change_of_var}
    \end{align}
    Note $\E[v_1(\theta)] = \int_\Rnpp A_1(v) \diff v= \int_\Rnpp A_2(t) \diff t = 1 $.
    We use Fubini's theorem and obtain 
    \begin{align*}
        \cref{eq:after_change_of_var} = \int_0^{\frac{\beta_1}{\beta_2}} \diff t_2 \int_0^{\frac{\beta_1 }{\beta_3}} \diff t_3 \dots 
        \int_0^{\frac{\beta_1}{\beta_n}} h(t_{-1}) \diff t_{n}
        \;,
    \end{align*}
    where we have defined $h(t_{-1}) = \int_\Rp t_1^n f_v(t_1, t_2 t_1, \dots, t_n t_1) \diff t_1$. 
    By the smoothness assumption on $h$ and \cref{lm:integral_is_cont_diff}, 
    we know that the map $u_{-1} \mapsto \int_{0}^{u_2} \diff t_2 \dots \int_{0}^{u_n}  h (t_{-1})  \diff t_n $ is $C^1$ for all $u_{-1} \in \Rnmpp $. 
    Moreover, 
    the map $\beta \mapsto [\frac{\beta_1}{\beta_2}, \dots, \frac{\beta_1}{\beta_n}]$ is $C^1$. We conclude the first entry of $\nabla \fbar(\beta)$ is $C^1$ in the parameter $\beta$. 
    A similar argument applies to other entries of the gradient. We complete the proof of \cref{thm:smooth_density_implies_hessian}.
\end{proof}

\subsection{Closed-form Expression for Hessian}
\label{sec:closedformhessian}

In this section we derive a closed form expression for $\nabla \sq \fbar (\b)$ for some $\b \in \Rnp$ using tools from differential geometry.
Let $f_v(v)$ be the density of values w.r.t.\ the Lebesgue measure on $\Rn$. Let $\cK = \nabla\sq \fbar(\beta)$. Then $\nabla\sq H(\b)= \cK + \Diag(b_i / (\betai)\sq)$.
Fix a buyer $k$. Now we derive $\cK_{k,h}$ for $h\in [n]$. 
We need to introduce a few sets.
Let
\begin{align*}
    & V = \{ v \in \Rnp:  \betai\vi \leq \betak\vk, i \in [n] \}
    \;,
    \\
    & S_h = \{ v \in \Rnp: \betak\vk = \betah \vh, \betai\vi \leq \betak\vk, i \neq k, h \} = V \cap  \{v: \betak\vk = \betah \vh\} 
    \;,
    \\
    & \Pi_ h = \{ v_{-k} \in \R^{n-1}_+: \betai\vi \leq \betah \vh, i \neq k, h \}
    \;.
\end{align*}
\begin{figure}[ht!]
    \centering
    \includegraphics[scale=0.5]{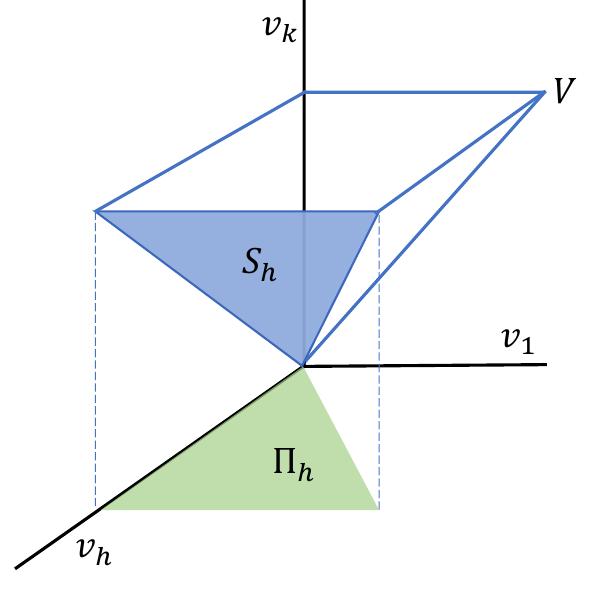}
    \caption{Illustration of $V$, $S_h$ and $\Pi_h$.}
    \label{fig:winningcone}
\end{figure}
The cone $V$, which we call the \emph{winning cone} of buyer $k$, 
is the set of valuation vectors such that buyer $k$ wins under the pacing profile $\b$.
The set $S_h$ is a face of $V$, which represents
the values for which there is a tie between buyers $k$ and $h$.
It is clear $\Pi_h$ is the projection of $S_h$ onto $\{v \in \Rn: v_h = 0 \}$.
In \cref{fig:winningcone} we present an illustration of these sets.

We will show that the Hessian $\mathcal K$ can be characterized as
\begin{align}
    & \cK_{k,h} = - \frac{\betah}{\betak\sq} \int_{\Pi_h} \vh\sq f_v(v_1, \dots, v_{k-1}, \frac{\betah\vh}{\betak}, v_{k+1}, \dots, v_n) \diff v_{-k}
    \quad \text{ for } h \neq k \;,
    \\
    & \cK_{k,k} = \sum_{h\neq k} \frac{\betah\sq}{\betak^3} \int_{\Pi_h} \vh\sq f_v (v_1, \dots, v_{k-1}, \frac{\betah\vh}{\betak}, v_{k+1}, \dots, v_n) \diff v_{-k}  \;.
\end{align}
The formulae indicate that 
only the value of $f_v$ on the faces $S_h$, $h\neq k$ matters for the Hessian.

Let $\Delta = \betak e_k - \betah e_h$ and 
$\Delta_u = \betak e_k - u e_h$ for a scalar $u$.
The $k$-th entry of $\nabla \fbar (\b)$ is $\int \indi(V(\b)) \vk f_v \diff v$ and so the second-order derivative of $\fbar$ can be written as
\begin{align*}
    \cK_{k,h} = \frac{\partial}{\partial \betah} \int \indi(V(\b)) \vk f_v \diff v
    \;.
\end{align*}
Define the rotation matrix 
\begin{align*}
    T_u = \bigg(I_n - \frac{\Delta_u \Delta_u \tp}{\| \Delta_u \|_2\sq}\bigg)\bigg(I_n - \frac{\Delta \Delta \tp}{\| \Delta \|_2\sq}\bigg) + \frac{1}{\|\Delta\|_2 \|\Delta_u\|_2} \Delta_u\Delta\tp
    \;.
\end{align*}
The map $T_u$ is a diffeomorphism that maps $\{v: \Delta \tp v \geq 0\}$
to the region $\{v: \Delta_u \tp v \geq 0\}$.
It can be seen that $T_u (\Delta / \|\Delta\|_2) = \Delta_u / \|\Delta_u\|_2$ and $T_{\betah} = I_n$.
Notice that if we increase $\betah$, $h\neq k$, the face $S_h$ rotates inward in $V$, pivoting at the origin.
By a result from Section 5 in \citet{kim1990cube}, 
the derivative of the volume of a parametrized region can be written as 
a surface integral on the boundary of the region. Concretely, in our 
case, we have that
\begin{align*}
    \frac{\partial}{\partial \betah} \int \indi(V(\b)) \vk f_v \diff v
    =
    \int \indi(S_h) \vk f_v  n(v)\tp \Big(\frac{\partial}{\partial u} T_u v \Big| _{u=\betah}\Big)
    \diff \sigma_h \;,
\end{align*}
where $n(v)$ is the normal vector of $S_h$ pointing inside the winning cone $V$, and $\diff \sigma_h$ is the surface measure on $S_h$.
Now plug in 
\begin{align*}
    & \frac{\partial}{\partial u} T_u v \Big| _{u=\betah} = - \frac{\vh}{\|\Delta\|_2\sq} \Delta \text{ for $ v \in S_h$} 
    \;,
    \\
    & n(v) = \frac{\Delta}{\|\Delta\|_2}\;, \quad \diff \sigma_h = \sqrt{1 + (\betah/ \betak)\sq} \diff v_{-k}
    \;,
\end{align*}
and obtain 
\begin{align*}
    & \int \indi(S_h) \vk f_v  n(v)\tp \Big(\frac{\partial}{\partial u} T_u v \Big| _{u=\betah}\Big)
    \diff \sigma_h 
    \\
    & = - \int_{S_h} \vk \vh f_v \|\Delta\|_2\inv \diff \sigma_h
    \\
    & = - \frac{\betah}{\betak\sq} \int_{\Pi_h} \vh\sq f_v(v_1, \dots, v_{k-1}, \frac{\betah\vh}{\betak}, v_{k+1}, \dots, v_n) \diff v_{-k} 
    \;.
\end{align*}
Now we investigate $\cK_{k,k}$. 
Notice that if we increase $\betak$, all $n-1$ faces, $\{S_h\}_{h\neq k}$, of the winning cone $V$ rotate outward, pivoting at the origin. And so we could again use the results from \citet{kim1990cube}.
However, we show a simpler approach to deriving $\cK_{k,k}$ using first-order homogeneity of the function $\fbar(\b)=\E[\max_i \betai \vithe]$.
By Euler's homogenous function theorem, 
we have $\fbar(\b) = \sumiton \betai (\partial / \partial \betai) \fbar(\b)$.
Taking $(\partial / \partial \betak)$ on both sides we obtain 
$$\frac{\partial }{ \partial \betak } \fbar(\b) =
\frac{ \partial }{\partial \betak } \fbar(\b) + \betak \cK_{k,k} + \sum_{h\neq k}  \betah \cK_{k,h} \;,$$ 
and thus $ \cK_{k,k} = -\sum_{h\neq k} \betah \cK_{k,h} / \betak$.

\newpage
\ECHead{Proofs of Fisher Market Results}

\section{Appendix to Linear Fisher Market}

\subsection{Fairness, efficiency, and Scale Invariance}\label{sec:lfm_scale_invariance}

A major use case for LFM is fair and efficient allocation of resources. 
Similar to the classical finite LFM, the infinite LFM enjoys fairness and efficiency properties~\citep{gao2022infinite}. Let $\xst$ be an equilibrium allocation. 
First, this allocation is \textit{Pareto optimal}, meaning there does not exist an $\tilde x \in (L_+^\infty) ^ n$, $\sumi \tilde x_i \leq 1$, such that $\int v_i \tilde x _i s\dt \geq \int v_i \xsti  s\dt $ for all $i$ and one of the inequalities is strict.
The allocation $\xst$ is \textit{envy-free} in a budget-weighted sense, meaning $\int v_i \xsti  s\dt/ b_i \geq \int \new{ v_i} \xst_j  s\dt / b_j$ for all $i\neq j$. 
Finally, it is \textit{proportional}: $\int v_i \xsti  s\dt \geq \int v_i s \dt \cdot (b_i / \sum_{i'}b_{i'}) $, that is, each buyer gets at least the utility of its proportional share allocation.

LFM enjoys certain scale-invariance properties.
First, buyers cannot change the equilibrium by scaling their value functions.
Suppose $\alpha_1,\dots, \alpha_n$
are positive scalars, and that
$(x,p)$ are the equilibrium allocations and prices in $\LFM(b,v,s,\Theta)$.
Then $(x,p)$ will also be the equilibrium quantities in $\LFM(b,(\alpha_1 v_1,
\dots ,\alpha_n v_n), s, \Theta)$.
This is easily seen by the fact that valuation scaling does not change the demand function.
Second, if all buyers' budgets are scaled by the same factor, or the supply is scaled by the same factor, the equilibrium does not change. That is, if $\alpha_1,\alpha_2$ are two positive scalars, and $(x,p) = \LFM(b,v,s,\Theta)$, then $(x, (\alpha_1  / \alpha_2) p) = \LFM(\alpha _1 b, v, \alpha_ 2 s, \Theta)$.
These scale-invariances hold for finite LFM as well \footnote{
    That is, 
    $(\xgam,\pgam) \in \oLFM (b,v,\sigma, \gam)$ implies 
$  (\xgam,\pgam)\in \oLFM ( b, (\alpha_i v_i), \sigma, \gam)$, 
and 
$(x, (\alpha_1/\alpha_2) p) \in \oLFM(\alpha_1 b, v, \alpha_2 \sigma, \gam)$
}.
Based on the invariance, we impose the normalization that 
$\sumiton b_i = 1$ and that all buyers' expected values are 1, i.e., 
$\int v_i (\t) s(\t)\diff \t = 1$.
Then in the limit LMF, the budge-supply ratio is 
$\sumiton b_i / \int s \diff \t = 1$.
With the normalization, we have $C_\LFM = \prod_{i=1}^{n} [b_i/2, 2]$.
In order for the budget-to-supply ratio to match in the sampled finite LFM and the limit LFM, we 
use supplies of $1/t$ for each item in the finite LFM.
Thus we study finite LFM of the form $\oLFM(b,v, 1/t, \gamma)$.

\subsection{Consistency} \label{sec:consistency}

\begin{theorem}[Consistency]\label{thm:consistency}
    It holds that 
    \begin{enumthmresult}
        \item NSW and individual utilities in the finite LFM are strongly consistent estimators of their limit LFM counterparts, i.e., $\sumiton b_i \log (u^\gam_i) \toas  \sumiton b_i \log (u^*_i)$ and $u^\gam_i \toas u^*_i.$
        \label{it:thm:consistency:1}
        
        \item The pacing multiplier in the finite LFM is a strongly consistent estimator of its limit LFM counterpart, i.e., $\beta^\gam_i \toas \beta^*_i$.
        \label{it:thm:consistency:2}


        \item Convergence of approximate market equilibrium: 
        $\limsup _{t} \cB^\gam(\epsilon) \subset \cB^*(\epsilon) \text { for all } \epsilon \geq 0$ 
        and 
        $\limsup _{t} \cB^\gam(\epsilon_t) \subset \cB^*(0) = \{\betast\} \text { for all } \epsilon_t \downarrow 0$. Recall the approximate solutions set, $\cBgam$ and $\cBst$, are defined in  \cref{eq:def:cBstar}.
        \label{it:thm:consistency:4}
    \end{enumthmresult}
    Proof in \cref{sec:proof:thm:consistency}.

\end{theorem}

We briefly comment on \cref{it:thm:consistency:4}.
The set limit result can be interpreted from a set distance point of view.
We define the inclusion distance from a set $A$ to a set $B$ by $d_\subset(A,B ) \defeq \inf _\epsilon\{\epsilon \geq 0: A \subset\{y: \operatorname{dist}(y, B) \leq \epsilon\}\}$ where $\operatorname{dist}(y, B) 
\defeq \inf\{\|y - b\| : b\in B \}$. 
Intuitively, $d_\subset(A,B )$ measures how much one should enlarge $B$ such that it covers $A$.
Then for any sequence $\epsilon_n \downarrow 0$, by the second claim in \cref{it:thm:consistency:4}, we know $d_\subset(\cB^\gam(\epsilon_t), \{\betast\}) \to 0$. This shows that the set of approximate solutions of $H_t$ with increasing accuracy centers around $\betast$ as market size grows.

\section{Technical Lemmas for LFM}
\label{sec:tech_lemmas}

We abbriviate $C_\LFM$ to $C$, and recall 
that under the normalization that $\sumiton b_i  = 1$ and $\nu_i = \int v_i s \dt = 1$, the set $C 
= \prod_{i=1}^n [b_i/2, 2] \subset \R^n$.

\begin{lemma} 
    \label{lm:value_concentration} 
    Define the event $A_t = \{ \betagam \in C   \}$
    .
    (i) If $t \geq 2\vbarsq {\log(2n/\eta)}$, then $\P(A_t) \geq \P(\frac12 \leq \frac1t \sumtau \vithetau \leq 2,\forall i ) \geq 1- \eta$.
    \label{it:lm:value_concentration:1}
    (ii) It holds $\P( A_t \text{ eventually}) = 1$.
    \label{it:lm:value_concentration:2}
    Proof in \cref{sec:tech_lemmas}.
\end{lemma}

\begin{proof}[Proof of \cref{lm:value_concentration}]
    Recall the event $A_t = \{ \betagam \in C   \}$.
Define $\vbarit = \frac1t \sumtau \vithetau$.

        First we notice concentration of values implies membership of $\betagam$ to $C$, i.e.,
        $\{ 1/2 \leq \vbarit \leq 2,\forall i \} \subset \{\betagam \in C \}$. To see this, note that
        $\ugami \leq \frac1t \sumtau\vithetau$ and $\ugami \geq \frac1t\frac{b_i}{\sumiton b_i} \sumtau \vithetau$, and through the equation $\betagami = b_i / \ugami$ the inclusion follows.
        Note $0\leq \vithetau\leq \vbar$ is a bounded random variable with mean $\E[\vithetau] =1 $. By Hoeffding's inequality we have
        $   
            \P(|\vbarit - 1| \geq \delta)\leq 2\exp(-\frac{2\delta\sq t}{\vbar\sq})
        $.
        Next we use a union bound and obtain
        \begin{align}
            \label{eq:hoeffding_average_values}
            \P(\betagam \notin C) \leq \P \bigg(
            \bigcup_{i=1}^n \big\{ |\vbarit - 1| \geq \delta \big\} \bigg) \leq 2n \exp \Big(-\frac{2\delta\sq t}{\vbar\sq}\Big)
            \;.
        \end{align}

        By setting $2n \exp(-\frac{2\delta\sq t}{\vbar\sq}) = \eta$ and $\delta = 1/2$ and solving for $t$ we obtain item (i) in claim.
    
        To show item (ii),
        we use the Borel-Cantelli lemma. By choosing $\delta  = 1/2$ in the \cref{eq:hoeffding_average_values} we know $\P(A_t^c) \leq \P(\{ 1/2 \leq \vbarit \leq 2,\forall i \}^c)\leq 2n\exp(-t/(2\vbarsq))$. Then we have
        $
            \sum_{t=1}^\infty \P(A_t^c) < \infty
            \;.
        $
        By the Borel-Cantelli lemma it follows that $\P(\{A_t^c \text{ infinitely often}\}) = 0$, or equivalently $\P(A_t \text{ eventually}) = 1$.
    \end{proof}
    
\begin{lemma}[Smoothness and Curvature] \label{lm:smooth_curvature}
    It holds that 
both $H$ and $H_t$ are $L$-Lipschitz and $\lambda$-strongly convex w.r.t\ the $\ell_\infty$-norm on $C$ with $L = 2n+\vbar$ and $\lambda = \ubar{b}/4$. Moreover, $H_t$ and $H$ are $(\vbar + 2\sqrt{n})$-Lipschitz w.r.t.\ $\ell_2$-norm.
\end{lemma}
\begin{proof}[Proof of \cref{lm:smooth_curvature}]

    Now we verify that $H_t$ and $H$ are $(\vbar + 2n)$-Lipschitz on the compact set $C$ w.r.t.\ the $\ell_\infty$-norm. For $\beta,\beta' \in C$,
    \begin{align*}
        & |H_t(\beta)  - H_t(\beta')| 
        \\
        & \leq \frac1t\sumtau \big|\max_i \{\vithetau \beta_i \} - \max_i \{\vithetau \beta_i' \}\big| + 
        \sumiton b_i \big| \log \beta_i -  \log \beta_i'\big|
        \\
        & \leq \vbar \| \beta - \beta'\|_\infty + \sumiton b_i \cdot \frac{1}{\ubarbetai/2} |\beta_i - \beta_i'|
        \\
        & = (\vbar + 2n) \| \beta - \beta'\|_\infty
        \;.
    \end{align*}
    This concludes the $(\vbar + 2n)$-Lipschitzness of $H_t$ on $C$. Similar argument goes through for $H$. 
    From the above reasoning we can also conclude $|H_t(\beta)  - H_t(\beta')|\leq \vbar \|\beta-\beta'\|_2 + 2 \|\beta - \beta'\|_1 \leq (\vbar + 2\sqrt{n})\|\beta - \beta'\|_2$. This concludes $(\vbar + 2\sqrt n)$-Lipschitzness of $H_t$ w.r.t.\ $\ell_2$-norm.

    Recall $H = \fbar + \Psi$ where $\fbar(\beta) = \E[\max_i\{ v_i(\theta) \betai \}]$ and $\Psi(\beta) = -\sumiton b_i \log \beta_i$. The function $\Psi$ is smooth with the first two derivatives 
    \begin{align*}
        \nabla\Psi(\beta) =  - [b_1/\beta_1, \dots, b_n/\beta_n]\tp, \quad \nabla\sq \Psi(\beta) =  \Diag( \{ b_i/(\betai)\sq \})
        \;.
    \end{align*}
    It is clear that for all $\beta \in C$ it holds $\betai\leq 2$. So 
    $\nabla\sq \Psi(\beta) \succ \min_i\{b_i/4\} I = \lambda I$.
    To verify the strong-convexity w.r.t\ $\|\cdot \|_\infty$ norm, we note
    for all $\beta', \beta \in C$.
    \begin{align*}
        H(\beta ')  - H(\beta) - \lg z + \nabla \Psi(\beta), \beta' - \beta  \rg \geq (\lambda/2) \|\beta' - \beta \|_2\sq \geq (\lambda/2) \|\beta' - \beta \|_\infty\sq
        \;,
    \end{align*}
    where $z \in \partial \fbar(\beta)$ and $z + \nabla \Psi(\beta) \in \partial H(\beta)$.
    This completes the proof.
\end{proof}

\begin{lemma}[Estimation of $\cov(\must)$]
    \label{thm:covnablaf_estimation}
    Let $\betast$ be a deterministic point in $\Rnpp$ and  
    $\betagam = \betast + o_p(1)$.
    Recall $f$ defined in \cref{eq:def:F} and $H$ in \cref{eq:def_pop_eg}.
    Suppose $H$ be differentiable on a neighborhood of $\betast$.
    Let $ \mu^\tau \in \partial_\b f(\thetau, \betagam)$ be any selection.
    Then 
    \begin{align}
        & \frac1t \sumtau \mutau  - \E[\nabla f(\t, \betast)] = o_p(1), 
        \\
        & \frac1t \sumtau \mutau (\mutau )\tp  -\E[ \nabla f(\t, \betast) \nabla f(\t, \betast) \tp ] = o_p(1)
    \end{align}
\end{lemma}
\begin{proof}[Proof of \cref{thm:covnablaf_estimation}]
    By the differentiability condition, there exists a neighborhood $N$ of $\betast$ so that 
    for all $\b \in N $ 
    it holds
    $\P(\t: f \text{ is differentiable at $\b$}) = 1$.
    Define $D_f: \Rnp \times \Theta \to \Rn$, 
    $D_{f,i} (\b, \t) = v_i(\theta)\prod_{k=1}^n \indi( \betai \vithe \geq \beta_k v_k(\theta) )$.
    Then $D_f = \nabla f$ if $f$ is differentiable at $\b$. 
    Moreover, for $\b \in N$, it holds $\P(\t: D_f(\t, \b) \text{ not continuous at $\b$}) = 0$.
    By Theorem 7.53 of~\citet{shapiro2021lectures} (a uniform law of large number result for continuous random functions), it holds 
    \begin{align}
        & \sup_{\b \in N} \big\| \frac1t\sumtau D_f(\thetau, \b) - \E[\nabla f(\t, \b)]\big \| = o_p(1),
        \\
        & \sup_{\b \in N} \| \frac1t\sumtau D_f(\thetau, \b) D_f(\thetau, \b)\tp - \E[\nabla f(\t, \b) \nabla f(\t, \b)\tp]\|_2 = o_p(1) 
    \end{align}
    By $\betagam \toprob \betast$, we know $\P(\betagam \in N) \to 1$.
    And under event $\{ \betagam \in N\}$, it must be that $\mutau = D_f( \thetau, \betagam,)$. The desired claim is proved.
\end{proof}

\begin{definition}[Definition 7.29 in~\citet{shapiro2021lectures}]
  \label{def:epiconvergence}
    A sequence $f_k:\Rn \to \bar \R$, $k=1,\dots$, of extended real valued functions epi-converge to a function $f:\Rn \to \bar \R$, if for any point $x\in \Rn$ the following conditions hold

    (1) For any sequence $x_k \to x$, it holds $\liminf _{k \rightarrow \infty} f_{k}\left(x_{k}\right) \geq f(x)$,

    (2) There exists a sequence $x_k \to x$ such that $\limsup _{k \rightarrow \infty} f_{k}\left(x_{k}\right) \leq f(x)$.
\end{definition}



\begin{definition}
    A function $f:\Rn \to \bar \R$ is level-coercive if $\liminf_{\|x\|\to \infty} f(x) / \|x\| > 0$. It is equivalent to $\lim _{\|x\| \rightarrow+\infty} f(x)=+\infty$.
    This is Definition 3.25, \citet{rockafellar2009variational}, see also Definition 11.11 and Proposition 14.16 from \citet{bauschke2011convex}
\end{definition}
\begin{lemma}[Corollary 11.13, \citet{rockafellar2009variational}]
    \label{lm:coercivity_level_boundedness}
    For any proper, lsc function $f$ on $\Rn$, level coercivity implies level boundedness. When $f$ is convex the two
properties are equivalent.
\end{lemma}

\begin{lemma}[Theorem 7.17,~\cite{rockafellar2009variational}] \label{lm:def_of_epiconv}
    Let $h_n: \R^d \to \bar \R$, $h:\R^d\to \bar \R$ be closed convex and proper. Then $h_n \toepi h$ is equivalent to either of the following conditions.

    (1) There exists a dense set $A \subset \R^d$ such that $h_n(v) \to h(v)$ for all $v\in A$.

    (2) For all compact $C \subset \dom h$ not containing a boundary point of $\dom h$, it holds 
    $$\lim_{n\to \infty} \sup_{v\in C} |h_n(v) - h(v)| = 0
    \;.
    $$
\end{lemma}

\begin{lemma}[Proposition 7.33,~\citet{rockafellar2009variational}] \label{lm:inf_conv}
    Let $h_n: \R^d \to \bar \R$, $h:\R^d\to \bar \R$ be closed and proper. If $h_n$ has bounded sublevel sets and $h_n \toepi h$, then $\inf_v h_n(v) \to \inf_v h(v)$.
\end{lemma}

\begin{lemma}[Theorem 7.31,~\citet{rockafellar2009variational}] \label{lm:appro_sol_conv}
    Let $h_n: \R^d \to \bar \R$, $h:\R^d\to \bar \R$ satisfy $h_n \toepi$ and $-\infty < \inf h < \infty$. Let $S_{n}(\varepsilon)=\left\{\theta \mid h_{n}(\theta) \leq \inf h_{n}+\varepsilon\right\}$ and $S(\varepsilon)=\left\{\theta \mid h(\theta) \leq \inf h+\varepsilon\right\}$. Then $\limsup_n S_n(\varepsilon) \subset S(\varepsilon)$ for all $\varepsilon \geq 0$, and $\limsup_n S_n(\varepsilon_n) \subset S(0)$ whenever $\varepsilon_n \downarrow 0$.

\end{lemma}

\begin{lemma}[Theorem 5.7,~\citet{shapiro2021lectures}, Asymptotics of SAA Optimal Value] \label{lm:clt_optimal_value}
    Consider the problem $$\min_{x\in X}  f(x) = \E[F(x,\xi)]$$ where $X$ is a nonempty closed subset of $\R^n$, $\xi$ is a random vector with probability distribution
    $P$ on a set $\Xi$ and $F:X\times \Xi \to \R$. Assume the expectation is well-defined, i.e., $f(x)<\infty$ for all $x\in X$. Define the sample average approxiamtion (SAA) problem $$\min_{x\in X}  f_N(x) = \frac1N \sum_{i=1}^N F(x,\xi_i)$$ where $\xi_i$ are i.i.d.\ copies of the random vector $\xi$. Let $v_N$ (resp., $v^*$) be the optimal value of the SAA problem (resp., the original problem). Assume the following.
\begin{enumconditions}
    \item The set $X$ is compact. \label{it:lm:clt_optimal_value:1}
    \item For some point $x\in X$ the expectation $\E[F(x,\xi)\sq]$ is finite. \label{it:lm:clt_optimal_value:2}
    \item There is a measurable function $C:\Xi \to \R_+$ such that $\E[C(\xi)\sq] < \infty$ 
    and $\left|F(x, \xi)-F\left(x^{\prime}, \xi\right)\right| \leq C(\xi)|| x-x^{\prime} \|$ for all $x, x' \in X$ and almost every $\xi \in \Xi$.
    \label{it:lm:clt_optimal_value:3}
    \item The function $f$ has a unique minimizer $x^*$ on $X$. \label{it:lm:clt_optimal_value:4}
\end{enumconditions}
Then $${v}_{N}=  {f}_{N}(x^*)+o_{p}(N^{-1 / 2})
\;, 
\quad \sqrt{N} ( v_N - v^*) \tod \cN\big(0, \var
\big({F(x^*, \xi)}\big) \big)
\;.
$$ 
\end{lemma}

\section{Proofs of Main Theorems of LFM}
\label{sec:proof:thm:consistency}

\subsection{Proof of Theorem~\ref{thm:consistency}}
\begin{proof}[Proof of \cref{thm:consistency}]

    We show epi-convergence (see \cref{def:epiconvergence}) of $H_t$ to $H$.
    Epi-convergence is closely related to the question of whether we have convergence of the set of minimizers. In particular, epi-convergence is a suitable notion of convergence under which one can guarantee that the set of minimizers of the sequence of approximate optimization problems converges to the minimizers of the original problem.

    To work under the framework of epi-convergence, we extend the definition of $H_t$ and $H$ to the entire Euclidean space as follows. 
    We extend $\log$ to the entire real by defining $\log(x) = -\infty$ if $x<0$.
    Let
    \begin{align*}
        \tilde F(\t,\b) = 
        \begin{cases}
            F(\t,\b) = \max_i v_i(\theta)\betai - \sumiton b_i \log \betai & 
            \text{if $\beta \in \Rnpp$}
            \\
            + \infty & \text{else}
        \end{cases},
    \end{align*}
    and
    \begin{align*}
            \tilde H(\beta): \R^n \to \bar \R,  \beta \mapsto \begin{cases} H(\beta) & \text{if $\beta \in \Rnpp$} \\
            +\infty & \text{else} \end{cases},
            \quad 
            \tilde H_t(\beta): \R^n \to \bar \R,  \beta \mapsto \begin{cases} H_t(\beta) & \text{if $\beta \in \Rnpp$} \\
            +\infty & \text{else} \end{cases}.
    \end{align*}
    It is clear that for $\beta \in \Rn$ it holds $\tilde H(\beta) = \E[\tilde F(\t,\b)]$ and $\tilde H_t(\beta) = \frac1t \sumtau \tilde F(\thetau, \beta)$.
    In order to prove the result, we will invoke Lemmas~\ref{lm:def_of_epiconv}, \ref{lm:inf_conv}, and \ref{lm:appro_sol_conv}.
    To invoke those lemmas, we will need the following four properties that we each prove immediately after stating them.

    \begin{enumerate}
        \item {Check that $ \tilde H$ is closed, proper and convex, and $\tilde H_t$ is closed, proper and convex almost surely.} Convexity and properness of the functions $ \tilde H_t$ and $\tilde H$ is obvious.
        Recall for a proper convex function, closedness is equivalent to lower semicontinuity \citep[Page 52]{rockafellar1970convex}. 
        It is obvious that $\tilde H_t$ is continuous and thus closed almost surely. 

        It remains to verify lower semicontinuity of $\tilde H$,
        i.e., for all $\beta \in \Rn$, $\liminf_{\beta' \to \beta}\tilde H(\beta') \geq \tilde H(\beta)$.
        For any $\beta \in \Rn$, we have that 
        $\tilde f(\t,\b) \defeq \max_i \vithe \beta_i  + \delta_{\Rnp}(\beta)\geq 0$, where $\delta_A(\beta) = \infty $ if $\beta\notin A$ and $0$ if $\beta\in A$. 
        With this definition of $\tilde f$ we have 
        $\tilde F(\t,\b) =\tilde f(\t,\b)  - \sumiton b_i \log \beta_i$.
        Applying Fatou's lemma (for extended real-valued random variables), we get 
        $\liminf_{\beta'\to\beta} \E[\tilde f(\theta, \beta') ]\geq \E[\liminf_{\beta' \to \beta} \tilde f(\theta, \beta') ] \geq  \E[\tilde f(\t,\b) ]$ where in the last step we used lower semicontinuity of $\beta \mapsto \tilde f(\t,\b) $.
        And thus 
        \begin{align*}
            & \liminf_{\beta'\to\beta} \tilde H(\beta') 
            \\
            & = \liminf_{\beta'\to\beta}\E \bigg[\tilde f(\theta, \beta')   - \sumiton b_i \log \beta'_i \bigg] 
            \\
            & \geq \liminf_{\beta'\to\beta}\E[ \tilde f(\theta, \beta')  ] - \sumiton b_i \log \betai 
            \\
            & \geq  \E \bigg[ \tilde f(\t,\b)   - \sumiton b_i \log \betai \bigg]
            \\
            & =  \tilde H(\beta) 
        \end{align*}
        This shows $\tilde H$ is lower semicontinuous.
        \label{it:consistency:1}
        
        \item {Check $\tilde H_t$ pointwise converges to $\tilde H$ on $\Q^n$.}
        Let $\Q^n$ be the set of $n$-dimensional vectors with rational entries.
        For a fixed $\beta \in \Rn$, define the event $E_\beta \defeq\{   \lim_{t\to \infty} \tilde H_t(\beta)  = \tilde H (\beta)\}$.
        Since $\vithetau \leq \vbar$ almost surely by assumption, the strong law of large numbers implies that $\P( E_\beta) = 1$. Define 
        \begin{align*}
            E\defeq \Big\{ \lim_{t\to \infty} \tilde H_t(\beta)  = \tilde H (\beta), \text{ for all } \beta \in \Q^n \Big\} = \bigcap_{\beta \in \Q^n} E_\beta.
        \end{align*}
        Then by a union bound we obtain $\P(E^c) = \P(\bigcup_{\beta \in \Q^n} E_\beta^c) \leq \sum_{\beta \in \Q^n} \P(E_\beta^c)= 0$, implying $E$
        has measure one.
        \label{it:consistency:2}

        \item {Check $-\infty < \inf_\beta \tilde H < \infty$.} This is obviously true since valuations are bounded.
        \label{it:consistency:3}

        \item {Check that for almost every sample path $\omega$, $\tilde H_t$ has bounded sublevel sets (eventually).} 
        By \cref{lm:coercivity_level_boundedness}, this property is equivalent to eventual coerciveness of $\tilde H_t$, i.e., there is a (random) $N$ such that for all $t \geq N$, it holds $\lim_{\|\beta \| \to \infty} \tilde H_t (\beta) = +\infty $.
        By \cref{lm:value_concentration}, we know for almost every $\omega$, there is a finite constant $N_\omega$ such that for all $t\geq N_\omega$ it holds $\vbarit \geq 1/2$. Then it holds for this $\omega$, all $t\geq N_\omega$, and all $\beta\in \Rn$, 
        \$
            \tilde H_t(\beta) 
            & = 
            \frac1t \sumtau \max_i \vithetau  \beta_i -\sumiton b_i\log \beta_i 
            \\
            & \geq \max_i (\vbarit \beta_i)-\sumiton b_i\log \beta_i  
            \\
            & \geq \frac12 \|\beta\|_\infty -\sumiton b_i\log \beta_i  \to +\infty \quad \text{as $\|\beta\|\to \infty$}
            \;.
        \$
        This implies $\tilde H_t$ has bounded sublevel sets.
        \label{it:consistency:4}

    \end{enumerate}

    With the above \cref{it:consistency:1} and \cref{it:consistency:2}  we invoke \cref{lm:def_of_epiconv} and obtain that 
    \#\P\big( \tilde H_t(\beta) \toepi \tilde H(\beta) \big) = 1
    \;,
    \label{eq:epi_conv}\# 
    and that the convergence is uniform on any compact set.

    The epi-convergence result \cref{eq:epi_conv} along with \cref{it:consistency:4} allows us to invoke \cref{lm:inf_conv} and obtain 
    \#\inf_{\beta \in \Rn} \tilde H_t(\beta) \to \inf_{\beta \in \Rn} \tilde H(\beta) \text{ a.s.} \label{eq:min_conv}\#
    which also implies $\inf_{\Rnpp} H_t \to \inf_\Rnpp H$ a.s.

    With the epi-convergence result \cref{eq:epi_conv} along with \cref{it:consistency:3} we invoke \cref{lm:appro_sol_conv} and obtain 
    \begin{equation} \label{eq:containment}
        \begin{split} 
            \limsup _{t} \cB^\gam(\epsilon) \subset \cB^*(\epsilon) \text { for all } \epsilon \geq 0
            \;,\\
        \limsup _{t} \cB^\gam(\epsilon_t) \subset \cB^*(0) \text { for all } \epsilon_t \downarrow 0
        \;.
        \end{split}
    \end{equation}

    \underline{Putting together.}
    At this stage all statements in the theorem are direct implications of the above results. 
    
    \emph{Proof of \cref{it:thm:consistency:1}}
    
    Convergence of Nash social welfare follows from \cref{eq:min_conv} and strong duality, i.e., $\LNSW^\gam = \inf_{\beta \in \Rnpp } H_t(\beta) + \sumiton (b_i\log b_i - b_i)$ and $\LNSW^* = \inf_{\beta \in \Rnpp } H(\beta)+\sumiton (b_i\log b_i - b_i)$. 

    \emph{Proof of \cref{it:thm:consistency:2}}
    
    Now we show consistency of the pacing multiplier via 
    \cref{lm:def_of_epiconv} and \cref{lm:inf_conv}. 
    Recall the compact set $C = \prod_{i=1}^n  [\ubarbetai/2, 2\betabar] = \prod_{i=1}^n [b_i/2, 2] \subset \R^n$. 
    By construction, $\betast \in C$. 
    First note that for almost every sample path~$\omega$, $1/2 \leq \vbarit \leq 2$ eventually, and thus $\beta^\gamma_i = b_i / u^\gam_i \leq b_i / (b_i \vbarit ) \leq 2$ and $\beta^\gam_i \geq b_i/2$ eventually. 
    So $\betagam \in C$ eventually.
    Now we can invoke \cref{lm:def_of_epiconv} Item (2) to get 
    \begin{align}
        \lim_{t\to\infty}\sup_{\beta \in C} | H_t(\beta) - H(\beta)| \to 1 \quad\text{a.s.}
    \label{eq:as conv ht h}
    \end{align}
    Now we can show that the value of $H$ on the sequence $\beta^\gam$ converges to the value at $\beta^*$:
    \$ 0\leq \lim_{t \to \infty}  H(\beta^\gam) - H(\betast) = \lim_{t \to \infty} [H (\beta^\gam) - H_t(\beta^\gam)] + \lim_{t \to \infty} [H_t(\beta^\gam) - H(\betast)] = 0
    \;.
    \$
    Here the first term tends to zero due to \eqref{eq:as conv ht h},
    and the second term by \cref{eq:min_conv}. 
    For any limit point of the sequence $\{\beta^\gam\}_t$, $\beta^\infty$, by 
    lower semicontinuity of $H$,
    \$ 0 &\leq H(\beta^\infty) - H(\betast) \leq \liminf_\ttinf H(\beta^\gam) - H(\betast) = 0
    \;.
    \$
    So it holds that $H(\beta^\infty) = H(\betast)$ for all limit points $\beta^\infty$. By uniqueness of the optimal solution $\betast$, we have $\beta^\gam \to \betast$ a.s. 
    
    \emph{Proof of \cref{it:thm:consistency:4}}
    
    Convergence of approximate equilibrium follows from \cref{eq:containment}.
    \end{proof}

\subsection{Proof of Theorem~\ref{thm:lnsw_concentration}}
\label{sec:proof:thm:lnsw_concentration}
\begin{proof}[Proof of \cref{thm:lnsw_concentration}]

    We abbreviate
    $C_\LFM$ to $C$, and recall 
    its definition $C 
    = \prod_{i=1}^n [b_i/2, 2] \subset \R^n$ and the normalization
    $\sumiton b_i = 1$ and $\nu_i = 1$ for all $i$.
    Recall the event $A_t = \{\betagam \in C \}$. 
    By \cref{lm:value_concentration} we know that if $t \geq 2 {\vbarsq {\log(4n/\eta)}}$ then 
    event $A_t$ happens
    with probability $\geq 1-\eta/2$. 
    Now the proof proceeds in two steps.

    \paragraph{Step 1. A covering number argument. }
    Let $\cB^o$ be an $\epsilon$-covering of the compact set $C$, i.e, for all $\beta \in C$ there is a $\beta^o(\beta) \in \cB^o$ such that $\| \beta - \beta^o(\beta) \|_\infty \leq \epsilon$. It is easy to see that such a set can be chosen with cardinality bounded by $|\cB^o|\leq (2/\epsilon)^n$. 

    Recall $H_t$ and $H$ are $L$-Lipschitz w.r.t.\ $\ell_\infty$-norm on $C$.
    Using this fact we get the following uniform concentration bound over the compact set $C$. 
    \begin{align*}
        & \sup_{\beta \in C} |H_t (\beta) - H(\beta)|
        \\
        & \leq \sup_{\beta \in C} \big\{  |H_t(\beta) - H_t(\beta^o(\beta))|
        + |H(\beta) - H(\beta^o(\beta))|
        + |H_t(\beta^o(\beta)) - H(\beta^o(\beta))|\big\}
        \\
        & \leq 2(\vbar + 2n) \epsilon + \sup_{\beta^o \in \cB^o}|H_t(\beta^o) - H(\beta^o)|
        \;.
    \end{align*}

    Next we bound the second term in the last expression. For some fixed $\beta \in C$, let $X^\tau \defeq \max_i \vithetau \beta_i$ and let its mean be $\mu$. Note $0 \leq X^\tau  \leq \vbar \|\beta\|_\infty \leq 2\vbar$ due to $\beta \in C$. So $X^\tau$'s are bounded random variables. By Hoeffding's inequality we have
    \begin{align*}
        \P\big(|H_t(\beta) - H(\beta)| \geq \delta\big)
        = \P\bigg( \Big|\frac1t \sumtau X^\tau - \mu\Big| \geq \delta\bigg)
        \leq 2\exp\bigg(-\frac{\delta\sq t }{2\vbarsq}\bigg)
        \;.
    \end{align*}
    By a union bound we get
    \begin{align*}
        \P\bigg(\sup_{\beta^o \in \cB^o}|H_t(\beta^o) - H(\beta^o)| \geq \delta\bigg)
        \leq 2|\cB^o| \exp\bigg(-\frac{\delta\sq t }{2\vbarsq}\bigg) \leq 2\exp\bigg(-\frac{\delta \sq t }{2\vbarsq} + n\log(2/\epsilon)\bigg)
        \;.
    \end{align*}
    Define the event 
    \begin{align}
        \label{eq:def:eventEt}
    E_t \defeq \Big\{\sup_{\beta^o \in \cB^o}|H_t(\beta^o) - H(\beta^o)| 
    \leq \frac{2\vbar}{\sqrt{t}} \sqrt{\log(4/\eta) + n\log(2/\epsilon)} 
    =: \iota \Big\} 
    \;.
    \end{align}
    By setting $2\exp(-{\delta \sq t }/{(2\vbarsq)} + n\log(2/\epsilon)) = \eta/2$ and solving for $\eta$, we have that $\P(E_t) \geq 1-\eta/2$.

    \paragraph{Step 2. Putting together. }
    Recall the event $A_t = \{ \betagam \in C   \}$.
    Now let events $A_t$ and $E_t$ hold. Note $\P(A_t \cap E_t) \geq 1-\eta$ if $t \geq 2 {\vbarsq {\log(4n/\eta)}}$. Then 
    \begin{align}
        & \Big| \sup_{\beta\in\Rnpp} H_t(\beta) - \sup_{\beta\in\Rnpp} H(\beta) \Big| 
        \notag
        \\
        & = \Big| \sup_{\beta\in C} H_t(\beta) - \sup_{\beta\in C} H(\beta)\Big|  
        \notag
        \\
        & \leq \sup_{\beta \in C} |H_t(\beta)-H(\beta)| 
        \notag
        \\
        & \leq 2(\vbar + 2n) \epsilon + \iota 
        \;,
        \label{eq:08251314}
    \end{align}
    where the first equality is due to event $A_t$ and the last inequality is due to event $E_t$ defined in \cref{eq:def:eventEt}.
    Now we choose the discretization error as $\epsilon = \frac{1}{\sqrt{t} (\vbar + 2n)}$. 
    Then, the expression in \cref{eq:08251314} can be upper bounded as follows.
    \begin{align*}
        & 2(\vbar + 2n) \epsilon + \iota 
        \\
        & = \frac{2}{\sqrt{t}} + 
        \frac{2\vbar}{\sqrt t}  \sqrt{\log(4/\eta) + {n} {\log(2\sqrt{t} (\vbar + 2n))}}
        \;.
    \end{align*}
    This completes the proof.
\end{proof}

\subsection{Proof of Theorem~\ref{thm:high_prob_containment}}
\label{sec:proof:thm:high_prob_containment}
    \begin{proof}[Proof of \cref{thm:high_prob_containment}]

    The proof idea of this theorem closely follows Section 5.3 of~\citet{shapiro2021lectures}.

    We first need some additional notations. Define the approximate solutions sets of surrogate problems as follows: For a closed set $ A\subset \Rnpp$, let 
    \begin{align*}
        \cB^*_A(\epsilon) &\defeq \{ \beta \in A: H(\beta) \leq \min_A H  +\eps  \}
        \;,
        \\
        \cB^\gam_A(\epsilon) &\defeq\{ \beta \in A: H_t(\beta) \leq \min_A H_t  +\eps \}
        \;.
    \end{align*}

    In words, they solve the surrogate optimization problems which are defined with a new constraint set $A$. Note that if $\betast \in A$ then $\cB^*_A(\epsilon) = A\cap \cB^*(\eps)$. Recall on the compact set $C$, both $H_t$ and $H$ are $L$-Lipschitz and $\lambda$-strongly convex w.r.t\ the $\ell_\infty$-norm, where $L = (\vbar + 2n)$ and $\lambda = \ubar{b}/4$.

    Let $r \defeq \sup\{ H(\beta) - H^*: \beta \in C \}$. Then if $\epsilon \geq r$ then $C\subset \cBst(\epsilon)$ and the claim is trivial. Now we assume $\epsilon < r$. 
    
    Define $a = \min\{2\epsilon, (r + \eps)/2 \}$. Note $\eps < a < r$. Define $S = C\cap \cBst(a)$. The role of $S$ will be evident as follows. We will show that, with high probability, the following chain of inclusions holds 
    \begin{align*}
        \cBgamC(\delta) \overset{(1)}{\subset} \cBgamS (\delta) \overset{(2)}{\subset} \cBstS(\epsilon) \overset{(3)}{\subset} \cBstC(\epsilon) \, .
    \end{align*}
    
    \textbf{Step 1. Reduction to discretized problems.}
    We let $S'$ be a $\nu$-cover of the set $S = \cBst(a) \cap C$.
    Let $X = S' \cup \{\betast \}$. 
    In this part the goal is to show 
    \begin{align*}
        \P\big(     \cBgamC(\delta)\subset \cBstC(\eps)
        \big) 
        \geq 
        \P\big(\cB^\gam_X(\delta')  \subset \cB^*_X(\epsilon') \big)
    \end{align*}
    where 
    \begin{align*}
        \nu = (\eps' -\delta')/4 > 0 \,, \quad \delta' =  \delta + L\nu > 0
        \,, \quad \epsilon' = \epsilon - L\nu >0
        \,.
    \end{align*}

    First, we claim
\begin{claim} \label{claim:1}
It holds    $
    \cB^\gam_X(\delta')  \subset \cB^*_X(\epsilon') 
    \implies \cBgamS (\delta) \overset{}{\subset} \cBstS(\epsilon)
    $ (Inclusion (2)).
\end{claim}

    Next, we show 
    \begin{claim}\label{claim:2}
        Inclusion (2) implies Inclusion (1):
    $  \cBgamS (\delta) \overset{}{\subset} \cBstS(\epsilon)
    \implies \cBgamC(\delta) \overset{}{\subset} \cBgamS (\delta)\,.$   
    \end{claim}

    Proofs of \cref{claim:1} and \cref{claim:2} are deferred after the proof of \cref{thm:high_prob_containment}. At a high level, \cref{claim:1} uses the covering property of the set $X$. \cref{claim:2} exploits convexity of the problem.

    Finally, we show Inclusion (3) $ \cBstS(\epsilon) \overset{}{\subset} \cBstC(\epsilon)$.  
    Note that $\betast$ belongs to both $C$ and $S$. And thus for any $\beta \in \cBstS (\epsilon)$, it holds $H(\beta) \leq \min_X H + \eps = \Hst + \eps = \min_S H + \eps \,.$ We obtain $\beta \in \cBstC(\epsilon)$.

    To summarize, \cref{claim:1} shows that $\cB^\gam_X(\delta')  \subset \cB^*_X(\epsilon') $ implies Inclusion (2).
    Inclusion (3) holds automatically.
    By \cref{claim:2} we know Inclusion (2) implies Inclusion (1). So it holds deterministically that 
    \begin{align*}
        \{\cB^\gam_X(\delta')  \subset \cB^*_X(\epsilon') \} 
        \subset \{ \cBgamC(\delta)\subset  \cBstC(\eps) \} 
        \;.
    \end{align*}

    \textbf{Step 2. Probability of inclusion for discretized problems.}
    Now we bound the probability $\P(\cB^\gam_X(\delta')  \subset \cB^*_X(\epsilon') )$.

    For now, we forget the construction $X = S' + \{ \betast\}$ where $S'$ is a $\nu$-cover of $S$.
    Let $X\subset C$ be any discrete set with cardinality $|X|$.

    Let $\betastX \in \argmin_X H$ be a minimizer of $H$ over the set $X$. For $\beta \in X$ define the random variable $Y^\tau_\beta \defeq F(\thetau, \betastX) - F(\thetau, \beta)$.  Also let $\mu_\beta \defeq \E[Y^\tau_\beta]$, which is well-defined by the i.i.d.\ item assumption. Let $D\defeq \sup_{\beta \in X} \| \beta - \betastX\|_\infty$.

    Consider any $0 \leq \delta ' < \epsilon'$. If $X- \cB^*_X(\epsilon')$ is empty, then all elements in $X$ are $\epsilon'$-optimal for the problem $\min_X H$. Next assume $X- \cB^*_X(\epsilon')$ is not empty. We upper bound the probability of the event $\cB^\gam_X(\delta') \not \subset \cB^*_X(\epsilon')$.
    \begin{align}
        &\P \big(\cB^\gam_X(\delta') \not \subset \cB^*_X(\epsilon')\big)
        \notag
        \\
        &= \P \big(\text{there exists } \beta \in X- \cB^*_X(\epsilon'),\,  H_t(\beta)\leq H_t(\betastX) + \delta'\big)
        \notag
        \\
        &\leq \sum_{\beta \in X- \cB^*_X(\epsilon')} \P \big(H_t(\beta)\leq H_t(\betastX) + \delta'\big)
        \notag
        \\
        &= \sum_{\beta \in X- \cB^*_X(\epsilon')} \P\bigg(\frac1t \sumtau Y^\tau_\beta \geq - \delta'\bigg)
        \notag
        \\
        &\leq \sum_{\beta \in X- \cB^*_X(\epsilon')} \P\bigg(\frac1t \sumtau Y^\tau_\beta - \mu_\beta \geq \eps' - \delta'\bigg) \tag{A}
        \\
        &\leq \sum_{\beta \in X- \cB^*_X(\epsilon')} \exp\Big(-\frac{2t (\eps'-\delta')\sq}{L_f\sq \|\beta - \betast \|_\infty\sq}\Big) \tag{B}
        \\ 
        &\leq |X| \exp\Big(-\frac{2t (\eps'-\delta')\sq}{L_f\sq \|\beta - \betastX \|_\infty\sq}\Big) \,.
        \label{eq:bound_discretization}
    \end{align}
    Here in (A) we use the fact that $\mu_\beta = H(\betastX) - H(\beta) > - \eps' $ for $\beta \in X- \cB^*_X(\epsilon')$. In (B), using $L_f$-Lipschitzness of $f$ on the set $C$, we obtain $|Y^\tau_\beta - \mu_\beta| \leq 2L_f \|\beta - \betastX \|_\infty$ and then apply Hoeffding's inequality for bounded random variables. Setting \cref{eq:bound_discretization} equal to $\alpha$ and solving for $t$, we have that if 
    \begin{align}\label{eq:t_for_discrete_X}
        t \geq c' \cdot \frac{L_f\sq D\sq}{(\eps' - \delta')\sq } \Big(\log|X| + \log\frac1\alpha\Big) \,,
    \end{align}
    then $\P\big(\cB^\gam_X(\delta') \not \subset \cB^*_X(\epsilon')\big) \leq \alpha$. Note the above derivation applies to any finite set $X\subset S$.

    Now we use the construction $X= S' + \{ \betast\}$.
    Then the cardinality of $X$ can be upper bounded by $(4/\nu)^n$. Note since $\betast \in X$ it holds $\betast = \betastX$.
    We apply the result in \cref{eq:t_for_discrete_X} with the following parameters
    \begin{align*}
        & \nu = (\eps' -\delta')/(4L)
        \,, \quad 
        \delta' =  \delta + L\nu
        \,, \quad 
        \epsilon' = \epsilon - L\nu
        \,, \quad 
        \epsilon' - \delta' = \frac12(\epsilon - \delta)\,,
    \\
        & D =  \min \{\sqrt{2a/\lambda} , 2 \}
        \,, 
        \quad |X| \leq 
        \Big(\frac{16L}{\eps - \delta}\Big)^n 
        \,.
    \end{align*}
    We justify the choice of $D$.   First, $S \subset C$ implies $D \leq 2$.
    By the $\lambda$-strong convexity of $H$ on $C$: for all $\beta \in X\subset S \subset \cBst(a)$, it holds
    \begin{align*}
        &(1/2) \lambda \|\beta - \betastX\|^2_\infty = (1/2) \lambda \|\beta - \betast\|^2_\infty \leq H(\beta) - \Hst \leq a
        \\
        \implies 
        & D = \sup_{\beta\in X}\|X - \betastX\|_\infty \leq \sqrt{2a/\lambda} \,. 
    \end{align*}
    Substituting these quantities into the bound \cref{eq:t_for_discrete_X} the expression becomes
    \begin{align*}
        t 
        & \geq  c'\cdot \frac{ L_f\sq }{ (\eps -\delta)\sq } 
        \cdot 
        \min \bigg\{\frac{2a}{\lambda },4\bigg\} 
        \cdot 
        \bigg(n\log\Big(\frac{16 L }{\eps - \delta}\Big) + \log \frac1\alpha\bigg) 
        \,.
    \end{align*}
    Here $c'$ is an absolute constant that changes from line to line.
    Moreover, noting that $a \leq 2\epsilon$ and $\delta \leq \epsilon/2$ implies $a / (\epsilon - \delta)\sq \leq 8/\epsilon$, we know that if 
    \begin{align} \label{eq:final_t_bound}
        t & \geq c'\cdot  L_f \sq  \min \bigg\{ \frac{1}{ \lambda \epsilon} , \frac{1}{\epsilon\sq} \bigg\} \cdot \bigg(n\log\Big(\frac{16L}{\eps - \delta}\Big) + \log \frac1\alpha\bigg)
        \;,
    \end{align}
    then
    $\P\big(\cB^\gam_X(\delta')  \subset \cB^*_X(\epsilon')\big) \geq 1- \alpha$. 
    By plugging in $L_f = \vbar $, $L = (2n+\vbar)$ and $\lambda = \ubar{b} / 4$, we know
    $\P \Big( \cBgamS (\delta) \overset{}{\subset} \cBstS(\epsilon)\Big) \geq 1-\alpha $ as long as 
    \begin{align*}
        t \geq c'\cdot  \vbar \sq  \min \bigg\{ \frac{1}{ \ubar{b}\epsilon} , \frac{1}{\epsilon\sq} \bigg\} \cdot \bigg(n\log\Big(\frac{16 (2n+\vbar) }{\eps - \delta}\Big) + \log \frac1\alpha\bigg) 
        \;.
    \end{align*}

    \textbf{Step 3. Putting together.}
    By \cref{lm:value_concentration},
    if $t \geq 2\vbarsq {\log(2n/\alpha)}$ then $\betagam \in C$ with probability $\geq 1 - \alpha$. Under the event $\betagam \in C$, it holds $\cBgamC(\delta) = C \cap \cBgam(\delta)$. Since $\betast\in C$ it holds that $\cBstC(\eps) = C\cap \cBst(\eps)$.
    Moreover, if $t$ satisfies the bound in \cref{eq:final_t_bound}, we know Inclusion (2) holds with probability $\geq 1- \alpha$, which then implies Inclusion (1). 
    So if $t$ satisfies the two requirements, $t \geq 2\vbarsq {\log(2n/\alpha)}$ and \cref{eq:final_t_bound}, then with probability $\geq 1- 2\alpha$,
    \begin{align*}
        C \cap \cBgam(\delta) = \cBgamC(\delta ) \subset \cBstC(\eps) = C\cap \cBst(\eps)\,.
    \end{align*}

\end{proof}

\begin{proof}[Proof of \cref{claim:1}]
    To see this, for $\beta \in \cBgamS(\delta)$ let $\beta' \in X$ be such that $\|\beta - \beta' \|_\infty \leq \nu$. By Lipschitzness of $H_t$ on $C$, we know
    \begin{align*}
        H_t(\beta')
        & \leq H_t(\beta)+L\nu \tag{Lipschitzness of $H_t$}
        \\
        &\leq \min_S H_t + \delta + L\nu  \tag{$\beta \in \cBgamS(\delta)$}
        \\ 
        &\leq \min_X H_t + \delta + L\nu \tag{$X\subset S$} 
        \\
        & =  \min_X H_t + \delta'\,.
    \end{align*}
    This implies the membership $\beta' \in \cBgamX(\delta')$. Furthermore, we have 
    \begin{align*}
        \cBgamX(\delta') \subset \cBstX(\eps') \subset \cBstC (\epsilon') \,.
    \end{align*} 
    Here the first inclusion is simply the assumption that $\cB^\gam_X(\delta')  \subset \cB^*_X(\epsilon')$. The second inclusion follows by the construction of $X$; since $\betast \in X$, we know $\cBstX(\epsilon') \subset \cBstC (\epsilon')$ and thus $\min_X H = \min_X H= \Hst$.  We now obtain
    \begin{align*}
        \beta' \in \cBstC (\epsilon') \,.
    \end{align*}

    Using the Lipschitzness of $H$ on $C$, we have for all $\beta \in \cBgamS(\delta)$
    \begin{align*}
        H(\beta) &\leq H(\beta') + L\nu \tag{Lipschitzness of $H$}
        \\
        &\leq \min_C H + \epsilon' + L\nu \tag{$\beta' \in \cBstC (\epsilon') $}
        \\
        & = \min_C H + \epsilon \,.
    \end{align*}
    So we conclude $\beta \in \cBstC(\epsilon)$, implying $\cBgamS(\delta) \subset  \cBstC(\epsilon)$.     This completes the proof of \cref{claim:1}.
\end{proof}

\begin{proof}[Proof of \cref{claim:2}]

    This claim relies on convexity of the problem.
    
        Assume, for the sake of contradiction, there exists $\betadia \in \cBgamC(\delta)$ but $\betadia \not \in \cBgamS(\delta)$. The only possibility this can happen is $\betadia \in C$ but $\betadia \not \in S = C\cap \cBst(a)$. So $\betadia \not \in \cBst(a)$ (note $a < r$ implies the set $C- \cBst(a)$ is not empty), which by definition means 
        \begin{align}\label{eq:betadia_gt_a}
            H(\betadia) - \Hst > a \,.
        \end{align}
        Now define 
        \begin{align*}
            \betabar =  \argmin_{\beta \in S} H_t(\beta) \in \cBgamS(\delta) \,.
        \end{align*}
        By the assumption $\cBgamS (\delta) \overset{}{\subset} \cBstS(\epsilon)$, we know $\betabar \in \cBstS(\eps)$ and so
        \begin{align}\label{eq:betabar_leq_eps}
            H(\betabar) - \Hst \leq \epsilon \,.
        \end{align}
        
        Next, let $\beta^c = c \betabar + (1-c) \betadia$ with $c\in [0,1]$, which is a point lying on the line segment joining the two points $\betabar$ and $\betadia$. By the optimality of $\betadia \in \cBgamC(\delta)$ and $\betabar \in C$, we know $H_t(\betadia) \leq H(\betabar) + \delta$.
        By convexity of $H_t$, we have for all $c\in[0,1]$,
        \begin{align}\label{eq:line_segment_implications}
            H_t(\beta^c) \leq \max\{ H_t(\betabar), H_t(\betadia)\} \leq H_t(\betabar) + \delta \,.
        \end{align} 
    
        Now consider the map $K: [0,1] \to \R_+, c \mapsto H(\beta^c) - \Hst$. Since any convex function is continuous on its effective domain~\cite[Corollary 10.1.1]{rockafellar1970convex}, we know $H$ is continuous. Continuity of $H$ implies continuity of $K$.
        Note $K(0) = H(\betadia) - \Hst > a$ by \cref{eq:betadia_gt_a} and $K(1) = H(\betabar) - \Hst \leq \epsilon$ by \cref{eq:betabar_leq_eps}. By intermediate value theorem, there is $c^* \in [0,1]$ such that $\eps < H(\beta^{c*}) - \Hst < a$. Moreover, by $H(\beta^{c*}) - \Hst < a$ and $\beta^{c*} \in C$ we obtain $\beta^{c*} \in S = \cBst(a) \cap C$. In addition, recalling $H_t(\beta^{c*}) \leq H_t(\betabar) + \delta$ (\cref{eq:line_segment_implications}), we conclude by definition $\beta^{c*}\in \cBgamS(\delta)$.
        
        At this point we have shown the existence of a point $\beta^{c*}$ such that
        \begin{align*}
            \beta^{c*} \in \cBgamS(\delta) \,, \quad \beta^{c*} \not \in \cBst(\eps)\,.
        \end{align*}
        This clearly contradicts the assumption $\cBgamS (\delta) \overset{}{\subset} \cBstS(\epsilon) = \cBst(\epsilon) \cap S$.
        This completes the proof of \cref{claim:2}.
    
    \end{proof}

\begin{proof}[Proof of \cref{cor:H_concentration}]
    Under the event $\{ \betagam \in C\}$, the set $C \cap \cBgam (0) = \{ \betagam \}$. 
    Moreover, $\betagam \in C \cap \cBst (\epsilon)$ implies $H(\betagam) \leq H(\betast) + \epsilon$. This completes the proof.
\end{proof}

\begin{proof}[Proof of \cref{cor:beta_u_concentration}]
    Under the event $\{ \betagam \in C\}$, we use strong convexity of $H$ over $C$ w.r.t.\ $\ell_2$-norm and obtain $\frac{\lambda}{2} \| \betagam - \betast\|_2 \sq \leq H(\betagam) - H(\betast)$ where $\lambda = \ubar{b}/4$ is the strong-convexity parameter. 

    For the second claim we use the equality $\betagami = b_i / \ugami$ and $\betasti = b_i / \usti$. 
    For $\beta, \beta' \in C$, it holds $|\frac{1}{\betai} - \frac{1}{\beta'_i} | \leq \frac{4}{{b_i}^2} | \betai - \beta'_i|$.
    And so $\| \ugam - \ust \|_2 = \sumi (b_i)\sq (\frac{1}{\betagami} - \frac{1}{\betasti})\sq \leq \sumi \frac{16}{(b_i)\sq} | \betagami - \betasti|\sq \leq \frac{16}{(\ubar{b})\sq} \| \betagam - \betast\|_2\sq$. 
    So we obtain $\|\ugam -\ust\|_2 \leq \frac{4}{\ubar{b}} \|\betagam - \betast\|_2$.
    We complete the proof.
\end{proof}

\subsection{Proof of Theorems~\ref{thm:normality} and \ref{thm:clt_beta_u}} 
\label{sec:proof:thm:normality}
\begin{proof}[Proof of \cref{thm:normality}]
    By strong duality of EG programs, $\NSWgam - \NSWst = H_t(\betagam) - H\st(\betast)$.
    Denote $H_t(\betagam)$ by $H^\gam$ and $H(\betast)$ by $H^*$.
    We aim to apply \cref{lm:clt_optimal_value} to our problem. To do this we first introduce surrogate problems 
    \begin{align*}
        H^\gam_C \defeq \min_{\beta \in C} H_t(\beta) 
        \;,\quad  
        H^*_C   \defeq \min_{\beta \in C} H(\beta)
        \;.
    \end{align*}
    Since $\betast \in C$ we know $H^*_C = \Hst$. We write down the decomposition
    \begin{align*}
        \sqrt{t} (H^\gam - \Hst) = \sqrt{t} (H^\gam -  H^\gam_C ) + \sqrt{t}( H^\gam_C  - H^*_C)
        \;.
    \end{align*}

    For the first term we show that $\sqrt{t} (H^\gam -  H^\gam_C ) \toprob 0$. Choose any $\eps > 0$, define the event $A^\eps_t = \{ \sqrt{t} |H^\gam - H^\gam_C | \geq \eps\}$. By \cref{lm:value_concentration} we know that with probability 1, $\betagam \in C$ eventually and so $H^\gam - H^\gam_C = 0$ eventually. This implies $\P( (A^\epsilon_t)^c \text{ eventually}) = 1 \Leftrightarrow$ 
    $\P(A^\eps_t \text{ infinitely often}) = 0$. By Fatou's lemma, 
    $
        \limsup_{t\to\infty} \P (A^\eps_t) \leq \P (\limsup_{t\to\infty}A^\eps_t ) = 0
        .
    $ 
    We conclude for all $\epsilon > 0$, $\lim_{t\to\infty} \P(\sqrt{t} | H^\gam -  H^\gam_C | > \epsilon)=0$.

    For the second term, we invoke \cref{lm:clt_optimal_value} and obtain $\sqrt{t}( H^\gam_C  - H^*_C) \tod \cN(0,\var[F(\theta, \betast)])$, where we recall $F(\t,\b) = \max_i \betai v_i(\theta) - \sumiton b_i \log \beta_i$.
    To do this we verify all hypotheses in \cref{lm:clt_optimal_value}.
    \begin{itemize}
        \item The set $C$ is compact and therefore \cref{it:lm:clt_optimal_value:1} is satisfied. 
        \item The function $F$ is finite for all $\beta \in \Rnpp$ and thus \cref{it:lm:clt_optimal_value:2} holds. 
        \item The function $F(\cdot, \theta)$ is $(2n+\vbar)$-Lipschitz on $C$ for all $\theta$, and thus \cref{it:lm:clt_optimal_value:3} holds. 
        \item \cref{it:lm:clt_optimal_value:4} holds because the function $H$ has a unique minimizer over $C$. 
    \end{itemize}
Now we calculate the variance term. 
    \begin{align*}
        \var (F(\theta, \betast)) = \var(f(\t, \betast))  = \var (\max_i  v_i(\theta) \betasti )
        = \var (p^*(\theta) )
    \end{align*}
    By Slutsky's theorem, we obtain the claimed result.

\end{proof}

\begin{proof}[Proof of Theorem~\ref{thm:clt_beta_u}]

    We verify all the conditions in Theorem 2.1 from~\citet{hjort2011asymptotics}.
    This theorem is handy since it uses convexity and avoids verifying stochastic equicontinuity of certain processes.

    Because $H$ is $C^2$ at $\betast$, there exists a neighborhood $N$ of $\betast$ 
    such that $H$ is continuously differentiable on $N$. By \cref{thm:first_differentiability} this implies that 
    the random variable $\bidgap(\beta,\cdot)\inv$ is finite almost surely for each $\beta \in N$. This implies $I(\beta,\theta)$ is single valued a.s.\ for $\beta \in N$.

    Define $$D(\theta)\defeq \nabla F(\theta, \betast) = \nabla f(\theta, \betast) - \nabla \Psi(\betast)$$ where we recall the subgradient $\nabla f(\theta, \betast) = e_{i(\betast,\theta)}v_{i(\betast,\theta)}$ and $i(\betast,\theta) = \argmax_i \betasti v_i(\theta)$ is the winner of item $\theta$ when the pacing multiplier of buyers is $\betast$. 
    By optimality of $\betast$ we know $\nabla H (\betast)=\E[D(\theta)] = 0$. Moreover, by twice differentiability of $H$ at $\betast$, the following expansion holds: 
    \begin{align*}
        H(\betast + h) - H(\betast) = \frac12 h \tp \big(\nabla \sq H(\betast)\big) h + o(\| h\|_2\sq)
        \;.
    \end{align*}
    
    For any $h \to 0$, define
    $$R(\theta) \defeq (F(\theta, \betast + h) - F(\theta, \betast) - D(\theta) \tp h) / \|h \|_2$$ measure the first-order approximation error.
    To invoke Theorem 2.1 from~\citet{hjort2011asymptotics}, we check the following stochastic version of differentiability condition holds 
    \begin{align}
        \E [R(\theta, h) \sq ] = o(1) \quad \text{as } 
        h \to 0
        \;.
        \label{eq:diff_of_D}
    \end{align} 
    By $H$ being differentiable at $\betast$, 
    we know $R(\t, h)  \toas 0$. Since we assume $\max_i \esssup v_i(\t) < \infty$, we know the sequence of random variables $R(\t, h)$ is bounded. We conclude \cref{eq:diff_of_D}.

    At this stage we have verified all the conditions in Theorem 2.1 from~\citet{hjort2011asymptotics}. Invoking the theorem we obtain
    \begin{align*}
        \sqrt{t} (\betagam - \betast) = - [\nabla \sq H(\betast)]\inv \Bigg( \frac{1}{\sqrt t}\sumtau D(\thetau) \Bigg) + o_p(1)
        \;.
    \end{align*} 
    In particular,  $\sqrt{t} (\betagam - \betast) \tod \cN (0, [\nabla \sq H(\betast)]\inv \cov(D) [\nabla \sq H(\betast)]\inv )$.
    Finally, note $\cov(D) = \cov(\must)$.

    \emph{Proof of Asymptotic Distribution for $\beta$.}
    This follows from the discussion above.

    \emph{Proof of Asymptotic Distribution for $u$.}  We use the delta method. Take $g(\beta) = [b_1/\beta_1,\dots, b_n/\beta_n]$. Then the asymptotic variance of $\sqrt{t}(g(\betagam) - g(\betast))$ is $\nabla g(\betast) 
    \tp 
    \Sigma_\beta
    \nabla g (\betast)$. Note $\nabla g(\betast)$ is the diagonal matrix $\Diag(\{-b_i/\betasti\sq \})$. 
    From here we obtain the expression for $\Sigma_u$.

\end{proof}

\subsection{Proof of \cref{thm:nsw_aym_risk}}
\label{sec:proof:thm:nsw_aym_risk}
\begin{proof}[Proof of \cref{thm:nsw_aym_risk}]

\textit{Pacing multiplier $\b$.}
The lower bound result for $\b$ is an immediate application of Theorem 1 from \citet{duchi2021asymptotic}. 

\textit{NSW and utility.}
Based on Le Cam's local asymptotic normality theory \citep{le2000asymptotics}, 
to establish the local asymptotic minimax optimality of a statistical procedure,
one needs to verify two things.
First, the class of perturbed distributions (the class $\{ s_{\alpha,g} \}_{\alpha,g}$ in our case) satisfies the locally asymptotically normal (LAN) condition \citep{vaart1996weak,le2000asymptotics}. This part is completed by Lemma 8.3 from \citet{duchi2021asymptotic} since our construction of perturbed supply distributions follows theirs.
Second, one should verify the asymptotic variance of the statistical procedure equals to the minimax optimal variance.

For a given perturbation ${(\alpha,g)}$, we let $\pst_{\alpha,g}$ and $\REVst_{\alpha,g}$ be the limit FPPE price and revenue under supply distribution $s_{\alpha,g}$.
Let $S_{\alpha,g}(\theta) = \nabla _\alpha \log s_{\alpha,g}(\theta)$ be the score function. 
So $\nabla_\alpha s_{\alpha,g} = s_{\alpha,g} S_{\alpha,g}$
and $\int S_{\alpha,g} s_{\alpha,g} \diff \t = 0$.
Obviously with our parametrization of $s_{\alpha,g}$ we have $S_{0,g}(\theta) = g(\theta)$ by \cref{eq:perturbed_is_roughly_expo}.

Let $\NSW_{\alpha,g}$ be the Nash social welfare under supply $s_{\alpha,g}$. 
Then $\NSW_{\alpha,g} = \int F(\t, \betast_{\alpha,g}) s_{\alpha,g} \diff \t + \text{constant that does not depend on $\alpha$}$.
So 
\begin{align*}
    \nabla_\alpha \NSW_{\alpha,g}
     = 
     \int [\nabla_\beta F( \t, \betast_{\alpha,g})\nabla_ \alpha \betast_{\alpha,g} + F( \t, \betast_{\alpha,g}) S_{\alpha,g}(\theta)] s_{\alpha,g}(\theta)\diff \theta
    = 0 + \int F( \t, \betast_{\alpha,g})  g(\theta) s_{\alpha,g}\diff \theta , 
\end{align*}
and $\nabla_\alpha \NSW_{\alpha,g} | _{\alpha = 0} = \int F(\t, \betast)g s\diff \t$.
Following the argument in \citet[Sec.\ 8.3]{duchi2021asymptotic} it holds that 
the asymptotic local mimimax risk $\geq \E[ L(\cN(0, \cov(F(\t, \betast))))] = \E[L(\cN(0,\sigma^2_{\NSW}))]$.

Let $u^*_{\alpha,g}$ be the utility under supply $s_{\alpha,g}$.
Note $u^*_{\alpha,g} = [b_1/\betast_{\alpha,g,1},\dots, b_n / \betast_{\alpha,g,n}]$.
By a perturbation result by Lemma 8.1 and Prop.\ 1 from \citet{duchi2021asymptotic}, under twice differentiability, 
$\nabla_\alpha \betast_{\alpha,g}|_{\alpha = 0} = - \cH\inv \E[ \nabla F( \t, \betast ) g(\t)\tp]$.
Then $\nabla_\alpha u^*_{\alpha,g}  | _{\alpha = 0} = \Diag(-b_i / (\betasti)\sq ) (\nabla_\alpha \betast_{\alpha,g} |_{\alpha = 0} ) = - \Diag(-b_i / (\betasti)\sq ) \cH\inv \E[\nabla F( \t, \betast) g(\t) \tp]  $.
We conclude the asymptotic local mimimax risk
is lower bounded by $\E[L(\cN(0, \Sigma_u))]$
where $\Sigma_u = \Diag(-b_i / (\betasti)\sq ) \cH\inv \E[\nabla F( \t, \betast ) \nabla   F( \t, \betast)\tp] \cH \inv \Diag(-b_i / (\betasti)\sq )$.
\end{proof}

\subsection{Proof of Theorem~\ref{thm:ci_lnsw}}
\label{sec:proof_var_est}
\begin{proof}[Proof of \cref{thm:ci_lnsw}]
    Define the functions 
        $\hat \sigma \sq (\beta)  \defeq \frac{1}{t} \sumtau( F(\thetau, \beta)  - H_t(\beta))\sq
        $ and 
        $\sigma \sq (\beta)  \defeq \var(F(\t,\b)) = \E[(F(\t,\b) - H(\beta))\sq ]
    $.
    We will show uniform convergence of $\hat \sigma \sq$ to $\sigma\sq$ on $C$, i.e., $\sup_{\beta \in C} |\hat \sigma \sq - \sigma\sq|  \toas 0$.
    We first rewrite $\sighatsq$ as follows 
    $\sighatsq(\beta) = {\frac{1}{t}   \sumtau \big(F(\thetau, \beta) - H(\beta)\big)\sq}
        - {(H_t(\beta)  - H(\beta))\sq  }
       := \I(\b) - \II(\b)
    $.
    By Theorem 7.53 of~\citet{shapiro2021lectures} (a uniform law of large number result for convex random functions), the following uniform convergence results hold
        $\sup_{\beta \in C} |\I(\beta)  - \sigma\sq(\beta) | \toas 0
        $
    and $ \sup_{\beta \in C} |\II(\beta)| \toas 0$.
    The above two inequalities imply $\sup_{\beta \in C} |\hat \sigma \sq - \sigma\sq|  \toas 0$. Note the variance estimator 
    $ \hat \sigma\sq_{\NSW} = \hat \sigma\sq (\betagam)$ 
    and the asymptotic variance 
    $ \sigma\sq_{\NSW} = \sigma \sq(\betast) $.
    By $\betagam \toas \betast$ we know, 
    $ 
        | \hat \sigma\sq_{\NSW} - \sigma\sq_{\NSW}| 
         = |\hat \sigma\sq(\betagam) - \sigma \sq(\betast) | 
         \leq |\hat \sigma\sq (\betagam) -  \sigma\sq (\betagam)|
        + |  \sigma\sq (\betagam) -  \sigma\sq(\betast) |
         \to 0 \text{ a.s.}
    $ where the first term vanishes by the uniform convergence just established, the second term by continuity of $\sigma \sq(\cdot)$ at $\betast$. 
    Now we have shown $\hat \sigma\sq_{\NSW}$ is a consistent variance estimator for the asymptotic variance. Then by Slutsky's theorem we know $ {\sqrt t { (\hat \sigma _\NSW)\inv}  ( \LNSW^\gam - \LNSW^*) }\tod \cN (0,1)$. This completes the proof of \cref{thm:ci_lnsw}.
\end{proof}

%% file: fppe_app.tex
\newpage
\ECHead{Proofs of FPPE Results}
\section{Appendix to FPPE}
\subsection{Scale-Invariance of FPPE}
\label{sec:fppe_scale_invariance}

FPPE has some of the same scale-invariance properties as LFM.
In particular, scaling the budget and supply at the same time does not change the market equilibrium. That is, given a positive scalar $\alpha$, if $(\beta,p)$ are the equilibrium pacing multiplier and prices in the market $\FPPE(b,v,s,\Theta)$, then 
  $(\beta,   p)$ are
  the equilibrium quantities in the market $ \FPPE(\alpha b,  v, \alpha s, \Theta)$.
The same scale-invariance holds for the finite FPPE, i.e., 
if $(\beta, p) = \oFPPE(b,v,\sigma, \gamma)$, then $(\b, p) = \oFPPE(\alpha b, v, \alpha \sigma, \gamma)$.
Given the invariance, 
we can see that in order for the budget-supply ratio to match the limit market $\FPPE(b,v,s,\Theta)$, the finite market should be configured as either $\oFPPE(tb,v, 1, \gamma)$
or $\oFPPE(b,v,1/t, \gamma)$. We will study the latter, and simply refer to it as the finite FPPE.
Unlike LFM, FPPE does not enjoy invariance to valuation scaling, because buyers have a pacing multiplier of at most one in FPPE.

\section{Technical Lemmas for FPPE}
\label{sec:techlemma_fppe}

\begin{proof}[Proof of \cref{lm:fppe_relation}]
    To show the first equality, note 
    \begin{align*}
        \nabla H (\betast) &= \E[\nabla_\b \max_i \betai \vithe] - [b_1 / (\betast_1) , \dots, b_n / (\betast_n)]\tp
        \\ 
        & = \mubarst - [b_1 / (\betast_1) , \dots, b_n / (\betast_n)]\tp = - \deltasti
        \;.
    \end{align*}
    To show the second equality,
    note for any twice differentiable first-order homogenous function $\bar f: \Rn \mapsto \R$, 
    it must hold $\nabla\sq \fbar (\b)\beta = 0$.
    And we have $\nabla\sq H(\betast ) \betast = \Diag(b_i / (\betasti )\sq) \betast =  [ b_1 / \betast_1, \dots ,b_n / \betast_n]\tp$.
\end{proof}
\subsection{A CLT for constrained $M$-estimator}
We introduce a CLT result from \cite{shapiro1989asymptotic} that handles $M$-estimation when the true parameter is on the boundary of the constraint set. 
Throughout this section, when we refer to assumptions A1, A2, B2, etc, we mean those assumptions in \citet{shapiro1989asymptotic}.

Let $(\Theta, P)$ be a probability space. Consider $f: \Theta \times \R^n \to \R$ and a set $B \subset \R^n$. Let $\theta_1, \ldots, \theta_t$ be a sample of independent random variables with values in $\Theta$ having the common probability distribution $P$.
Let $\phi( \beta) = P f(\cdot, \beta) = \E[f(\theta, \beta)]$, and $\psi_t(\beta) = P_t f(\cdot,\beta) = \frac1t \sum_{i=1}^t f(\theta_i, \beta)$. Let $\beta_0$ be the unique minimizer of $\phi$ over $B$  (Assumption A4 in  \citet{shapiro1989asymptotic}).
Let $\vartheta_ t = \inf _B \psi_t$ and $\hat \beta $ be an optimal solution.

We begin with some blanket assumptions.
Suppose the geometry of $B$ at $\beta_0$ is given by functions $g_i(\beta)$ (Assumption B1), i.e., there exists a neighborhood $N$ such that
$$
B \cap {N}=\left\{\beta \in {N}: g_i(\beta)=0, i \in K ; g_i(\beta) \leq 0, i \in J\right\},
$$
where $K$ and $J$ are finite index sets and the constraints in $J$ are active at $\beta_0$, meaning   $g_i\left(\beta_0\right)=0$ for all $i \in J$. Assume the functions $g_i, i \in K \cup J$, are twice continuously differentiable in a neighborhood of $\beta_0$ (Assumption B2).
Define the Lagrangian function by
$
l(\beta, \lambda)=\phi(\beta)+\sum_{i \in K \cup J} \lambda_i g_i(\beta).
$
Let $\Lambda_0$ be the set of optimal Lagrange multipliers, i.e., $\lambda \in \Lambda_0$ iff
$
\nabla l\left(\beta_0, \lambda\right)=0
$ (assuming differentiability)
and $\lambda_i \geq 0, i \in J$.

\begin{lemma}[Theorems 3.1 and 3.2 from \citet{shapiro1989asymptotic}] \label{thm:the_shapiro_thm}
    Assume there exists a neighborhood $N$ of $\beta_0$ such that the following holds.
    \begin{enumconditions}
        \item Conditions on the objective function $f$ and the distribution $P$.
        \begin{itemize}
            \item 
        (Assumption~A1 in the original paper) For almost every $\theta$, $f(\theta,\beta)$ is a continuous function of $\beta$, and for all $\beta\in B$, $f(\theta,\beta)$ is a measurable function of $\theta$.
        \item 
        (Assumption~A2) The family $\{f(\theta, \beta)\}, \beta \in B$, is uniformly integrable.
        \item 
        (Assumption~A4) For all $\theta$, there exist a positive constant $K(\theta)$ such that
        $
        |f(\theta, w)-f(\theta, \beta)| \leq K(\theta)\|w-\beta\|
        $
        for all $\beta, w \in {N}$.
        \item 
        (Assumption~A5) For each fixed $\beta \in {N}, f(\theta, \cdot)$ is continuously differentiable at $\beta$ for almost every $\theta$.
        \item
        (Assumption~A6) The family $\{\nabla f(\theta, \beta)\}_{\beta \in {N}}$, is uniformly integrable.
        \item
        (Assumption~D) The expectation $\E[\left\|\nabla f\left(\theta, \beta_0\right)\right\|^2]$ is finite. 
        \item
        (Assumption~B4) The function $\phi$ is twice continuously differentiable in a neighborhood of $\beta_0$.
        \end{itemize}

        \item Conditions on the optimal solution.
        \begin{itemize}
            \item 
        (Assumption~B3) A constraint qualification, the Mangasarian-Fromovitz condition: The gradient vectors $\nabla g_i\left(\beta_0\right), i \in K$, are linearly independent,
        and there exists a vector $w$ such that
        $
        w \tp \nabla g_i\left(\beta_0\right)=0,  i \in K$ and $w \tp \nabla g_i\left(\beta_0\right)<0,  i \in J .
        $
        \item
        (Assumption~B5) Second-order sufficient conditions: 
        Let $C$ be the cone of critical directions
        \begin{align}
            \label{eq:def_C}
            C=\left\{w: w \tp \nabla g_i\left(\beta_0\right)=0, i \in K ; w \tp \nabla g_i\left(\beta_0\right) \leq 0, i \in J ; w \tp \nabla \phi\left(\beta_0\right) \leq 0\right\} .
        \end{align}
        The assumption requires that for all nonzero $w \in C$,
        $
        \max _{\lambda \in \Lambda_0} w \tp \nabla^2 l\left(\beta_0, \lambda\right) w>0,
        $
        \end{itemize}

        \item Stochastic equicontinuity, a modified version of Assumption~C1 in the original paper.
            For any sequence $\delta_t = o(1)$, the variable
            \begin{align}
                \label{lm:se_of_subgradient}
                \sup _{\beta:\| \beta -\beta_0\| \leq\delta_t} \frac{\left\|(\nabla \psi_t - \nabla\phi)(\beta)
                -(\nabla \psi_t - \nabla\phi)(\beta_0)
                \right\|}{t^{-1 / 2}+\left\|\beta-\beta_0\right\|} = o_p(1)
            \end{align}
            as $t \rightarrow \infty$. Here the supremum is taken over $\beta$ such that $\nabla \psi_t(\beta)$ exists. 
            \label{it:stoc_equic}
          
    \end{enumconditions}
    Then it holds that $ \hat\beta \toprob \beta_0$.
    Let
    \begin{align}
        \label{eq:def_zeta}
        \zeta_t=\nabla \psi_t\left(\beta_0\right)-\nabla \phi\left(\beta_0\right),
    \end{align}
    and 
    \begin{align}
        \label{eq:def_q}
        q(w)=\max_{  \lambda \in \Lambda_0} \{w \tp  \nabla^2 l\left(\beta_0, \lambda\right)w \}.
    \end{align}
    Then 
    $$\vartheta_t\defeq \inf _B \psi_t =\psi_t\left(\beta_0\right)+\min _{w \in C}\{w \tp \zeta_t +\frac{1}{2} q(w)\}+o_p (t^{-1} ).$$ 
    Furthermore, suppose for all $\zeta$ the function $w \mapsto w \tp \zeta+\frac{1}{2} q(w)$ has a unique minimizer $\bar{\omega}(\zeta)$ over $C$. Then $$\|\hat {\beta} -\beta_0-\bar{\omega}(\zeta_t)\|=o_p(t^{-1 / 2}) .$$
\end{lemma}

\begin{remark}[The stochastic equicontinuity condition]
    By inspecting the proof, the original Assumption~C1, 
    $\sup _{\beta \in B \cap {N}} {\left\|\nabla \psi_t(\beta)-\nabla \phi(\beta)-\nabla \psi_t\left(\beta_0\right)+\nabla \phi\left(\beta_0\right)\right\|}/[{t^{-1 / 2}+\left\|\beta-\beta_0\right\|}] = o_p(1)$, which requires uniform convergence over a fixed neighborhood $N$, 
    can be relaxed to the uniform convergence in a shrinking neighborhood of $\beta_0$.
    The shrinking neighborhood condition is in fact
    standard, see, e.g., \citet{pakes1989simulation,newey1994large}.
\end{remark}

\begin{remark}
    The limit distribution of the minimizer is characterized by three objects: the limit distribution of $\zeta_t$ defined in \cref{eq:def_zeta}, the critical cone $C$ defined in \cref{eq:def_C} and the piecewise quadratic function $q$ defined in \cref{eq:def_q}.
\end{remark}


Hessian matrix estimation at the optimum $\beta_0$ can be done via the numerical difference method.

\begin{lemma}[Hessian estimation via numerical difference]
    \label{lm:hessian_estiamtion}
    This lemma is adapted from Theorem 7.4 from \citet{newey1994large}
    Recall $\phi(\beta) = Pf(\cdot, \beta)$, $\psi_t(\beta) = P_t f(\cdot, \beta)$ and $\zeta_t=\nabla \psi_t(\beta_0) - \nabla \phi(\beta_0)$. We are interested in the Hessian matrix $H = \nabla\sq\phi(\beta_0)$.
    Let $\beta_0$ be any point and let $\hat \beta$ be an estimate of $\beta_0$. 
    Assume 
    \begin{enumconditions}
    \item $\hat \beta - \beta_0 = O_p(t^{-1/2})$; \label{it:1} 
    \item $\phi$ is twice differentiable at $\beta_0$ with non-singular Hessian matrix $H$; 
    \label{it:2} 
    \item $\sqrt{t} \zeta_t \tod N(0, \Omega)$ for some matrix $\Omega$; 
    \label{it:3} 
    \item for any positive sequence $\delta_t = o(1)$,
    the stochastic equicontinuity condition \cref{lm:se_of_subgradient} holds. 
    \label{it:4} 
    \end{enumconditions}
    Suppose $\varepsilon_t \to 0$ and $\varepsilon_t \sqrt{t} \to \infty$. Then $\hat H \toprob H$, where $\hat H$ is the numerical difference estimator whose $(i,j)$-th entry is
\begin{align*} 
    \hat{H}_{i j}=& {[\psi_t(\hat{\beta}+e_i \varepsilon_t+e_j \varepsilon_t)-\psi_t(\hat{\beta}-e_i \varepsilon_t+
    {e}_j \varepsilon_t)-\psi_t(\hat{\beta}+e_i \varepsilon_t-e_j \varepsilon_t).} \\ 
    &+\psi_t(\hat{\beta}-e_i \varepsilon_t-e_j \varepsilon_t)] / 4 \varepsilon_t^2. 
\end{align*}
\end{lemma}
\begin{proof}[Proof of \cref{lm:hessian_estiamtion}]
    We provide a proof sketch following Theorem 7.4 from \citet{newey1994large} and 
    Lemma 3.3 in \citet{shapiro1989asymptotic}.
    By \cref{it:1} and $\invtroot = o(\varepsilon_t)$ we know for any vector $a\in\Rn$, it holds $\| \hat \beta + \varepsilon_t a - \beta_0 \| = O_p(\varepsilon_t)$. 
    Let $\beta = \hat \beta + a \varepsilon_t$.
    By a mean value theorem for locally Lipschitz functions (see \citet{Clarke1990}; the lemma is also used in the proof of Lemma 3.3 in \citet{shapiro1989asymptotic}),
    there is a (sample-path dependent) $\beta'$ on the segment joining $\beta$ and $\beta_0$ such that 
    $$ (\psi_t - \phi)(\beta) - (\psi_t - \phi)(\beta_0) = (\zeta\st_t)\tp (\beta - \beta_0).$$
    for some $\zeta\st_t \in \partial \psi_t(\beta') - \nabla \phi(\beta')$.
    Then 
    \begin{align}
        |(\psi_t - \phi)(\beta) - (\psi_t - \phi)(\beta_0)| 
        \notag
        &\leq \|\zeta_t \| \|\beta-\beta_0\| + \| \zeta\st_t - \zeta_t\|\|\beta-\beta_0\|
        \\
        & = \|\zeta_t \| \|\beta-\beta_0\| + o_p(\invtroot + \|\beta' - \beta_0 \|)\|\beta-\beta_0\| 
        \tag{by \cref{it:4} }
        \\
        &= O_p(\invtroot) O_p(\varepsilon_t) + o_p(\invtroot + O_p(\varepsilon_t)) O_p(\varepsilon_t ) \tag{by \cref{it:3} }
        \\ &= o_p(\varepsilon_t\sq) 
        \label{eq:from_psi_to_phi}
    \end{align}
    Next by \cref{it:2} we have a quadratic expansion
    \begin{align}
        \label{eq:quad_expan}
        \phi(\beta) - \phi(\beta_0) - \nabla \phi(\beta_0)\tp (\beta-\beta_0) - \frac12  (\beta-\beta_0) \tp H  (\beta-\beta_0) = o_p(\varepsilon_t\sq). 
    \end{align}
    Let $a_{\pm\pm} = \pm e_i \varepsilon_t \pm e_j \varepsilon_t$, $\hat\beta_{\pm\pm} = \hat\beta + a_{\pm \pm}$ and $ d_{\pm\pm} = \hat\beta_{\pm\pm} -\beta_0$.
    Then $d_{\pm\pm} = O_p(\varepsilon_t)$ and $d_{\pm\pm} = a_{\pm\pm} + o_p(\varepsilon_t)$.
    Applying the above bounds with $\beta \leftarrow \hat\beta_{\pm \pm} $, recalling the definition of $\hat H_{ij}$, we have 
    \begin{align}
        \hat H_{ij} &= [\psi_t(\hat\beta_{++})
        -\psi_t(\hat\beta_{-+})
        -\psi_t(\hat\beta_{+-})
        +\psi_t(\hat\beta_{--})] / (4\varepsilon_t\sq)
        \notag
        \\ 
        &=[\phi(\hat\beta_{++})
        -\phi(\hat\beta_{-+})
        -\phi(\hat\beta_{+-})
        +\phi(\hat\beta_{--}) + o_p(\varepsilon_t\sq)] / (4\varepsilon_t\sq)
       \tag{by \cref{eq:from_psi_to_phi}}
        \\ 
        &= [\nabla\phi(\beta_0) \tp (d_{++} - d_{-+} - d_{+-} + d_{--}) 
        \notag
        \\  & \quad
        + 
        \frac12( d_{++} \tp H d_{++} 
        - d_{+-} \tp H d_{+-}
        - d_{-+} \tp H d_{-+}
        + d_{--} \tp H d_{--}  ) + o_p(\varepsilon_t\sq) ] / 
        (4\varepsilon_t\sq)
        \tag{by \cref{eq:quad_expan}}
        \\
        & = [0 + 
        \frac12(
          a_{++} \tp H a_{++} 
        - a_{-+} \tp H a_{-+}
        - a_{+-} \tp H a_{+-}
        + a_{--} \tp H a_{--}  ) 
        + o_p(\varepsilon_t\sq) ] / 
        (4\varepsilon_t\sq)
        \notag
        \\
        &  = [
        4 \varepsilon_t\sq H_{ij}   + o_p(\varepsilon_t\sq) ] / 
        (4\varepsilon_t\sq)
        \notag
        \\ 
        &= H_{ij} + \frac{o_p(\varepsilon_t\sq)}{\varepsilon_t\sq} = H_{ij} +o_p(1).
        \notag
    \end{align}
    In the above we use $d_{++}\tp H d_{++} = (d_{++}-a_{++})\tp H d_{++}  + (d_{++}-a_{++})\tp H a_{++}+ a_{++}\tp H a_{++} = o_p(\varepsilon_t\sq) + a_{++}\tp H a_{++}$, and similarly for other terms.
    This completes the proof of \cref{lm:hessian_estiamtion}.
\end{proof}

The original conditions for \cref{lm:hessian_estiamtion} in \citet{newey1994large}
require the true parameter $\beta_0$ to lie in the interior of $B$. However, this condition is only used to derive the bound $\hat \beta - \beta_0 = O_p(t^{-1/2})$, which is assumed in our adapted version.

\subsection{Stochastic equicontinuity and VC-subgraph function classes}

Next we review classical results from the empirical process literature \citep{vaart1996weak,gine2021mathematical}.

We begin with the notions of Donsker function class and stochastic equicontinuity.

Let $(\Theta,P)$ be a probability space. 
Let $\cF$ be a class of measurable functions of finite second moment.
The class $\cF$ is called $P$-Donsker if 
a certain central limit theorem holds for the class of random variables $\{ \sqrt{t} (P_t - P) f : f\in\cF\}$,
where 
$P_t f = \frac1t \sum_{i=1}^t f(X_i)$ where $X_i$'s are i.i.d.\ draws from $P$.
Because Donskerness will be used as an intermediate step that we will not actually need to show directly or utilize directly,
we refer the reader to Definition~3.7.29 from \citet{gine2021mathematical} for a precise definition.


\begin{lemma}[Donskerness $\Leftrightarrow$ stochastic equicontinuity]\label{thm:donsker_se}
    \label{lm:donsker_to_SE}
 Let $d_P^2(f, g)=P(f-g)^2-(P(f-g))^2$ and consider the pseudo-metric space $\left(\mathcal{F}, d_P\right)$.
 Assume $\cF$ satisfies the condition $\sup_{f\in\cF}|f(x) - Pf | <\infty $ for all $x\in \Theta$. 
 Then the following are equivalent 
 \begin{itemize}
    \item $\cF$ is $P$-Donsker.
    \item $\left(\mathcal{F}, d_P\right)$ is totally bounded, and stochastic equicontinuity under the L2 function norm holds, i.e., for any $\delta_t = o(1)$,  
    $$
    \sup_{(f,g) \in [\delta_t]_{L^2}}|\sqrt{t}\left(P_t-P\right)(f-g)| = o_p(1)
    $$
    as $t\to \infty$,
    where $[\delta_t]_{L^2} = \{ (f,g):f, g \in \mathcal{F}, d_P(f, g) \leq \delta_t\}$.
 \end{itemize}
See Theorem~3.7.31 from \citet{gine2021mathematical}.
\end{lemma}
\cref{lm:donsker_to_SE} reduces the problem of showing stochastic equicontinuity under the L2 function norm to showing Donskerness.
In order to show Donskerness, we will show that our function class is \emph{VC-subgraph}, which implies Donskerness.
At the end, we will connect stochastic equicontinuity under the L2 function norm to the stochastic equicontinuity that we need (see \cref{lm:ltwo_to_parameter}).

Let $\mathcal{C}$ be a class of subsets of a set $\Theta$. Let $A \subseteq \Theta$ be a finite set. We say that $\mathcal{C}$ \emph{shatters} $A$ if every subset of $A$ is the intersection of $A$ with some set $C \in \mathcal{C}$. The subgraph of a real function $f$ on $\Theta$ is the set
$
G_f=\{(s, t): s \in \Theta, t \in \mathbb{R}, t \leq f(s)\} .
$

\begin{definition}[VC-subgraph function classes]
A collection of sets $\mathcal{C}$ is a Vapnik-\v Cervonenkis class $(\mathcal{C}$ is $V C)$ if there exists $k<\infty$ such that $\mathcal{C}$ does not shatter any subsets of $\Theta$ of cardinality $k$.
A class of functions $\mathcal{F}$ is $V C$-subgraph if the class of sets $\mathcal{C}=\left\{G_f: f \in \mathcal{F}\right\}$ is $V C$.
See Definition~3.6.1 and 3.6.8 from \citet{gine2021mathematical}
\end{definition}

\begin{lemma} [VC subgraph + envelop square integrability $\implies$ Donskerness]
    \label{lm:vc_to_donsker}

    If $\cF$ is VC-subgraph, and there exists a measurable $F$ such that $f\leq F$ for all $f\in \cF$ with 
    $P F\sq < \infty$, then $\cF$ is $P$-Donsker.
    See Theorem 3.7.37 from \citet{gine2021mathematical}
\end{lemma}

Since VC-subgraph implies Donskerness which is equivalent to stochastic equicontinuity, our problem reduces to showing the VC-subgraph property. The following lemmas show how to construct complex VC-subgraph function classes from simpler ones, and will be used in our proof.
\begin{lemma}[Preservation of VC class of sets, Lemma 2.6.17 from \citet{vaart1996weak}] \label{lm:CV_set_intersection}
    If $\cC$ and $\cD$ are VC classes of sets. Then $\cC \cap \cD = \{ C \cap D: C\in \cC, D \in \cD \}$ and $\cC ^c = \{C^c : C \in \cC \}$ are VC. 
\end{lemma}

\begin{lemma}[Preservation of VC-subgraph function classes]
    \label{lm:vc_preservation}
    Let $\mathcal{F}$ and $\mathcal{G}$ be $V C$-subgraph classes of functions on a set $\Theta$ and $g: \Theta \mapsto \mathbb{R}$ be a fixed function. Then
    $\mathcal{F} \vee \mathcal{G}= \{f \vee g: f \in \mathcal{F}, g \in \mathcal{G}\}$,
    $\mathcal{F}+g=\{f+g: f \in \mathcal{F}\}$,
    $\mathcal{F} \circ \phi = \{f \circ \phi : f\in \mathcal{F}\}$ is VC-subgraph for fixed $\phi: \cX \to \Theta$, 
    and $\mathcal{F} \cdot g=\{f g: f \in \mathcal{F}\}$ are VC-subgraph.
    See Lemma~2.6.18 from \citet{vaart1996weak}
\end{lemma}
\begin{lemma}[Problem 9 Section 2.6 from \citet{vaart1996weak}]
    \label{lm:vcset_to_indicators}
    If a collection of sets $\mathcal{C}$ is a VC-class, then the collection of indicators of sets in $\mathcal{C}$ is a VC-subgraph class of the same index.
\end{lemma}
In general, the VC-subgraph property is not preserved by multiplication, whereas Donskerness is. Thus, our proof will use the VC-subgraph property up until a final step where we need to invoke multiplication, which will instead be applied on the Donskerness property.
\begin{lemma}[Corollary 9.32 from \citet{kosorok2008introduction}]
    \label{lm:donsker_multiplication}
    Let $\cF$ and $\cG$ be Donsker, 
    then $\cF \cdot \cG$ is Donsker if both $\cF$ and $\cG$ are uniformly bounded.
\end{lemma}

For parametric function classes, if the parametrization is continuous in
a certain sense, then
stochastic equicontinuity holds w.r.t.\ the norm in the parameter space.

\begin{lemma}[From $L^2$-norm to parameter norm]
    \label{lm:ltwo_to_parameter}
    Suppose the function class $\cF = \{ f(\cdot, \beta), \beta\in B\}$, $B \subset \Rn$, is $P$-Donsker, with an envelope $F$ such that $P F\sq < \infty$. Suppose $\int\left[f(\cdot, \beta)-f\left(\cdot, \beta_0\right)\right]^2 d P \rightarrow 0 $ as $ \beta \rightarrow \beta_0$. Then for any positive sequence $\delta_t = o(1)$, it holds
    \begin{align}
        \label{eq:se_in_parameter}
        \sup _{\beta: \|\beta-\beta_0 \|<\delta_t}|
            \sqrt{t}\left(P_t-P\right)(            
            f(\cdot, \beta)- f\left(\cdot, \beta_0 \right))|=o_p(1).
    \end{align}
    See Lemma 2.17 from \citet{pakes1989simulation}; see also Lemma 1 from \citet{chen2003estimation}
\end{lemma}

\begin{lemma}[\citet{andrews1994asymptotics}]
    \label{lm:se}
    If for any $\delta_t = o(1)$ \cref{eq:se_in_parameter} holds,
    then for any random elements $\beta_t$ such that $\| \beta_t - \beta_0\|_2 =o_p(1)$, it holds 
    $ 
        \sqrt{t}(P_t - P) ( f(\cdot,\beta_t) - f(\cdot, \beta_0)) = o_p(1).
    $
    
\end{lemma}

\begin{lemma}\label{lm:verifying_VC}
    Given any $n$ fixed functions $v_i: \Theta \to \R$, $i\in [n]$,
    the following function classes
    \begin{align*}
    \cF_1 &= \{ \theta \mapsto  \max_i \{\beta_1v_1(\theta), \dots, \beta_n v_n(\theta)\}  :\beta \in B \}
    \\
    \cF_2 &= \{ \theta \mapsto  [D_{f,1}, \dots, D_{f,n}](\theta, \beta): \beta \in B \}  
    \\
    \cF_3 &= \{ \t \mapsto \max_i (\betai \vithe) \indi \{ \max_{i \in A} \betai\vithe = \max_{i \in [n]} \betai\vithe \} : \b \in B \}
    \end{align*}
    are VC-subgraph and Donsker.    
    Here $D_{f,i} (\theta, \beta) = v_i(\theta)\prod_{k=1}^n \indi( \betai \vithe \geq \beta_k v_k(\theta) )$ and $B=[0,1]^n$, and $A$ is a nonempty subset of $[n]$.
\end{lemma}
\begin{proof}[Proof of \cref{lm:verifying_VC}]
    We show $\cF_1$ is VC-subgraph.
    For each $i$, the class 
    $\{ \theta \mapsto v_i(\theta) \beta_i : \beta_i \in [0,1] \}$ is VC-subgraph (Proposition 4.20 from \citet{wainwright2019high}, and Example 19.17 from \citet{van2000asymptotic}). 
    By the fact the VC-subgraph function classes are preserved by pairwise maximum (\cref{lm:vc_preservation}), we know $\cF_1$ is VC-subgraph. 
    Moreover, the required envelope condition holds since $\esssup _\theta f \leq \vbar $ for all $f\in \cF_1$, so $\cF_1$ is Donsker by \cref{lm:vc_to_donsker}.

    We now show $\cF_2$ is VC-subgraph.
    For a vector-valued function class, we say it is VC-subgraph if each coordinate is VC-subgraph.
    First, the class of sets $\{ \{ v \in \Rn:\betai v_i \geq \beta_k v_k\} \subset \Rn : \beta \in B \}$ is VC, for all $k \neq i$. 
    By \cref{lm:CV_set_intersection}, we know 
    the class of sets $\{ \{v \in \Rn : \betai v_i \geq \beta_k v_k, \forall k\neq i\} \subset \Rn: \beta \in B \}$ is VC.
    By \cref{lm:vcset_to_indicators}, we obtain that the class 
    $\{ \theta \mapsto \prod_{k=1}^n  \indi( \betai \vithe \geq \beta_k v_k(\theta) ): \beta \in B \}$
    is VC-subgraph.
    Finally, multiplying all functions by a fixed function preserves VC-subgraph classes (\cref{lm:vc_preservation}), and so 
    $\{ \theta \mapsto v_i(\theta)\prod_{k=1}^n \indi( \betai \vithe \geq \beta_k v_k(\theta) ): \beta \in B\}$ is VC-subgraph.
    Repeat the argument for each coordinate, and we obtain that $\cF_2$ is VC-subgraph.
    Moreover, the required envelope condition holds since $\esssup_\theta \|D_f(\theta, \beta)\|_2 \leq n\vbar $ for all $D_f \in \cF_2$, and so $\cF_2$ is Donsker by \cref{lm:vc_to_donsker}. We conclude the proof of \cref{lm:verifying_VC}.

    Finally, to see $\cF_3$ is VC-subgraph and Donsker, rewrite its functions as 
    $    \indi \{ \max_A \betai\vithe = \max_{[n]} \betai \vithe \}
        = \prod_{k=1}^n (1 - \prod_{i \in A} \indi (\betai\vithe < \beta_k v_k(\t)))
    $
    and apply similar arguments as above. 
\end{proof}

\section{Proofs for Main Theorems in FPPE}

\subsection{Proof of \cref{thm:fppe_as_convergence}}
\label{sec:proof:thm:fppe_as_convergence}
The convergence $\betagam \toas \betast$ follows the same proof as \cref{thm:consistency} and is ommitted.
To show almost sure convergence of revenue, we note 
\begin{align*}
    & \Big|\frac1t \sumtau \ptau -  \int_\Theta p^*(\theta) s(\theta) \diff \theta\Big|
    \\
    & \leq \frac1t \sumtau |\max_i \{ \vithetau \beta^\gam_i\} - \max_i \{ \vithetau \beta^*_i \} | + 
    \Big| \frac1t \sumtau  \max_i \{ \vithetau \beta^*_i \}  -  \int_\Theta p^*(\theta) s(\theta) \diff \t \Big|
    \\
    & \leq  \vbar \| \beta^\gam - \betast\|_\infty +    
    \Big|\frac1t \sumtau  \max_i \{ \vithetau \beta^*_i \}  -  \int_\Theta p^*(\theta) s(\t)\diff \t\Big|
    \toas 0 \;.
\end{align*}
Here the first term converges to zero a.s.\ by $\beta^\gam \toas \betast$, and the second term converges to $0$ a.s.\ by strong law of large numbers and noting $\E[\max_i \{v_i(\theta)\beta^*_i \}] = \E[p^*(\theta)] = \REVst$.

\subsection{Proof of \cref{thm:rev_convergence}}
\label{sec:proof:thm:rev_convergence}

\begin{proof}[Proof of \cref{thm:rev_convergence}]

    First, by \citet{gao2022infinite}, there is a natural lower bound for the equilibrium pacing multipliers.
    To lower bound $\betasti$, 
    note
    $\betasti = b_i / \usti = b_i / (\deltasti + \musti) \geq b_i / (b_i + \int v_i s \diff \t) = b_i / (b_i + \nu_i)$.
    Then $b_i / (b_i + \nu_i) \leq \betasti \leq 1$. 
    Define the set
    \begin{align*}
        C_\FPPE \defeq \prod_{i=1}^n \bigg[\frac{b_i}{2 \nu_i + b_i}, 1 \bigg]
        \;.
    \end{align*}
    Clearly we have $\betast \in C_\FPPE$. Furthermore, for $t$ large enough $\betagam \in C_\FPPE$ with high probability.
    To see this, if $t$ satisfies $t \geq 2 (\vbar / \min_i \nui) \sq \log(2n/\eta)$, then $\frac1t \sumtau \vithetau \leq 2 \E[v_i(\theta)]$ for all $i$ with probability $\geq 1-\eta$. By a bound on $\betagam$ in the QME
    $ 
        \betagami \geq \frac{b_i}{b_i +  \frac1t \sumtau \vithetau}
    $
    (see Section 6 in \citet{gao2022infinite}), we obtain $\betagami \geq \frac{b_i}{b_i + 2\nu_i}$ (recall $\nu_i = \E[v_i(\theta)]$).

    Let $L_\FPPE$ and $\lambda_\FPPE$ be the Lipschitz constant and strong convexity constants of $H$ and $H_t$ w.r.t\ $\ell_\infty$-norm on $C_\FPPE$.
    We estimate $L_\FPPE$ and $\lambda_\FPPE$. On $C_\FPPE$, the minimum eigenvalue of $\nabla\sq\Psi  (\beta )= \Diag\{ \frac{b_i}{(\betai)\sq} \}$ can be lower bounded by $\ubar{b}$. So we conclude $\lambda_\FPPE = \ubar{b}$.
    And the Lipschitzness constant can be seen by the following. For $\beta,\beta' \in C_\FPPE$,
    \begin{align*}
        & |H_t(\beta)  - H_t(\beta')| 
        \\
        & \leq \frac1t \sumtau \big|\max_i \{\vithetau \beta_i \} - \max_i \{\vithetau \beta_i' \}\big| + 
        \sumiton b_i \big| \log \beta_i -  \log \beta_i'\big|
        \\
        & \leq \vbar \| \beta - \beta'\|_\infty + \sumiton b_i \cdot \frac{1}{b_i / (2\nui + b_i)} |\beta_i - \beta_i'|
        \\
        & \leq \big(\vbar + 2\nubar n + 1\big) \| \beta - \beta'\|_\infty
        \;,
    \end{align*}
    where we used $\sumi b_i =1$.
    Similar argument shows that $H$ is also $(\vbar + 2\nubar n + 1)$-Lipschitz on $C_\FPPE$. 
    We conclude $L_\FPPE =(\vbar + 2\nubar n + 1)$.

    To obtain the convergence rate, we simply repeat the proof of \cref{thm:high_prob_containment}. 
    We obtain from \cref{eq:final_t_bound} that with probability $\geq 1-2\alpha$,
    there exists a constant $c'$ such that
    as long as 
    \begin{align}
        \label{eq:t_bound_in_quasilinear_case}
        t & \geq c'\cdot  L_\FPPE \sq  \min \bigg\{ \frac{1}{ \lambda_\FPPE \epsilon} , \frac{1}{\epsilon\sq} \bigg\} \cdot \bigg(n\log\Big(\frac{16L_\FPPE}{\eps - \delta}\Big) + \log \frac1\alpha\bigg)
        \;,
    \end{align}
    it holds $|H(\betagam)- H(\betast)| < \epsilon$ and that $\betagam \in C_\FPPE$ (see \cref{cor:H_concentration}).
    Now \cref{eq:t_bound_in_quasilinear_case} shows that for $ \epsilon < \ubar{b}$ 
    (so that the $1/(\lambda_\FPPE \epsilon)$ term in the min becomes dominant) we have 
    \begin{align*}
        |H(\betagam) - H(\betast)| = 
        \tilde{O}_p
        \bigg(
            \frac{n\big(\vbar + 2\nubar n + 1\big)\sq  }{\ubar{b} t}
        \bigg)
        \;,
    \end{align*}
    where we use $\tilde{O}_p$ to ignore logarithmic factors of $t$. Moreover, 
    \begin{align*}
        \| \betagam - \betast\|_\infty \leq  
        \| \betagam - \betast\|_2
        \leq \sqrt{2|H(\betagam) - H(\betast)|  / \lambda_\FPPE}
        = \tilde O_p\bigg(\frac{\sqrt{n} \big(\vbar + 2\nubar n + 1\big)}{\ubar{b} \sqrt t}\bigg)
        \;.
    \end{align*}
    From here we obtain
    \begin{align*}
        & |\REV^\gam - \REV\st| 
        \\
        & \leq \vbar \| \betagam - \betast\|_\infty +\Big|\frac1t \sumtau  \max_i \{ \vithetau \beta^*_i \}  -  \int_\Theta p^*(\theta) s (\theta ) \dt \Big|
        \\
        & = \tilde O_p\bigg(\frac{\vbar \sqrt{n} \big(\vbar + 2\nubar n + 1\big)}{\ubar{b} \sqrt t}\bigg) + O_p \bigg(\frac{\vbar}{\sqrt{t}}\bigg)
        \\
        & = \tilde O_p\bigg(\frac{ \vbar \sqrt{n} \big(\vbar + 2\nubar n + 1\big)}{\ubar{b} \sqrt t}\bigg)
        \;.
    \end{align*}
    We conclude
    $ |\REV^\gam - \REV\st| = 
   \tilde{O}_p
   \Big(\frac{\vbar \sqrt{n} (\vbar + 2\nubar n + 1 ) }{ \ubar{b} \sqrt{t}}\Big)
   $. This completes the proof of \cref{thm:rev_convergence}.
\end{proof}

\subsection{Proof of \cref{thm:clt}} \label{sec:proof:thm:clt}
\begin{proof}[Proof of \cref{thm:clt}]

    The proof of \cref{thm:clt} proceeds by showing that FPPE satisfy a set of regularity conditions that are sufficient for asymptotic normality~\citep[Theorem 3.3]{shapiro1989asymptotic}; the conditions are stated in \cref{thm:the_shapiro_thm} in the appendix. Maybe the hardest condition to verify is the so called stochastic equicontinuity condition (\cref{it:stoc_equic}), which we establish with tools from the empirical process literature. 
    In particular, we show that the class of functions, parameterized by a pacing multiplier vector $\beta$, that map each item to its corresponding first-price auction allocation under the given $\beta$, is a VC-subgraph class. This in turn implies stochastic equicontinuity. 
    \cref{as:scs} is used to ensure normality of the limit distribution.

    We verify all the conditions in \cref{thm:the_shapiro_thm}. Recall $I_= = \{i:\betasti = 1\}$ is the set of active constraints. The local geometry of $B$ at $\betast$ is described by the $|I_=|$ constraint functions $g_i(\beta) = e_i\tp \beta - 1$, $i\in I_=$.

    First, we verify the conditions on the probability distribution and the objective function. 
    A1 holds obviously for the map $\beta \mapsto \max_i \betai v_i(\theta)$. 
    A2 holds by $f \leq \vbar$.
    A4 holds with Lipschitz constant $\vbar$.
    A5 holds since by \cref{as:smo} there is a neighborhood $\Ndiff$ of $\betast$ such that for all $\beta \in \Ndiff$, the set $\{\theta: f(\theta,\cdot) \text{ not differentiable at $\beta$} \}$ is measure zero. 
    A6 holds by $\|\nabla f(\theta,\beta)\|_2\leq  n \vbar$.
    B4 holds by \cref{as:smo}.

    Second, we verify the conditions on the optimality. 
    B3 holds since the constraint functions are $g_i(\beta) = e_i\tp \beta - 1$, $i\in I_=$, whose gradient vectors are obviously linear independent. Moreover, the set $\{\beta: \beta_i >0, i\in I_=\}$ is nonempty.
    B5 holds by $\nabla\sq H(\betast) = \nabla\sq\fbar (\betast) + \Diag(b_i/\betasti\sq) \succcurlyeq \Diag(b_i/\betasti\sq)$ being positive definite.

    Finally, we verify the stochastic equicontinuity condition. 
    Recall the definitions of the following two function classes from \cref{lm:verifying_VC} 
    \begin{align*}
        \cF_1 &= \{ \theta \mapsto  \max_i \{\beta_1v_1(\theta), \dots, \beta_n v_n(\theta)\}  :\beta \in B \},
        \\
        \cF_2 &= \{ \theta \mapsto  D_f (\theta,\beta): \beta \in B \}.
        \end{align*}    
        Here $B=[0,1]^n$.
    For any $\beta\in \Ndiff$ we have $\nabla f(\cdot , \beta) = D_f(\cdot,\beta) \in \cF_2$. In \cref{lm:verifying_VC} we show that $\cF_2$ is VC-subgraph and Donsker. By \cref{lm:vc_to_donsker} we know that a stochastic equicontinuity condition w.r.t.\ the $L^2$ norm holds, i.e., 
    \begin{align} 
        \label{eq:nabla_f_se}
        \sup_{\beta \in [\delta_t]_{L^2}} \nu_t ( D_f(\cdot, \beta) - D_f(\cdot, \betast) ) = o_p(1)
    \end{align}
    where 
    $[\delta_t]_{L^2}  = \{ \beta: \beta \in \Ndiff, \int \|D_f(\cdot, \beta) - D_f(\cdot, \betast)\|_2 \sq s \dt \leq \delta_t\}$, 
    $\nu_t D_f = \invtroot (P_t - P) D_f= \invtroot \sumtau( D_f(\thetau) - \int D_f s \dt)$.
    Next, we note for (almost every) fixed $\theta$, $ \lim_{\beta \to \betast}\|D_f(\theta , \beta)  - D_f(\theta , \betast) \|\sq =0$ by $\Thetatie(\betast)$ is measure zero (a condition implied by \cref{as:smo}). 
    Moreover, note 
    \begin{align*}
       \lim_{\beta \to \betast} \E[\|D_f(\theta , \beta)  -D_f(\theta , \betast) \|_2 \sq ] = 
       \E \big[ \lim_{\beta \to \betast}\|D_f(\theta , \beta)  - D_f(\theta , \betast) \|_2\sq \big] = 0
    \end{align*}
    where the exchange of limit and expectation is justified by bounded convergence theorem, and
    by \cref{lm:ltwo_to_parameter}, we can replace $[\delta_t]_{L^2}$ with $ [\delta_t] = \{\beta: \beta \in \Ndiff, \|\beta - \betast\|_2\leq \delta_t\}$ in \cref{eq:nabla_f_se}. Finally, note $\nabla  \fbar (\betast) = \E[D_f(\theta,\betast)]$, and 
    if $H_t$ is differentiable at $\beta \in \Ndiff$, then $\nabla f (\thetau,\beta) = D_f(\thetau, \beta)$ for all $\tau\in[t]$. Then 

    \begin{align}
        & \sup_{[\delta_t] \cap \{ \nabla H_t(\beta)\text{ exists} \}} \frac{\| (\nabla H_t - \nabla H) (\beta) - (\nabla H_t - \nabla H) (\betast) \| _ 2 }{\invtroot +  \|\beta - \betast\|_2 } 
        \notag
        \\ & = 
        \sup_{[\delta_t] \cap \{ \nabla H_t(\beta)\text{ exists}  \} } \frac{\| (P_t - P) D_f(\cdot,\beta) - (P_t - P) D_f(\cdot,\betast) \| _ 2 }{\invtroot +  \|\beta - \betast\|_2 } 
        \\
        & \leq 
        \sup_{[\delta_t]} \sqrt{t}  \| (P_t - P) D_f(\cdot,\beta) - (P_t - P) D_f(\cdot,\betast) \| _ 2 
        = o_p(1)  \tag{by \cref{eq:nabla_f_se}}
    \end{align}
    and thus the required stochastic equicontinuity condition holds. 

    Now we are ready to invoke \cref{thm:the_shapiro_thm}. We need to find the 
    three objects, $C, q, \zeta_t$ as in the lemma that characterize the limit distribution.
    The critical cone $C$ is 
    \begin{align}
        C &= \{w \in \Rn: w\tp e_i = 0 \text{ if $i\in I_=$ and $\deltasti > 0$},\quad 
        w \tp e_i \leq  0 \text{ if $i \in I_=$ and $\deltasti = 0$} \}
        \notag
        \\
       & = \{ w: Aw = 0\} \tag{\cref{as:scs}}
    \end{align}
    where $A \in \R^{|I_=| \times n}$ whose rows are $\{ e_i \tp, i \in I_=\}$. From here we can see the role of \cref{as:scs} is to ensure the 
    critical cone is a hyperplane, which ensures asymptotic normality of $\betagam$.

    If $|I_=| = 0$, i.e., $\betast$ lies in the interior of $B$, then $\cP$ is identity matrix, and the limit distribution is ture. 

    Now assume $|I_=| \geq 1$.
    Note $A A\tp $ is an identity matrix of size ${|I_=|}$ and $A\tp A 
    = \Diag(\indi(i \in I_=))  = \Diag(\indi (\betasti = 1))$.
        The optimal Lagrangian multiplier is unique and so the piecewise quadratic function $q$ is $q(w) = w\tp \cH w$. Finally, the gradient error term is 
    \begin{align}
        \zeta_t = \frac1t \sumtau\big( \mutau - \mubarst \big).
    \end{align}
    The unique minimizer of $w\mapsto \frac12 w\tp \cH w + \zeta w$ over $\{w:Aw=0\}$ is $
    - (\cP \cH \cP)^\dagger  \zeta$ where 
    \begin{align*}
        \cP = I_n - A\tp (A A\tp ) ^{\dagger} A = \Diag( \indi(i \in \Icc))= \Diag(\indi (\betasti < 1)).
    \end{align*}
    For completeness, we provide details for solving this quadratic problem. 
    By writing down the KKT conditions, the optimality condition is 
    \begin{align*}
        \begin{bmatrix}
            \cH & A\tp 
            \\
            A & 0
        \end{bmatrix} 
        \begin{bmatrix}
            w 
            \\
            \lambda 
        \end{bmatrix}
        =
        \begin{bmatrix}
            -\zeta 
            \\
            0
        \end{bmatrix}
        \implies 
        \begin{bmatrix}
            w 
            \\
            \lambda 
        \end{bmatrix} =
        \begin{bmatrix}
           - (\cH\inv - \cH\inv A\tp (A\cH\inv A\tp)\inv A \cH \inv) \zeta 
            \\
           - ((A\cH\inv A\tp)\inv A \cH \inv) \zeta
        \end{bmatrix}
    \end{align*}
    where $\lambda \in \R^{|I_=|}$ is the Lagrangian multiplier.
    By a matrix equality, for any symmetric positive definite $\cH$ of size $n$ and $A\in \R^{|I_=|\times n}$ of rank $|I_=|$, it holds
    \begin{align}
        \cH\inv - \cH\inv A\tp (A\cH\inv A\tp)\inv A \cH \inv =  P_A ( P_A \cH  P_A)^\dagger  P_A = ( P_A\cH P_A)\pinv
    \end{align}
    with $ P_A = I_n - A\tp (A A\tp ) ^{\dagger} A$.
    We conclude that the asymptotic expansion 
    \begin{align}
        \label{eq:expansion_beta}
        \sqrt{t} (\betagam - \betast) =  \frac{1}{\sqrt t} \sumtau D_\beta(\thetau) + o_p(1)
    \end{align}
    holds, where $$D_\beta(\theta) = - (\cP \cH \cP)^\dagger  (\must - \mubarst ) ,$$ and that the asymptotic distribution of $\sqrt{t}(\betagam - \betast)$ is $\cN(0, 
    \Sigma_\beta)$ with $\Sigma_\beta = \E[D_\beta D_\beta \tp] $. Note $\E[D_\beta] = 0$.


    \emph{Proof of $\betagam \toprob \betast$.} This follows from \cref{thm:the_shapiro_thm}. 

    \emph{Proof of Asymptotic Distribution for pacing multiplier $\beta$.} This follows from the above discussion.

    \emph{Proof of Asymptotic Distribution for revenue $\REV$.} 
    We use a stochastic equicontinuity argument.
    Given the item sequence $\gam = (\theta^1, \theta^2 \dots )$, define the (random) operator 
    \begin{align*}
    \nu_t g = \sqrt t (P_t - P)g = \frac{1}{\sqrt t} \sumtau ( g(\thetau) - \E[g]) 
    \;.
    \end{align*}
    where $g:\Theta \to \R$, $ \E[g]=\int g s \dt$.
    Note $\pst(\theta) = \max_i \betasti \vithe = f(\theta,\betast)$, $\REV^* = P f(\cdot ,\betast)$, 
    $\ptau = f(\thetau, \betagam)$ and 
    $\REV^\gam = P_t f(\cdot, \betagam)$ we obtain the decomposition
    \begin{align*}
     \sqrt t (\REV^\gam -  \REV^* ) =
     \underbrace{ \frac{1}{\sqrt t} \sumtau \big( f(\thetau, \betast) - \fbar (\betast) \big) }_{ = : \I_t}
     + 
     \underbrace{    \nu_t (f(\cdot , \betagam  ) - f( \cdot,\betast ))
     }_{ = : \II_t }
     \\
     + 
     \underbrace{ \sqrt{t}( \bar{f} (\betagam) - \bar{f} (\betast) )}_{=: \III_t}
    \end{align*}
    
    For the term $\I_t$, it can be written as
$
        \I_t = \nu_t ( \pst(\cdot) - \REV^* )
$.
    By the linear representation for $\betagam - \betast$ in \cref{eq:expansion_beta}, applying the delta method, we get the linear representation result 
    \begin{align*}
        \III_t  = \frac{1}{\sqrt t} \sumtau \nabla \fbar(\betast) \tp D_\beta(\thetau) + o_p(1)
        =  \frac{1}{\sqrt t} \sumtau (\mubarst)  \tp D_\beta(\thetau) + o_p(1) 
    \end{align*}

    We will show $\II_t= o_p(1)$. The difficulty lies in that the operator $\nu_t$ and the pacing multiplier $\betagam$ depend on the same batch of items. This can be handled with the stochastic equicontinuity argument. The desired claim $\II_t = o_p(1)$ follows by verifying that the function class $\cF_1 = \{ \theta \mapsto  f(\cdot, \beta ) :\beta \in B \}$ (same as that defined in \cref{lm:verifying_VC}) is  VC-subgraph and Donsker. This is true by \cref{lm:verifying_VC}.
    By \cref{lm:donsker_to_SE} we know for any $\delta_t \downarrow 0$,
    \begin{align}
        \sup_{w \in [\delta_t]_{L^2}} \nu_t ( f(\cdot, w) - f(\cdot, \betast) ) = o_p(1)
        \label{eq:f_donsker}
    \end{align}
    where $[\delta_t]_{L^2} = \{ \beta :\beta  \in B, \int (f(\cdot, \beta ) - f(\cdot, \betast))\sq s \dt \leq \delta_t\}$. Noting that for all $\beta, w, \theta$, it holds $|f(\theta, \beta) - f(\theta,w)| \leq \vbar \|\beta -w\|_\infty$, we know that $\int\left[f(\cdot, \beta)-f\left(\cdot, \betast \right)\right]^2  s \dt \rightarrow 0 $ as $ \beta \rightarrow \betast$. Then by \cref{lm:ltwo_to_parameter}, we know \cref{eq:f_donsker} holds with $[\delta_t]_{L^2} $ replaced with $[\delta_t] = \{ \beta: \beta \in B,  \|\beta  - \betast\|_2 \leq \delta_t\} 
    $. Combined with the fact that $\betagam \toprob \betast$, by \cref{lm:se} we know $\II_t = o_p(1)$.
    
    To summarize, we obtain the linear expansion
    \begin{align}
        \label{eq:expansion_rev}
  \sqrt t (\REV^\gam -  \REV^* ) 
 =
    \frac{1}{\sqrt t} \sumtau ( 
    \pst(\thetau) - \REV^*
    + (\mubarst) \tp D_\beta(\thetau)
    ) + o_p(1) .
    \end{align}
    Let $D_\REV(\t) =     \pst(\thetau) - \REV^*
    + (\mubarst) \tp D_\beta(\thetau)
     $.

    We complete the proof of \cref{thm:clt}.
\end{proof}

\subsection{Proof of \cref{cor:set_I_estimate}}
\begin{proof}[Proof of \cref{cor:set_I_estimate}]
    
    For $i$ such that $\betasti = 1$, we know $\betagami - 1 = o_p(\frac{1}{\sqrt t})$ holds. Then 
    $ 
        \P( \betagami < 1-\varepsilon_t ) = \P( o_p(1) > \sqrt{t} \varepsilon_t)  \to 0,
    $
    using the smoothing rate condition $ \sqrt{t} \varepsilon_t\to c \in (0,\infty] $.
    For $i$ such that $\betasti < 1$, we know $\betagam - \betasti = o_p(1)$ by consistency of $\betagam$. Then 
    $ 
        \P( \betagami < 1-\varepsilon_t ) = \P (o_p(1) < (1-\betasti)  - \varepsilon_t) \to 1
    $
    by $\varepsilon_t = o(1)$ and $1-\betasti > 0$.
\end{proof}

\subsection{Proof of Normality of NSW, $u$, and $\delta$ in FPPE}  \label{sec:full:cor:clt_u_and_nsw}

We define NSW and utility in limit FPPE.
We use $(. )^*$ to denote limit FPPE quantities. 
Given an FPPE $(\betast, \pst)$, we define by
\begin{align}
    &\deltasti \defeq b_i - \int \pst  s\xsti \dt \; , \quad \mubarsti \defeq \int v_i s\xsti \dt \;,
    \notag
    \\
    &\usti \defeq \mubarsti + \deltasti \;.
\end{align}
the \emph{leftover budget}, the \emph{item utility} and the \emph{total utility} of buyer $i$. 
Let $\deltast, \mubarst, \ust$ be the vectors that collect these quantities for all buyers.
In FPPE it holds that 
\begin{align*}
    \usti = b_i / \betasti.
\end{align*}

We next define these quantities in finite FPPE.
To emphasize dependence on the item sequence $\gamma$, we use $(.)^\gamma$ to denote equilibrium quantities in $\oFPPE(b,v,1/t, \gamma)$. 
We let $(\betagam,\pgam)$ be an observed FPPE with $\xgam = (\xgam_1,\dots)$. 
The leftover budget $\delta^\gam_i \defeq b_i- \sigma \sum_\tau \pgamtau \xgamtaui$, item utility $\mubar_i \defeq \sigma \sumtau \vitau \xgamtaui $ and total utility $\ugami \defeq \delta^\gam_i + \mubar_i $ are defined similarly. Let $\deltagam, \mubar, \ugam$ be the vectors that collect these quantities for all buyers. The observed revenue is $\REVgam \defeq \sigma \sumtau \ptau$, and NSW is $\NSWgam \defeq \sumiton b_i \log \ugami$.

\begin{corollary}
    \label{cor:clt_u_and_nsw}

     Under the same conditions as \cref{thm:clt},
       $\sqrt{t} (\ugam - \ust) $, $\sqrt t(\deltagam - \deltast)$ and $\sqrt{t}(\NSW^\gam - \NSW^*) $
       are asymptotically normal with 
(co)variances
$\Sigma_u \defeq\allowbreak \Diag{({b_i}/({\betasti})\sq)} \allowbreak \Sigma_\beta  \allowbreak    \Diag({b_i}/({\betasti})\sq)$, 
$\Sigma_\delta \defeq \allowbreak  (I_n - \cH(\cH_B)\pinv) \allowbreak \Omega \allowbreak (I_n - \cH(\cH_B)\pinv) \tp$, and
$\sigma^2_\NSW \defeq\Vec(b_i/\betasti)\tp \allowbreak \Sigma_\beta \allowbreak\Vec(b_i / \betasti)$, respectively.
\end{corollary}
\begin{proof}[Proof of \cref{cor:clt_u_and_nsw}]

    \emph{Proof of Asymptotic Distribution for individual utility $u$.} We use the delta method; see Theorem 3.1 from \citet{van2000asymptotic}. Note $ \ust = g(\betast)$ with $g: \Rn \to \Rn, g(\beta) = [b_1/\beta_1,\dots, b_n / \beta_n] \tp$. By \cref{eq:expansion_beta}, it holds 
    \begin{align}
        \label{eq:expansion_u}
        \sqrt{t} (\ugam -\ust ) =  \frac{1}{\sqrt t} \sumtau \nabla g(\betast) D_\beta (\thetau) + o_p(1) .
    \end{align}
    Finally, noting $\nabla g(\betast) = \Diag(-b_i / (\betasti)\sq )$ we complete the proof.

    \emph{Proof of Asymptotic Distribution for Nash social welfare $\NSW$.} We use the delta method. Note $ \NSW^* = g(\betast)$ with $g: \Rn \to \R, g(\beta) =  \sumiton b_i \log(b_i / \beta_i) $. By \cref{eq:expansion_beta} it holds 
    \begin{align}
        \label{eq:expansion_nsw}
        \sqrt{t} (\NSW^\gam - \NSW^*) =  \frac{1}{\sqrt t} \sumtau \nabla g(\betast) \tp D_\beta (\thetau) + o_p(1) . 
    \end{align}
    Finally, noting $\nabla g(\betast) = \Vec(b_i / \betasti)$.

\emph{Proof of Asymptotic Distribution for leftover budget $\delta$.}
This is a direct consequence of Theorem 4.1 in \citet{shapiro1989asymptotic}. By that theorem, it holds that 
\begin{align*}
    \sqrt{t} 
    \begin{bmatrix}
        \betagam - \betast 
        \\
        \delta^\gam_{I_=} - \delta^*_{I_=} 
    \end{bmatrix}
    \tod 
    \cN(0, 
    \Sigma_{\text{joint}})
\end{align*} 
with 
\begin{align*}
    \Sigma_{\text{joint}} = 
    \begin{bmatrix}
        \cH & A  \tp
        \\
        A &  0
    \end{bmatrix} \inv
    \begin{bmatrix}
       \Omega & 0
        \\
        0 &    0     \end{bmatrix}
    \begin{bmatrix}
        \cH & A \tp
        \\
        A & 0
    \end{bmatrix} \inv
    = \begin{bmatrix}
        (\cH_B)\pinv  \Omega (\cH_B)\pinv &   
        [Q \Omega (\cH_B)\pinv  ]\tp 
        \\
        Q \Omega (\cH_B)\pinv  
        & 
        Q \Omega Q\tp  
    \end{bmatrix}
\end{align*}
where $A \in \R^{|I_=| \times n}$ whose rows are $\{ e_i \tp, i \in I_=\}$, $Q =  (A\cH\inv A\tp)\inv A \cH \inv \in \R^{|I_=|\times n}$ and $\Omega = \cov (\must)$.
By a matrix equality, noting matrix $A$'s rows are distinct basis vectors, it holds 
\begin{align*}
    (A\cH\inv A\tp)\inv A \cH \inv  = A (I_n - \cH (\cH_B)\pinv)
\end{align*}
Moreover, for other entries of $\deltagam$, i.e., $\deltagam_{\Ic}\,$, their asymptotic variance will be zero.
The matrix $(I_n - \cH (\cH_B)\pinv) \Omega (I_n - \cH (\cH_B)\pinv)\tp$ is zero at the $(i,j)$-th entry if $i$ or $j\in \Ic$.
Summarizing, the asymptotic variance of $\sqrt t(\deltagam - \deltast)$ is $(I_n - \cH (\cH_B)\pinv) \Omega (I_n - \cH (\cH_B)\pinv)\tp$.

An alternative proof is by the delta method and a stochastic equicontinuity argument. To summarize, it holds $\sqrt t ({\deltagam - \deltast}) \tod \cN(0, (I_n - \cH(\cH_B)\pinv) \Omega (I_n - \cH(\cH_B)\pinv) \tp )$ and the linear expansion 
\begin{align} \label{eq:delta_variance}
    \sqrt t (\deltagam - \deltast) = \invtroot \sumtau (I_n - \cH (\cH_B)\pinv) (\must (\thetau)- \mubarst) + o_p(1)
\end{align}
holds. In the case where $I_= = \emptyset$, i.e., $\deltasti = 0$ for all $i$, we have $\cH_B =\cH$ and so $I_n - \cH (\cH_B)\pinv = 0$.

We complete the proof of \cref{cor:clt_u_and_nsw}
\end{proof}

\subsection{Proof of Fast Convergence of Pacing Multipliers}

\begin{lemma} \label{lm:fast_beta_convergence}
    Let \cref{as:smo} and \cref{as:scs} hold. Then $\P(\betagami = 1)\to 1$ for all $i \in I_=$, and $\P(\deltagami = 0) \to 1$ for all $ i\in I_<$.
\end{lemma}
\begin{proof}[Proof of \cref{lm:fast_beta_convergence}]
   \cref{eq:delta_variance} implies $\deltagam_i - \deltasti = o_p(1)$ for all $i\in I_=$. \cref{as:scs} implies $\deltasti > 0$ for all $i \in I_=$. Combining the two gives $\P(\deltagami > 0, \forall i \in I_=) \to 1$, which implies the desired claim by the strict complementarity between $\betagam \leq 1$ and $\deltagam \geq 0$.
   The proof of $\P(\deltagami = 0) \to 1$ follows similarly.
\end{proof}

\subsection{Proof of \cref{thm:rev_localopt}}
\label{sec:proof:rev_local_asym_risk}

\begin{proof}[Proof of \cref{thm:rev_localopt}]
    According to the discussion in \cref{sec:proof:thm:nsw_aym_risk}, 
    we calculate the derivative of the map $\alpha \mapsto \REVst_{\alpha,g}$ at $\alpha=0$.

    For a given perturbation ${(\alpha,g)}$, we let $\pst_{\alpha,g}$ and $\REVst_{\alpha,g}$ be the limit FPPE price and revenue under supply distribution $s_{\alpha,g}$.
    Let $S_{\alpha,g}(\theta) = \nabla _\alpha \log s_{\alpha,g}(\theta)$ be the score function. 
    So $\nabla_\alpha s_{\alpha,g} = s_{\alpha,g} S_{\alpha,g}$
    and $\int S_{\alpha,g} s_{\alpha,g} \diff \t = 0$.
    Obviously with our parametrization of $s_{\alpha,g}$ we have $S_{0,g}(\theta) = g(\theta)$ by \cref{eq:perturbed_is_roughly_expo}.
    We next find the derivative of $\alpha\mapsto  \REVst_{\alpha,g} $ at $\alpha=0$.
    Recall $f$ is defined as $f(
        \theta,
        \beta) = 
        \max_i 
        \vithe \betai
    $ and the price is produced by the highest bid, i.e., $\pst_{\alpha,g}(\theta) = \max_i \betast_{\alpha,g}\vithe = f(\theta,\betast_{\alpha,g})$.
    \begin{align}
        & \nabla_ \alpha \REVst_{\alpha,g} = \nabla_ \alpha \int f(\theta,\betast_{\alpha,g}) s_{\alpha,g}(\theta)\diff \theta
        \notag 
        \\
        & = \int [\nabla_\beta f(\theta, \betast_{\alpha,g})\nabla_ \alpha \betast_{\alpha,g} + f(\theta,\betast_{\alpha,g}) S_{\alpha,g}(\theta)] s_{\alpha,g}(\theta)\diff \theta
        \notag 
        \; .
    \end{align}
    Above we exchange the gradient and the expectation and then apply the chain rule.
    By a perturbation result by Lemma 8.1 and Prop.\ 1 from \citet{duchi2021asymptotic}, under \cref{as:smo} and \cref{as:scs}, 
    \begin{align*}
        \nabla_\alpha \betast_{\alpha,g}|_{\alpha = 0} = - (\cH_B)\pinv \Sigma_{\must , g } 
    \end{align*} with 
    $\Sigma_{\must , g }  = \E_s[(\must(\theta) - \mubarst)g(\theta)\tp] $.
    Plugging in $\E_s[\nabla_\beta f(\theta,\betast_{0,g})] = \mubarst$, $f(\theta,\betast_{0,g}) = \pst(\theta)$ and $S_{0,g} = g$,
    we obtain 
    \begin{align*}
        \nabla_\alpha\REVst _{\alpha,g} |_{\alpha= 0} 
        & =-(\mubarst)\tp (\cH_B)\pinv \Sigma_{\must, g}  + \E_s[ (\pst (\theta) - \REVst) g(\theta) ] 
        \\
        & = \E\big[ \big( -(\mubarst)\tp (\cH_B) \pinv (\must(\theta) - \mubarst) + (\pst(\theta) - \REVst)\big) g (\theta)\big] 
        \\
        & = \E[ D_\REV(\theta) g(\theta)]
        . 
    \end{align*}
    Now we have the two components required to invoke the local minimax result. 
    Given a symmetric quasi-convex loss $L:\R\to\R$, we recall the local asymptotic risk for any procedure $\{\hat r _t : \Theta^t \to \R\}$ that aims to estimate the revenue:
    \begin{align*}
       & \LAR_\REV ( \{\hat r_t \}) =
        \\ 
        & 
        \sup_{ g\in G_d, d\in \mathbb{N}}
        \lim_{c\to \infty}
        \liminf_{t \to \infty}
        \sup_{\|\alpha\|_2\leq \frac{c}{\sqrt t}}
        \E_{s_{\alpha,g} ^{\otimes t}}[L(\sqrt{t} (\hat r_t  - \REVst_{\alpha,g} ))] \;. 
    \end{align*} 
    Following the argument in \citet[Sec.\ 8.3]{duchi2021asymptotic} it holds 
    \begin{align*}
       \inf_{ \{\hat r_t\}} \LAR_\REV ( \{\hat r_t \}) \geq \E[L(\cN(0, \E[D_\REV\sq(\theta)]))] \;.
    \end{align*}
    We complete the proof of \cref{thm:rev_localopt}.
 
\end{proof}

\subsection{Proof of \cref{thm:variance_estimation}}
\label{sec:proof:thm:variance_estimation}
\begin{proof}[Proof of \cref{thm:variance_estimation}]

    \emph{Proof of $\hat\Sigma_\beta \toprob \Sigma_\beta$.}
    We first show $\hat \cH \toprob \cH$ by verifying conditions in \cref{lm:hessian_estiamtion}.
    All conditions are easy to verify except the stochastic equicontinuity condition. By \cref{lm:verifying_VC} we know the SE condition holds. We conclude $\hat \cH \toprob \cH$.
    Next we show $\P(\hat \cP = \cP ) \to 1$. 
    This follows from \cref{cor:set_I_estimate}.
    We conclude $\P(\hat\cP = \cP) \to 1$.
    We now show $(\hat \cP \hat \cH \hat \cP)\pinv  \toprob ( \cP\cH \cP)\pinv $. 
    Wlog let $I_= = [k]$, and $I_< = [k+1,\dots, n]$.
    For any $\epsilon > 0$,
    $ 
        \P( \| ( \hat \cP \hat \cH \hat \cP)\pinv  - ( \cP\cH \cP)\pinv \|_F > \eps) 
         \leq \P( \|( \hat \cP \hat \cH \hat \cP)\pinv  - ( \cP\cH \cP)\pinv \|_F > \eps, \hat \cP = \cP) + \P(\hat \cP\neq \cP)  
         = \P( \| [\hat \cH_{<} ] \inv  - [\cH_{<}]\inv \|_F > \eps) + \P(\hat \cP\neq \cP)  
        \to 0
    $
    by $\hat\cH \toprob \cH$.
    Next we show $ \cov(\must) = \cov(\nabla f(\t, \betast))$ can be consistently estimated by 
    $   \frac1t\sumtau ( \mutau - \mubar)
       ( \mutau - \mubar) \tp
    $,
    $\mu^\tau = [x_1^{\tau} v_1^\tau, \dots, x_n^{\tau} v_n^\tau]\tp$, and $\bar \mu= \frac1t \sumtau \mutau$.
    This follows by \cref{thm:covnablaf_estimation}.
    We rewrite $\hat\Sigma_\beta$ as 
    $
        \hat\Sigma_\beta =  (\hat \cP \hat \cH \hat \cP)\pinv (\frac1t \sumtau (\mutau - \mubar) \ot ) (\hat \cP \hat \cH \hat \cP)\pinv
        $
        which 
        converges in probability to 
        $
        ( \cP  \cH  \cP)\pinv \cov(\must)  ( \cP  \cH  \cP)\pinv = \Sigma_\beta
        $. 

    \emph{Proof of $\hat\sigma\sq_\REV \toprob \sigma\sq_\REV$.}  
    We rewrite 
    $
        \sigma\sq_\REV = 
        \underbrace{\E[(p\st - \REV\st)\sq ]}_{\I_t} 
        + \allowbreak
        \underbrace{(\mubarst)\tp  ( \cP  \cH  \cP)\pinv  \cov(\must) ( \cP  \cH  \cP)\pinv \mubarst }_{\II_t}
        + \allowbreak
        \underbrace{2 \E[ (p\st - \REV\st)(\must - \mubarst)] \tp ( \cP  \cH  \cP)\pinv \mubarst }_{\III_t} 
    $
    and 
    $
    \hat \sigma\sq_\REV  
    = \allowbreak 
    \underbrace{\frac1t\sumtau (\ptau - \REV^\gam)\sq}_{\hat\I_t} + 
    \underbrace{(\mubargam)\tp (\hat \cP \hat \cH \hat \cP)\pinv \hat\Omega  (\hat \cP \hat \cH \hat \cP)\pinv \mubargam }_{\hat \II_t}
     + \allowbreak
    \underbrace{2 \bigg(\frac1t \sumtau 
    (\ptau - \REV^\gam) (\mutau - \mubargam)\bigg) \tp (\hat \cP \hat \cH \hat \cP)\pinv  \mubargam }_{\hat \III_t}\;.
    $
    We have $\hat\I_t\toprob \I_t$ by invoking \cref{lm:verifying_VC}, applying a uniform LLN and using the fact that $\betagam \toprob \betast$. And $\hat \II_t\toprob \II$ holds by 
    $\mubar\toprob \mubarst$, $(\hat \cP \hat \cH \hat \cP)\pinv  \toprob ( \cP\cH \cP)\pinv$ and $\hat \Omega \toprob \Omega$, and applying 
    Slutsky's theorem.
    Finally, $\hat\III_t\toprob \III$ by $\cF_1 \cdot \cF_2$ is Donsker by \cref{lm:donsker_multiplication} and thus a uniform law of large number holds, and that $\betagam\toprob\betast$. 

    We complete the proof of \cref{thm:variance_estimation}.
    \end{proof}
\subsection{Proof of Simplified Estimation Results}
\label{sec:proof:simplified_var_est}

\begin{proof}[Proof of \cref{thm:simplified_rev_var}]
    
    \new{
    First we show $D_\REV(\t) = \pst(\t) - \mubarst \tp (\cP \cH \cP)\pinv \must(\t) = \tilde p^* (\t) $. 
    It suffices to show $(\betast - (\cP\cH\cP)\pinv \mubarst )\tp \must(\t) = \tilde p ^* (\t)$.
    The cases of  $I_< = \emptyset$ or $I_= =\emptyset$ can be handled easily. We assume both of them are nonempty. Also let $I_=  = [k]$, and $I_< = [k+1, \dots, n]$ without loss of generality.
    The finite inverse moment assumption on $\bidgap$ implies $\nabla\sq \E[f(\t,\betast)] \betast = 0$, and $\cH = \Diag(b_i / (\betasti)\sq)$.
    So $\betast = [1_k; \betast_<]$ where $\betast_<$ is the subvector corresponding to $I_<$. Let $\cH_<$ be the lower-right matrix of $\cH$ corresponding to $I_<$.
    \cref{lm:fppe_relation}, \cref{eq:beta_u_relation_fppe} and $\nabla\sq \E[f(\t,\betast)] \betast = 0$ imply $(\cP\cH\cP)\pinv \mubarst = [0_k; \betast_<]$.
    So $\betast - (\cP\cH\cP)\pinv \mubarst = [1_k; 0_{n-k}]$. And $[1_k; 0_{n-k}] \tp \must(\theta) = \sum_{i\in I_=} \mu^*_i(\theta) = \tilde p ^* $.
    The simplified formula for $\Sigma_\beta$ follows easily.

    Next we show $\hat\sigma\sq_{\REV, \textup{sim}} \toprob \sigma\sq_\REV$. 
    Suppose $I_= = \emptyset$, then $\P(\hat I_= = I_=) \to 1$ implies $\hat \sigma_\REV \sq \to 0$. 
    Now suppose $I_= \neq \emptyset$.
    In \cref{lm:verifying_VC} we have shown both $\cF_1$ and $\cF_3$ are VC-subgraph, and thus a uniform law of large number holds.
    Since $\P(\hat I_= = I_=) \to 1$, we analysis is carried out under the event $\hat I_= = I_=$.
    Let $f(\t,\b)$ be the function in $\cF_1$ and $g$ in $\cF_3$. Let $h(\t,\b) = f(\t,\b)g(\t,\b)$.
    Under the event $\hat I_= = I_=$, we have that $\hat \sigma \sq _\REV$ is the sample variance of $h(\thetau, \betagam)$, $\tau \in [t]$, and $\hat \sigma _\REV\sq$ is the variance of the random variable $h(\t, \betast)$.
    The consistency of the simplified revenue variance estimator follows by the uniform law of large number.
    The consistency of the simplified estimator for $\Sigma_\beta$ follows easily.
}
\end{proof}
\subsection{Proof of \cref{thm:clt_ab_testing}}\label{sec:proof:thm:clt_ab_testing}
\begin{proof}[Proof of \cref{thm:clt_ab_testing}]
By the EG characterization of FPPE, we know that $\betagam(1)$, the pacing multiplier of the observed FPPE $\oFPPE \big(\pi b, v(1), \frac{\pi}{t_1}, \gamma(1)  \big)
$, solves the following dual EG program
\begin{align}
    \min_{B} \frac{1}{t_1} \sum_{\tau = 1}^{t_1} \max_i \vithetau \betai - \sumiton b_i \log(\betai)
\end{align}

The major technical challenge is that the number of summands in the first summation is also random. 
Given a fixed integer $k$ and a sequence of items $( \theta^{1,1}, \dots, \theta^{1,k})$, define
    \begin{align*}
        & \beta^{\lin, k} (1) = \betast(1) +  \frac1{{k}} \sum_{\tau = 1}^{k} D_\beta(1, \theta^{1,\tau}), 
        \\ 
        & \beta^k (1)= \text{the unique pacing multiplier in $\oFPPE(b,v(1), k\inv, ( \theta^{1,1}, \dots, \theta^{1,k}))$}
    \end{align*}
    Here $D_\beta(1,\cdot) =  - (\cH_B(1))\pinv (\must(1,\cdot) - \mubarst(1))$ 
    where
$\cH_B(1), \must(1, \cdot)$ and $\mubarst(1)$ are the projected Hessian in \cref{def:cHB}, item utility function in \cref{eq:def:must}, and total item utility vector in \cref{eq:def:mubarst} in the limit market $\FPPE(b, v(1), s)$.
Note $\E[D_\beta(1,\cdot)] = 0$.
Note $\betagam(1) = \beta^{t_1}$ since scaling the supply and the budget at the same time does not change the equilibrium pacing multiplier.
We introduce the following asymptotic equivalence results:
\begin{lemma} \label{lm:asym_equivalence_beta}
     Recall $t_1 \sim \text{Bin}(\pi, t)$. If \cref{as:scs} and \cref{as:smo} hold for the limit market $\FPPE(b,v(1),s)$, then
     \begin{itemize}
         \item $ \sqrt{t} (\betagam(1) - \beta^{\lin, t_1}) = o_p(1)$ as $t \to \infty$.
         \item $ \sqrt{t} (\beta^{\lin, t_1} - \beta^{\lin, \lfloor \pi t \rfloor}) = o_p(1) $ as $t \to \infty$.
     \end{itemize}
  
     Here $\lfloor a \rfloor $ is the greatest integer less than or equal to $a \in \R$. 
     A similar result holds for the market limit $\FPPE(b,v(0),s)$ and the influence function $D_\beta(0, \cdot)$ is defined similarly.
\end{lemma}

With \cref{lm:asym_equivalence_beta}, we write
\begin{align*}
 & \sqrt{t} (\hat \tau_\beta - \tau _\beta)
 \\
 & = \sqrt t (\betagam(1) - \betast(1) ) - \sqrt{t} (\betagam(0) - \betast(0))
 \\ 
 & = \sqrt t \bigg( \frac1{\sqrt{\lfloor \pi t \rfloor}} \sum_{\tau = 1}^{\lfloor \pi t \rfloor} D_\beta(1, \theta^{1,\tau})
  -   \frac1{\sqrt{\lfloor (1-\pi) t \rfloor}} \sum_{\tau = 1}^{\lfloor (1-\pi) t \rfloor} D_\beta(0, \theta^{0,\tau})\bigg) + o_p (1) \tag{\cref{lm:asym_equivalence_beta}}
 \\
 & \tod \cN \bigg(0,\frac1{\pi }\var[ D_\beta(1,\cdot)] 
 + \frac1{{(1-\pi) }} \var[ D_\beta(0,\cdot)]\bigg). \tag{independence between $\{\theta^{1,\tau}\}_\tau$ and $\{\theta^{0,\tau}\}_\tau$}
\end{align*}

\emph{Proof of CLT for $\tau_\beta$.} It follows from the above discussion.

\emph{Proof of CLT for $\tau_u$.} We begin with the linear expansion \cref{eq:expansion_u} and repeat the same argument.

\emph{Proof of CLT for $\tau_\REV$.} We begin with the linear expansion \cref{eq:expansion_rev} and repeat the same argument.

\emph{Proof of CLT for $\tau_\NSW$.} We begin with the linear expansion \cref{eq:expansion_nsw} and repeat the same argument.

We complete the proof of \cref{thm:clt_ab_testing}.
\end{proof}

In order to prove \cref{lm:asym_equivalence_beta}, we will need the following lemma.
\begin{lemma} \label{lm:random_index_op}
    If $X_t = o_p(1)$ and $T \sim \text{Bin}(\pi, t)$ and $T$ and the sequence are independent, then $X_T = o_p(1)$.
\end{lemma}
\begin{proof}[Proof of \cref{lm:random_index_op}]
By $X_t = o_p(1)$ we know for all $\eps > 0$ it holds $\P( |X_t| > \eps ) \to 0$, or equivalently $ \sup_{k \geq t} \P( |X_k| > \eps) \to 0$ as $ t \to \infty$. By a concentration for binomial distribution, we know for all $\delta> 0$, it holds $\P( |T - \pi t| > \delta \pi t ) \leq 2\exp(- \delta\sq  \pi t / 3)$. Now write 
\begin{align*}
\P( |X_T| > \epsilon)
&\leq \P( |X_T| > \eps, T\in (1\pm \delta) \pi t ) + \P(T \notin (1\pm \delta) \pi t)
\\
&\leq \sum_{k\in  (1\pm \delta) \pi t } \P(|X_k| > \eps) \P(T = k)+2 \exp(- \delta \sq  \pi t / 3)
\\
&\leq \sup_{k\geq (1-\delta) \pi t} \P(|X_k| > \eps) + 2\exp(- \delta  \sq \pi  t / 3)
\to 0 \text{ as $t \to \infty$}
\end{align*}
where in the second inequality we use the independence between $T$ and the sequence.
We conclude $X_T = o_p(1)$,  completing proof of \cref{lm:random_index_op}.
\end{proof}

\begin{proof}[Proof of \cref{lm:asym_equivalence_beta}]
    The first statement uses the independence between $t_1$ and the items $(\theta^{1,1},\theta^{1,2},\dots)$. 
    Define $R(k) = \sqrt{t} (\beta^k(1) - \beta^{\lin,k}(1))$. By \cref{eq:expansion_beta}, we have $R(k) = o_p(1)$ as $k\to\infty$.
    With this notation, the first statement is equivalent to $R(t_1) = o_p(1)$ where $t_1 \sim \text{Bin}(\pi, t)$, which holds true by \cref{lm:random_index_op}.
    
    The second statement is equivalent to $\sqrt{\lfloor \pi t \rfloor}  \big(\beta^{\lin, t_1} (1)- \beta^{\lin ,\lfloor \pi t \rfloor} (1)\big) = o_p(1)$.
    To prove this we use a Komogorov's inequality.
    By Theorem 2.5.5 from \citet{durrett2019probability}, for any $\eps > 0$, (let $\sigma_{D_\beta} = \E[\|D_\beta(1,\theta)\|_2\sq]^{1/2}$)
    \begin{align*}
        \P \bigg (\sqrt{\lfloor \pi t \rfloor} \sup_{(1-\eps) \lfloor \pi t \rfloor \leq m \leq (1 + \eps) \lfloor \pi t \rfloor} \|\beta^{\lin, m} (1)- \beta^{\lin ,(1-\eps)\lfloor \pi t \rfloor} (1) \| _2 \geq \delta \sigma_{D_\beta}  \bigg) \leq \frac{2 \eps}{\delta\sq}.
    \end{align*}
    Then 
\begin{align*}
    &\P\big( \sqrt{\lfloor \pi t \rfloor} \big\| \beta^{\lin, t_1}(1) - \beta^{\lin ,\lfloor \pi t \rfloor}(1) \big\|_2\geq \delta \big)
    \\
    & \leq \P\big(  \sqrt{\lfloor \pi t \rfloor} \big\| \beta^{\lin , t_1} (1)- \beta^{\lin ,\lfloor \pi t \rfloor} (1) \big\|_2\geq \delta , \; (1 - \eps) \lfloor \pi t \rfloor \leq t_1\leq  (1 + \eps) \lfloor \pi t \rfloor\big)
    \\
    &\quad + 
    \P\big( t_1 \notin \big[(1 - \eps) \lfloor \pi t \rfloor, (1 + \eps) \lfloor \pi t \rfloor\big]\big)
    \\
   & \leq \frac{2\eps \sigma_{D_\beta}\sq}{\delta\sq} + 
    \P\big( t_1 \notin \big[(1 - \eps) \lfloor \pi t \rfloor, (1 + \eps) \lfloor \pi t \rfloor\big] \big) \to\frac{2\eps \sigma_{D_\beta}\sq}{\delta\sq} 
\end{align*}
Finally, since the above holds for all $\epsilon > 0$, we obtain $\sqrt{\lfloor \pi t \rfloor}  (\beta^{\lin, t_1} - \beta^{\lin ,\lfloor \pi t \rfloor}) = o_p(1)$.
    We complete the proof of \cref{lm:asym_equivalence_beta}.
\end{proof}

%% file: experiment_app.tex

\section{Experiments}

\subsection{Hessian Estimation}
\label{app:hessian_estimation}
\textit{Setup.}
We look at the following configuration of markets and the smoothing parameter $d$.
Note we will be evaluating Hessian at a prespecified point and do not need to form any market equilibria in this experiment.
We consider $n=9$ buyers. The item size $t$ ranges from 200 to 5000, at a log scale. 
Budget does not need to be specified.
Buyers' values are drawn from uniform, exponential, or truncated standard normal distributions.
The smoothing parameter $d $ is chosen from the grid [0.10, 0.17, 0.25, 0.32, 0.40, 0.47, 0.55, 0.62, 0.70].
We evaluate the Hessian $\nabla\sq H$ at a pre-specified point $\beta = [0.200, 0.333, 0.467, 0.600, 0.733, 0.867, 1.000]$, and plot the estimated diagonal values, $\hat \cH_{ii}$ for $i \in [7]$, against the number of items $t$.
Under each configuration we repeat for 10 trials. 

\textit{Results.}
See
\cref{fig:smooth_uniform,fig:smooth_exponential,fig:smooth_normal}.
 We see that $d$ represents a bias-variance trade-off. For a small $d$ (0.10, 0.17, 0.25), the variance of the estimated value $\hat \cH_{ii}$ is small and yet bias is large (since the plots seem to be trending to some point as number of item increases). For a large $d$ (0.55, 0.62, 0.70) variance is large and yet the bias is small (the estimates are stationary around some point). It is suggested to use $d \in (0.32, 0.47)$.

\begin{figure}[h!]
    \center
    \includegraphics[scale=.4]{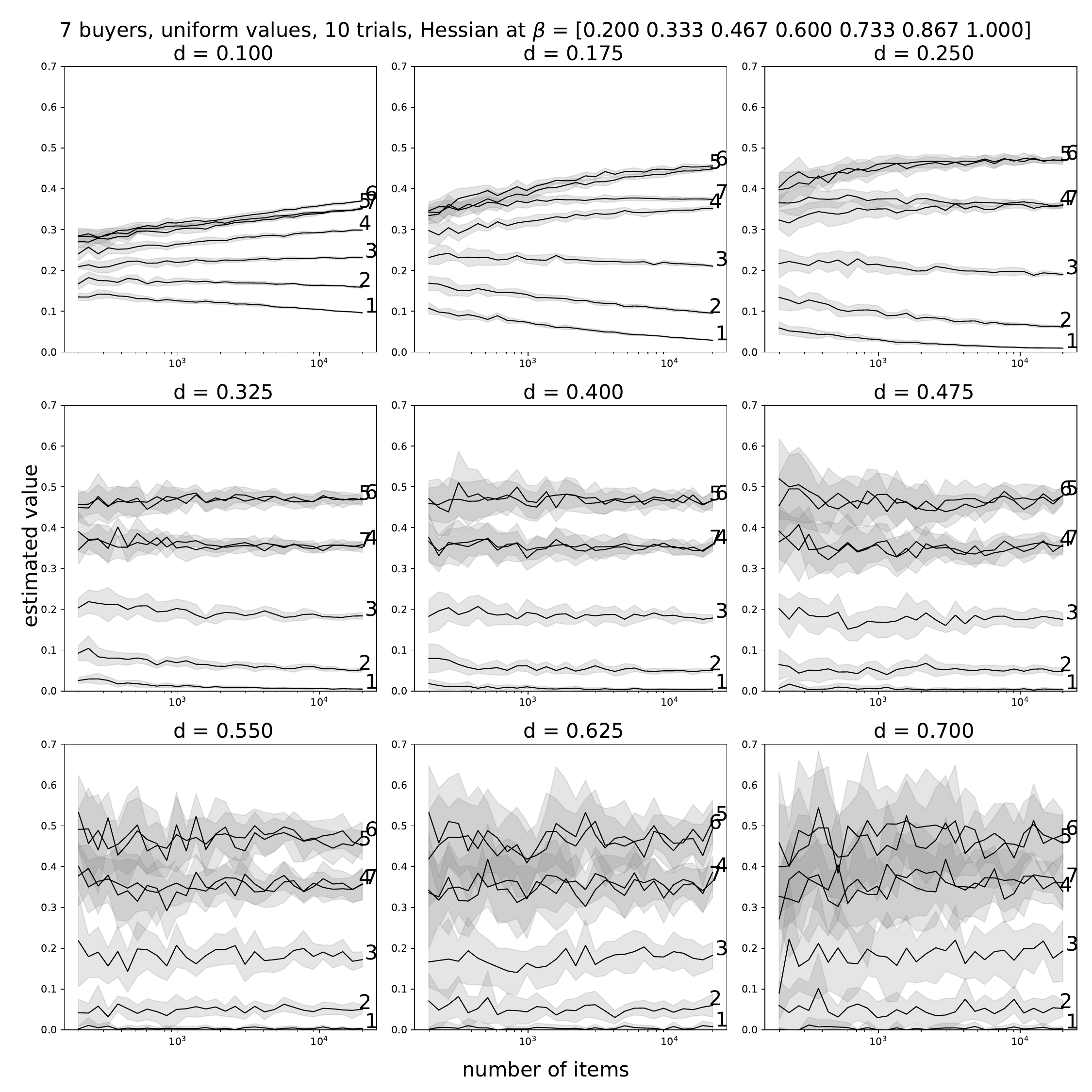}    
    \caption{Effect of smoothing parameter on numerical difference estimation of Hessian. Each curve represents the estimated value of $\cH_{ii}$. Uniform values.}
    \label{fig:smooth_uniform}
    \end{figure}

    \begin{figure}[h!]
        \center
        \includegraphics[scale=.4]{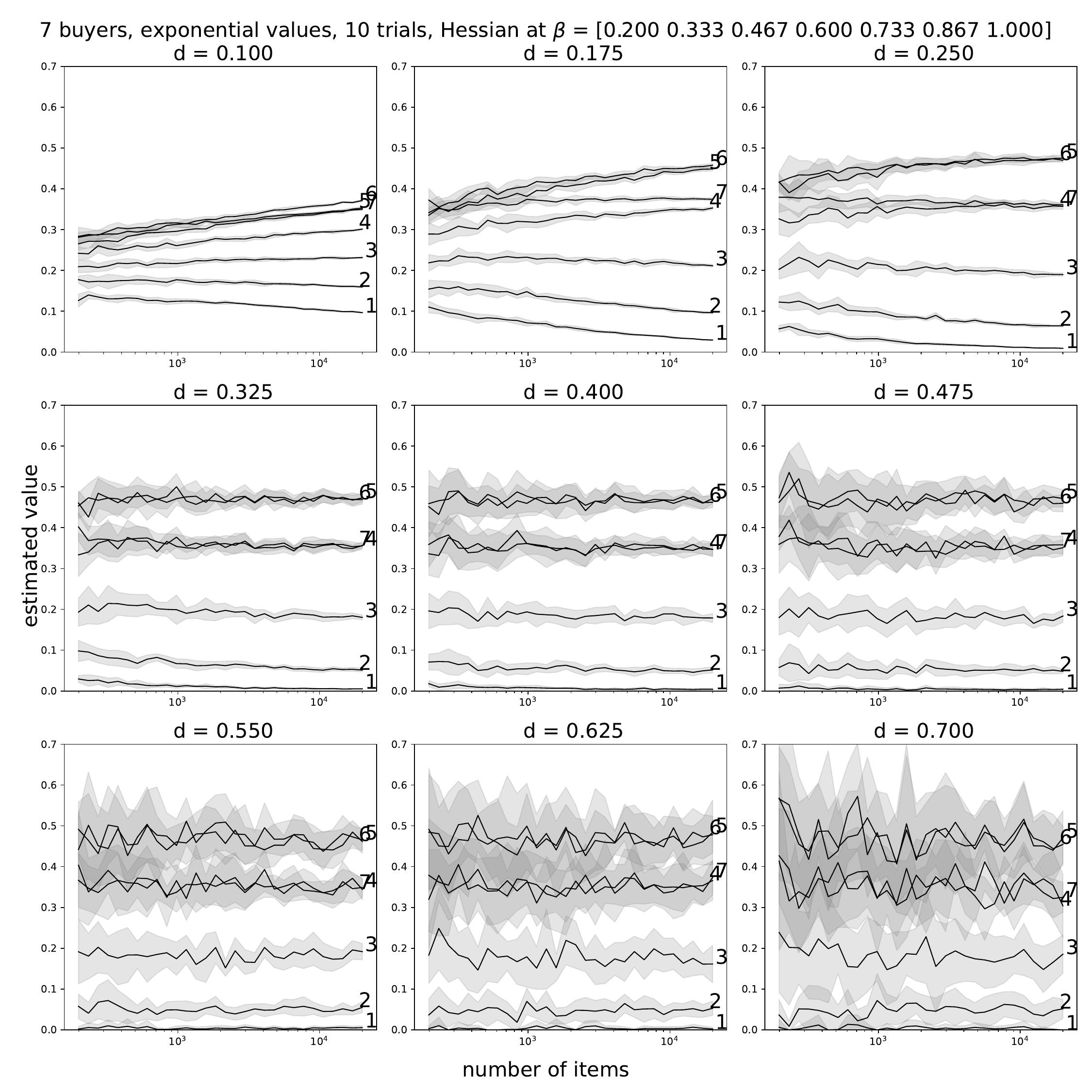}    
        \caption{Effect of smoothing parameter on numerical difference estimation of Hessian. Each curve represents the estimated value of $\cH_{ii}$. Exponential values.}
        \label{fig:smooth_exponential}
    \end{figure}

    \begin{figure}[h!]
        \center
        \includegraphics[scale=.4]{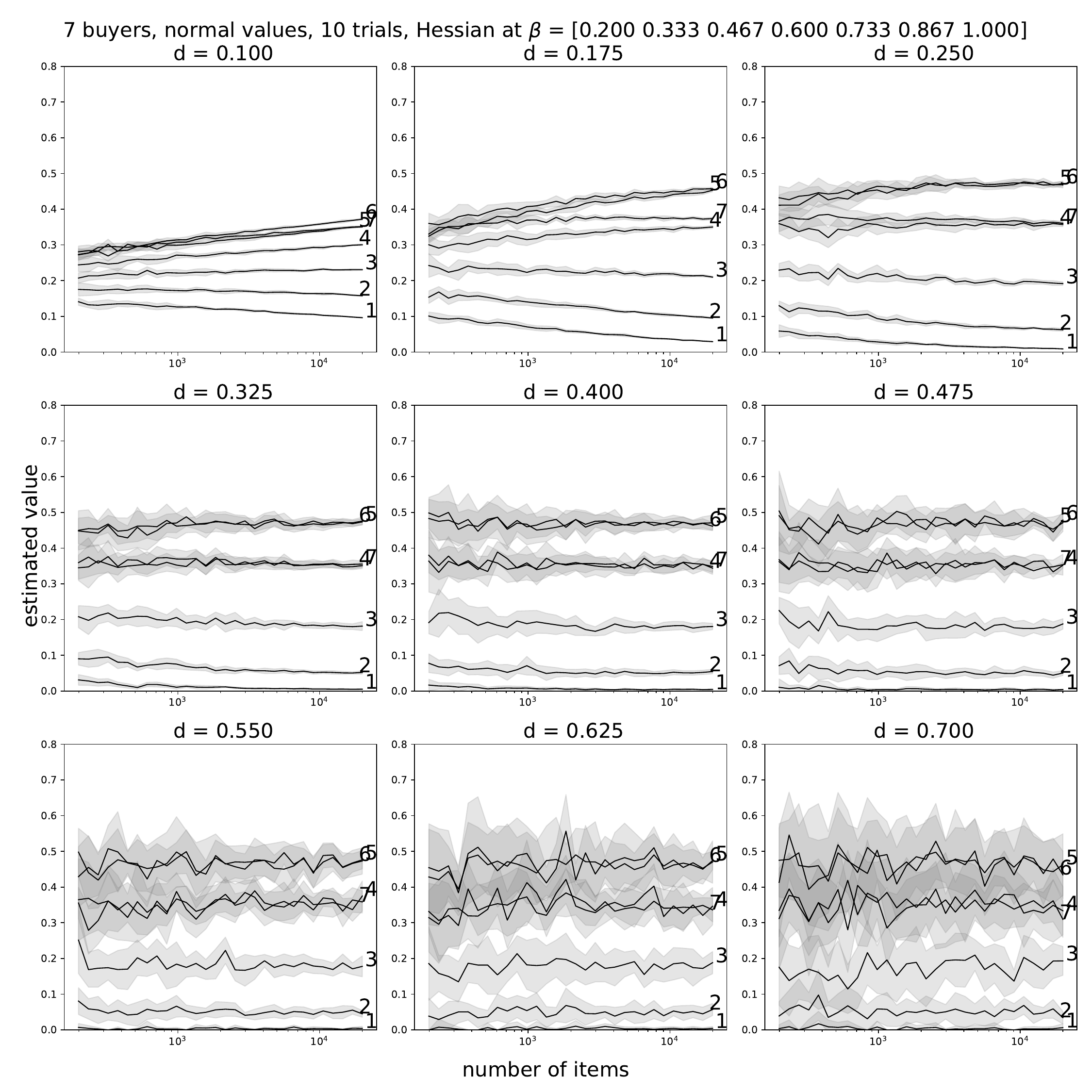}    
        \caption{Effect of smoothing parameter on numerical difference estimation of Hessian. Each curve represents the estimated value of $\cH_{ii}$. Truncated normal values.}
        \label{fig:smooth_normal}
    \end{figure}

\subsection{Visualization of FPPE Distribution}

\begin{figure}[h!]
    \center
    \includegraphics[scale=.6]{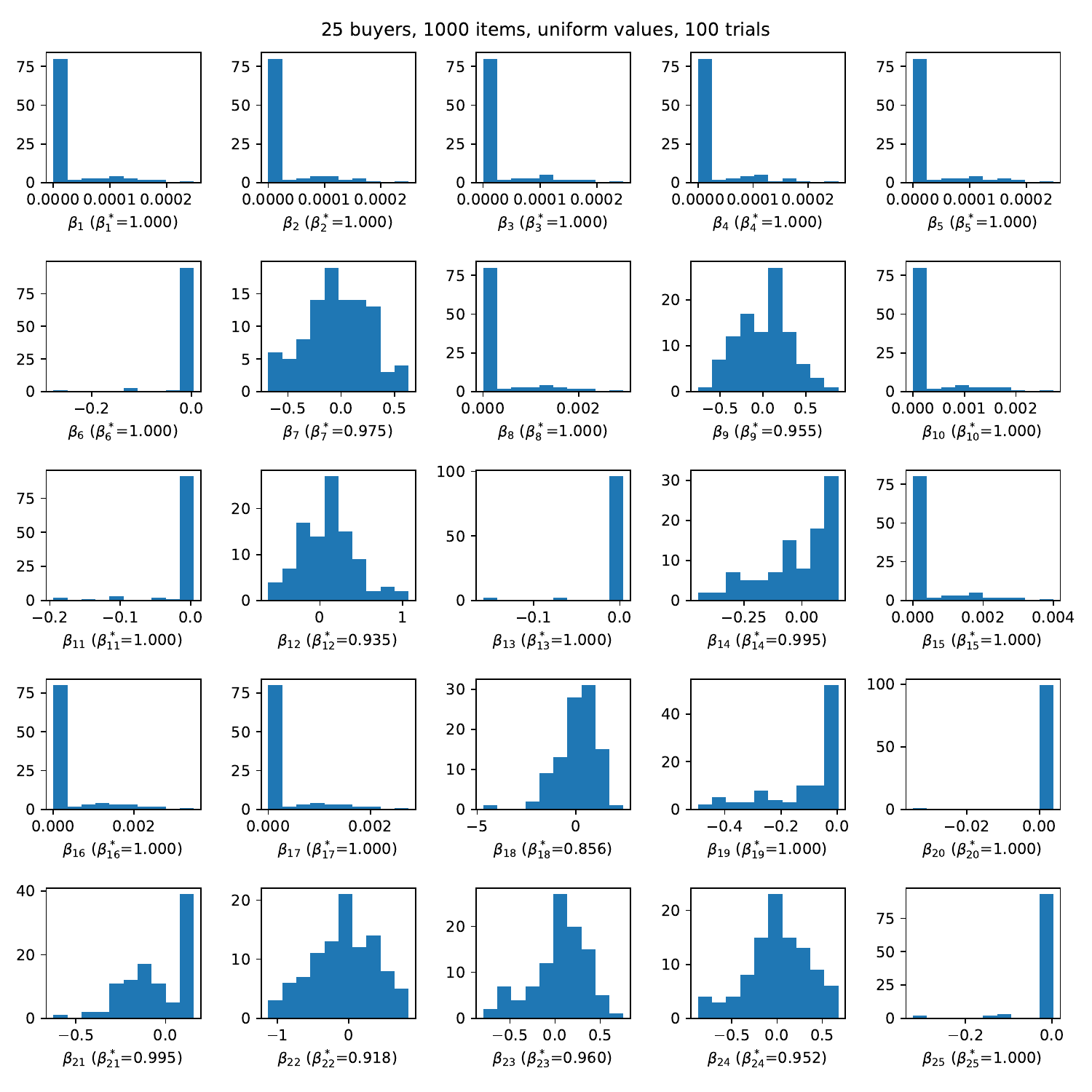}    
    \caption{Distribution of $\sqrt t{(\betagami - \betasti)}$ for all $i \in [n]$.    
    Non-normality and fast convergence for buyers with $\betasti = 1$. Uniform values.}
    \label{fig:nonnormality_unif}
\end{figure}

\begin{figure}[h!]
    \center
    \includegraphics[scale=.6]{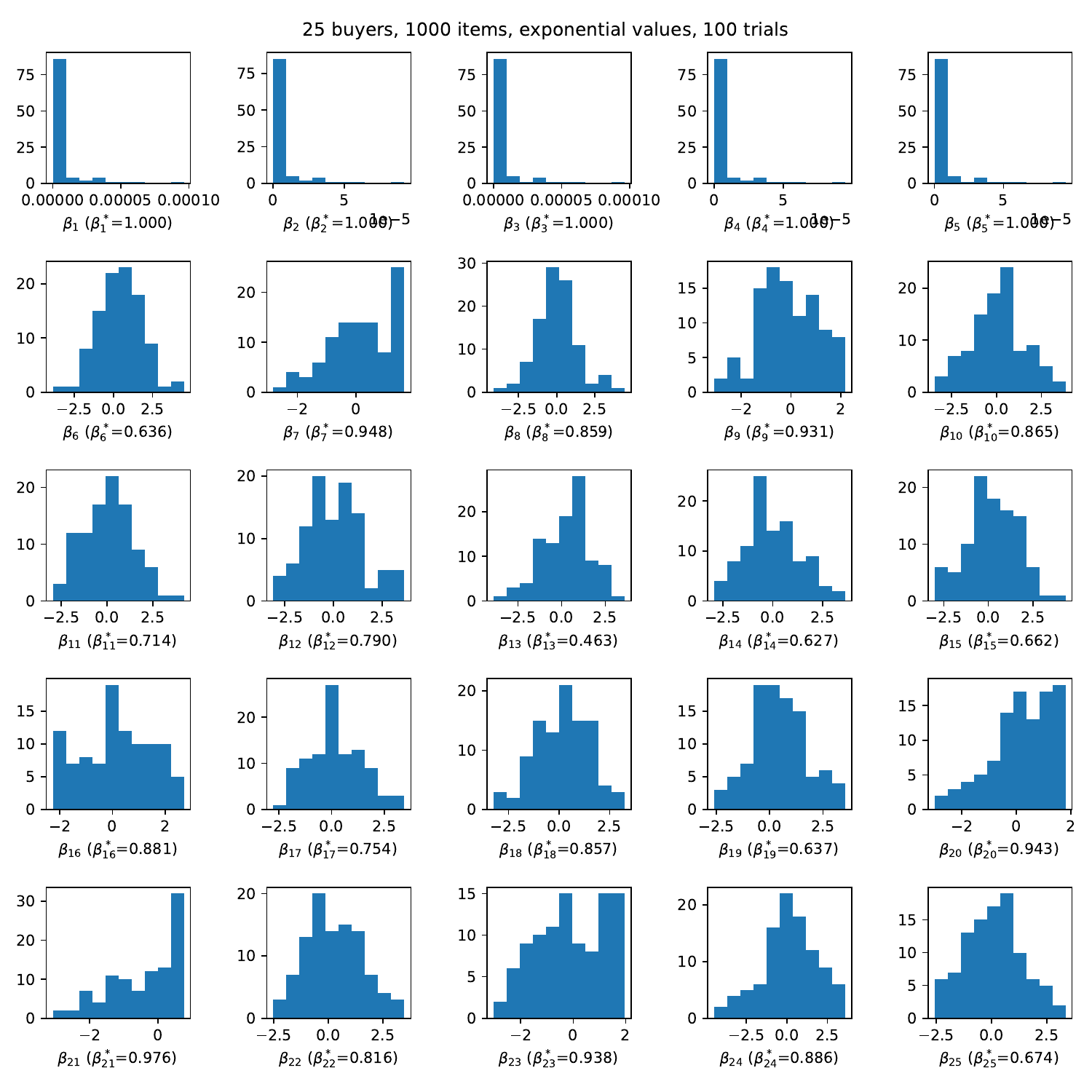}    
    \caption{Distribution of $\sqrt t{(\betagami - \betasti)}$ for all $i \in [n]$.    
    Non-normality and fast convergence for buyers with $\betasti = 1$. Exponential values.}
    \label{fig:nonnormality_expo}
\end{figure}

\begin{figure}[h!]
    \center
    \includegraphics[scale=.6]{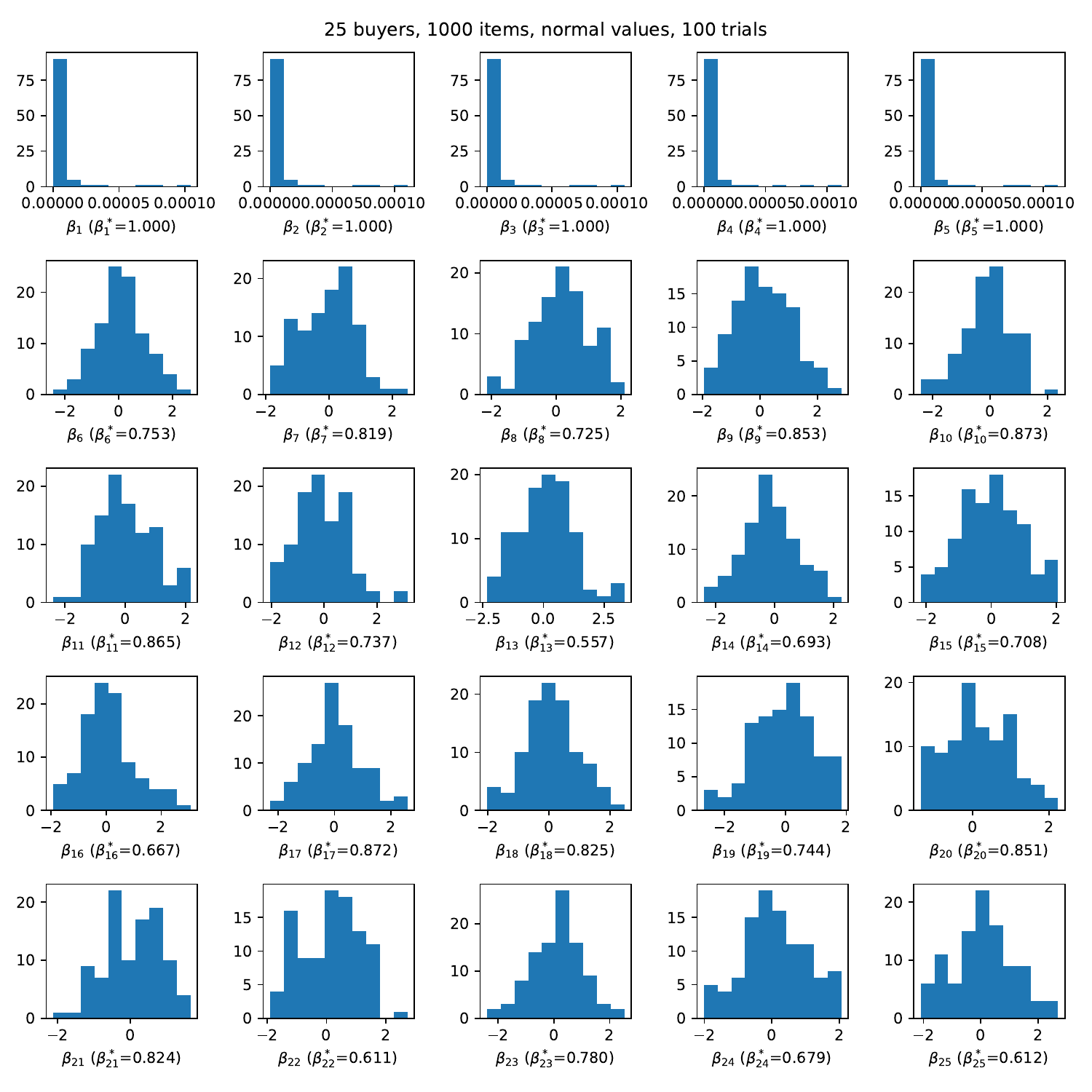}    
    \caption{
    Distribution of $\sqrt t({\betagami - \betasti})$ for all $i \in [n]$.    
    Non-normality and fast convergence for buyers with $\betasti = 1$. Truncated normal values.}
    \label{fig:nonnormality_normal}
\end{figure}

\subsection{Full Versions of \cref{tbl:rev_coverage_short},  \cref{tbl:rev_coverage_new_vs_naive_short}, and  \cref{tbl:abtest_short}}

\input{tbl_abtest.tex}

\input{tbl_rev_coverage.tex}

\input{tbl_rev_coverage_comparison.tex}

%% file: tbl_abtest.tex
\begin{table}[ht]
\centering
\caption{Coverage of treatment effect. $\pi$ = treatment probability, the finite difference stepsize $\epsilon_t = t^{-0.4}$, proportion of unpaced buyers $\beta_i=1$ is 30\%. \new{The numbers in parentheses represent the lengths of the confidence intervals. Nominal coverage rate is 90\%.}}
\label{tbl:abtest_full}
\begin{tabular}{llllll}
\toprule
   & items &                      100 &                      200 &                      400 &                      600 \\
\midrule buyers & $\pi$ &                          &                          &                          &                          \\
\midrule
\multirow{5}{*}{20} & 0.1 &  \makecell{0.8\\(5.09)} &  \makecell{0.89\\(4.96)} &  \makecell{0.89\\(2.71)} &  \makecell{0.92\\(1.60)} \\
   & 0.3 &  \makecell{0.91\\(3.43)} &  \makecell{0.91\\(1.74)} &  \makecell{0.93\\(1.12)} &  \makecell{0.92\\(0.94)} \\
   & 0.5 &  \makecell{0.91\\(3.53)} &  \makecell{0.92\\(1.47)} &  \makecell{0.92\\(1.07)} &  \makecell{0.92\\(0.86)} \\
   & 0.7 &  \makecell{0.92\\(3.76)} &  \makecell{0.88\\(1.96)} &  \makecell{0.89\\(1.16)} &   \makecell{0.9\\(0.96)} \\
   & 0.9 &  \makecell{0.79\\(3.50)} &  \makecell{0.89\\(3.57)} &  \makecell{0.96\\(2.75)} &  \makecell{0.91\\(2.10)} \\
\cline{1-6}
\multirow{5}{*}{30} & 0.1 &  \makecell{0.85\\(4.20)} &  \makecell{0.91\\(3.07)} &  \makecell{0.9\\(25.55)} &  \makecell{0.92\\(2.98)} \\
   & 0.3 &  \makecell{0.95\\(4.26)} &  \makecell{0.93\\(4.42)} &  \makecell{0.94\\(1.38)} &  \makecell{0.93\\(1.03)} \\
   & 0.5 &  \makecell{0.92\\(2.80)} &  \makecell{0.93\\(1.86)} &  \makecell{0.92\\(1.89)} &  \makecell{0.84\\(0.97)} \\
   & 0.7 &  \makecell{0.92\\(3.91)} &  \makecell{0.94\\(3.47)} &  \makecell{0.91\\(1.40)} &  \makecell{0.89\\(1.09)} \\
   & 0.9 &  \makecell{0.86\\(5.09)} &  \makecell{0.89\\(4.70)} &  \makecell{0.92\\(2.58)} &  \makecell{0.88\\(3.15)} \\
\cline{1-6}
\multirow{5}{*}{40} & 0.1 &  \makecell{0.86\\(4.09)} &  \makecell{0.91\\(3.52)} &  \makecell{0.92\\(2.00)} &  \makecell{0.87\\(2.08)} \\
   & 0.3 &  \makecell{0.93\\(6.37)} &  \makecell{0.94\\(4.67)} &  \makecell{0.93\\(1.26)} &  \makecell{0.93\\(1.11)} \\
   & 0.5 &  \makecell{0.94\\(9.78)} &   \makecell{0.9\\(2.17)} &  \makecell{0.87\\(1.21)} &  \makecell{0.96\\(1.31)} \\
   & 0.7 &  \makecell{0.91\\(9.30)} &  \makecell{0.96\\(6.72)} &  \makecell{0.94\\(1.73)} &  \makecell{0.89\\(1.17)} \\
   & 0.9 &  \makecell{0.85\\(5.00)} &  \makecell{0.85\\(5.87)} &   \makecell{0.9\\(6.01)} &  \makecell{0.97\\(6.13)} \\
\cline{1-6}
\multirow{5}{*}{50} & 0.1 &  \makecell{0.88\\(8.45)} &  \makecell{0.88\\(4.75)} &  \makecell{0.92\\(1.87)} &  \makecell{0.88\\(1.54)} \\
   & 0.3 &  \makecell{0.95\\(3.49)} &  \makecell{0.96\\(2.37)} &  \makecell{0.86\\(1.28)} &  \makecell{0.93\\(1.03)} \\
   & 0.5 &  \makecell{0.92\\(3.76)} &  \makecell{0.95\\(5.05)} &   \makecell{0.9\\(1.26)} &  \makecell{0.94\\(0.98)} \\
   & 0.7 &  \makecell{0.85\\(3.10)} &  \makecell{0.95\\(2.27)} &  \makecell{0.98\\(1.85)} &  \makecell{0.95\\(1.28)} \\
   & 0.9 &  \makecell{0.76\\(3.36)} &  \makecell{0.87\\(2.87)} &  \makecell{0.92\\(2.78)} &  \makecell{0.96\\(9.33)} \\
\bottomrule
\end{tabular}
\end{table}

%% file: tbl_rev_coverage.tex
\begin{table}[ht]
\centering
\caption{Coverage of revenue CI. $\alpha$ = proportion of $\beta_i = 1$, $d$ is the exponent in finite difference stepsize $\epsilon_t = t^{-d}$. \new{
Numbers in parentheses represent the lengths of CIs.
Nominal coverage rate is 90\%.
}}
\label{tbl:rev_coverage_full}
\footnotesize
\begin{tabular}{llllllllllllll}
\toprule
         & buyers & \multicolumn{4}{l}{20} & \multicolumn{4}{l}{50} & \multicolumn{4}{l}{80} \\
         & $\alpha$ &                     0.05 &                     0.10 &                     0.20 &                     0.30 &                     0.05 &                     0.10 &                     0.20 &                     0.30 &                     0.05 &                     0.10 &                     0.20 &                     0.30 \\
\midrule $d$ & items &                          &                          &                          &                          &                          &                          &                          &                          &                          &                          &                          &                          \\
\midrule
\multirow{4}{*}{0.40} & 100 &  \makecell{0.79\\(1.72)} &  \makecell{0.79\\(1.82)} &  \makecell{0.93\\(1.87)} &   \makecell{0.9\\(1.81)} &  \makecell{0.87\\(1.84)} &  \makecell{0.81\\(1.91)} &  \makecell{0.89\\(1.82)} &  \makecell{0.88\\(2.00)} &  \makecell{0.81\\(1.89)} &  \makecell{0.89\\(1.89)} &  \makecell{0.97\\(1.97)} &   \makecell{0.9\\(1.95)} \\
         & 200 &  \makecell{0.88\\(1.33)} &  \makecell{0.88\\(1.36)} &  \makecell{0.87\\(1.35)} &   \makecell{0.9\\(1.34)} &  \makecell{0.88\\(1.32)} &  \makecell{0.93\\(1.36)} &  \makecell{0.89\\(1.37)} &  \makecell{0.94\\(1.40)} &  \makecell{0.87\\(1.37)} &  \makecell{0.93\\(1.37)} &  \makecell{0.88\\(1.42)} &  \makecell{0.86\\(1.42)} \\
         & 400 &  \makecell{0.84\\(0.93)} &  \makecell{0.88\\(0.99)} &  \makecell{0.93\\(0.99)} &  \makecell{0.91\\(0.98)} &   \makecell{0.9\\(0.95)} &  \makecell{0.94\\(0.98)} &  \makecell{0.92\\(0.98)} &  \makecell{0.84\\(1.00)} &  \makecell{0.88\\(0.97)} &  \makecell{0.85\\(0.98)} &  \makecell{0.86\\(1.01)} &  \makecell{0.85\\(1.01)} \\
         & 600 &  \makecell{0.89\\(0.76)} &  \makecell{0.88\\(0.79)} &   \makecell{0.9\\(0.81)} &  \makecell{0.89\\(0.80)} &   \makecell{0.8\\(0.77)} &  \makecell{0.87\\(0.80)} &  \makecell{0.81\\(0.80)} &  \makecell{0.92\\(0.83)} &  \makecell{0.86\\(0.80)} &  \makecell{0.83\\(0.80)} &  \makecell{0.97\\(0.83)} &  \makecell{0.89\\(0.83)} \\
\cline{1-14}
\multirow{4}{*}{0.33} & 100 &   \makecell{0.8\\(1.75)} &   \makecell{0.8\\(1.88)} &   \makecell{0.9\\(1.89)} &  \makecell{0.87\\(1.87)} &   \makecell{0.9\\(1.82)} &  \makecell{0.88\\(1.93)} &  \makecell{0.87\\(1.84)} &  \makecell{0.86\\(1.96)} &  \makecell{0.88\\(1.88)} &  \makecell{0.88\\(1.90)} &  \makecell{0.87\\(1.97)} &  \makecell{0.93\\(1.99)} \\
         & 200 &  \makecell{0.85\\(1.29)} &  \makecell{0.84\\(1.37)} &  \makecell{0.91\\(1.38)} &  \makecell{0.84\\(1.38)} &  \makecell{0.88\\(1.31)} &  \makecell{0.93\\(1.38)} &  \makecell{0.89\\(1.38)} &  \makecell{0.87\\(1.43)} &  \makecell{0.89\\(1.38)} &  \makecell{0.96\\(1.37)} &  \makecell{0.89\\(1.42)} &   \makecell{0.9\\(1.42)} \\
         & 400 &  \makecell{0.87\\(0.93)} &  \makecell{0.92\\(0.97)} &  \makecell{0.91\\(0.99)} &  \makecell{0.84\\(0.98)} &  \makecell{0.88\\(0.95)} &  \makecell{0.93\\(0.98)} &  \makecell{0.88\\(0.97)} &  \makecell{0.85\\(1.00)} &  \makecell{0.95\\(0.98)} &  \makecell{0.86\\(0.96)} &  \makecell{0.86\\(1.01)} &  \makecell{0.86\\(1.01)} \\
         & 600 &   \makecell{0.9\\(0.76)} &  \makecell{0.89\\(0.80)} &  \makecell{0.84\\(0.80)} &  \makecell{0.83\\(0.80)} &  \makecell{0.87\\(0.78)} &  \makecell{0.87\\(0.80)} &  \makecell{0.88\\(0.81)} &  \makecell{0.85\\(0.83)} &   \makecell{0.9\\(0.80)} &   \makecell{0.9\\(0.80)} &  \makecell{0.84\\(0.83)} &  \makecell{0.89\\(0.83)} \\
\cline{1-14}
\multirow{4}{*}{0.17} & 100 &  \makecell{0.81\\(1.80)} &  \makecell{0.86\\(1.85)} &  \makecell{0.89\\(1.85)} &   \makecell{0.9\\(1.87)} &  \makecell{0.81\\(1.86)} &  \makecell{0.84\\(1.87)} &  \makecell{0.92\\(1.91)} &   \makecell{0.9\\(1.96)} &  \makecell{0.93\\(1.87)} &  \makecell{0.87\\(1.85)} &  \makecell{0.89\\(1.97)} &  \makecell{0.92\\(1.95)} \\
         & 200 &  \makecell{0.91\\(1.30)} &   \makecell{0.9\\(1.37)} &  \makecell{0.88\\(1.38)} &  \makecell{0.86\\(1.37)} &  \makecell{0.85\\(1.32)} &  \makecell{0.87\\(1.38)} &   \makecell{0.9\\(1.36)} &  \makecell{0.86\\(1.41)} &   \makecell{0.9\\(1.37)} &  \makecell{0.94\\(1.38)} &  \makecell{0.85\\(1.41)} &   \makecell{0.9\\(1.43)} \\
         & 400 &  \makecell{0.89\\(0.93)} &  \makecell{0.89\\(0.97)} &  \makecell{0.89\\(0.97)} &  \makecell{0.91\\(0.98)} &  \makecell{0.89\\(0.95)} &  \makecell{0.88\\(0.99)} &  \makecell{0.91\\(0.97)} &  \makecell{0.89\\(1.01)} &  \makecell{0.96\\(0.98)} &  \makecell{0.88\\(0.97)} &  \makecell{0.87\\(1.02)} &   \makecell{0.9\\(1.02)} \\
         & 600 &  \makecell{0.89\\(0.76)} &  \makecell{0.89\\(0.80)} &  \makecell{0.84\\(0.80)} &  \makecell{0.84\\(0.80)} &   \makecell{0.9\\(0.78)} &  \makecell{0.89\\(0.81)} &  \makecell{0.82\\(0.80)} &  \makecell{0.91\\(0.83)} &  \makecell{0.89\\(0.80)} &   \makecell{0.9\\(0.81)} &  \makecell{0.91\\(0.83)} &  \makecell{0.89\\(0.83)} \\
\cline{1-14}
\multirow{4}{*}{0.08} & 100 &  \makecell{0.87\\(1.80)} &  \makecell{0.85\\(1.80)} &  \makecell{0.87\\(1.91)} &  \makecell{0.83\\(1.87)} &  \makecell{0.86\\(1.80)} &  \makecell{0.95\\(1.90)} &  \makecell{0.86\\(1.89)} &  \makecell{0.87\\(1.95)} &  \makecell{0.89\\(1.88)} &  \makecell{0.87\\(1.86)} &  \makecell{0.88\\(1.94)} &  \makecell{0.89\\(1.94)} \\
         & 200 &  \makecell{0.92\\(1.30)} &  \makecell{0.87\\(1.34)} &   \makecell{0.9\\(1.36)} &  \makecell{0.92\\(1.35)} &  \makecell{0.85\\(1.35)} &  \makecell{0.88\\(1.37)} &  \makecell{0.93\\(1.36)} &  \makecell{0.92\\(1.42)} &  \makecell{0.87\\(1.39)} &  \makecell{0.83\\(1.36)} &  \makecell{0.92\\(1.40)} &  \makecell{0.91\\(1.47)} \\
         & 400 &   \makecell{0.9\\(0.93)} &  \makecell{0.91\\(0.97)} &  \makecell{0.93\\(0.99)} &  \makecell{0.88\\(0.97)} &  \makecell{0.86\\(0.95)} &  \makecell{0.91\\(0.99)} &  \makecell{0.89\\(0.98)} &  \makecell{0.91\\(1.01)} &   \makecell{0.9\\(0.97)} &  \makecell{0.89\\(0.97)} &  \makecell{0.86\\(1.01)} &  \makecell{0.89\\(1.02)} \\
         & 600 &   \makecell{0.9\\(0.76)} &  \makecell{0.82\\(0.80)} &  \makecell{0.93\\(0.80)} &  \makecell{0.95\\(0.81)} &  \makecell{0.91\\(0.78)} &  \makecell{0.88\\(0.81)} &  \makecell{0.93\\(0.81)} &  \makecell{0.93\\(0.83)} &  \makecell{0.92\\(0.81)} &  \makecell{0.93\\(0.80)} &   \makecell{0.9\\(0.84)} &  \makecell{0.88\\(0.83)} \\
\cline{1-14}
\multirow{4}{*}{0.06} & 100 &  \makecell{0.84\\(1.79)} &  \makecell{0.87\\(1.89)} &  \makecell{0.85\\(1.88)} &  \makecell{0.89\\(1.81)} &  \makecell{0.88\\(1.84)} &  \makecell{0.93\\(1.91)} &  \makecell{0.91\\(1.88)} &  \makecell{0.91\\(1.94)} &  \makecell{0.91\\(1.89)} &  \makecell{0.94\\(1.91)} &  \makecell{0.91\\(1.96)} &  \makecell{0.87\\(2.00)} \\
         & 200 &  \makecell{0.93\\(1.30)} &   \makecell{0.8\\(1.35)} &  \makecell{0.92\\(1.40)} &  \makecell{0.92\\(1.35)} &  \makecell{0.81\\(1.34)} &  \makecell{0.89\\(1.38)} &  \makecell{0.91\\(1.34)} &  \makecell{0.86\\(1.41)} &  \makecell{0.91\\(1.38)} &  \makecell{0.93\\(1.39)} &  \makecell{0.89\\(1.41)} &  \makecell{0.91\\(1.42)} \\
         & 400 &  \makecell{0.94\\(0.93)} &  \makecell{0.87\\(0.96)} &   \makecell{0.8\\(0.97)} &  \makecell{0.91\\(0.97)} &   \makecell{0.8\\(0.94)} &  \makecell{0.93\\(1.00)} &  \makecell{0.82\\(0.97)} &  \makecell{0.96\\(1.02)} &   \makecell{0.9\\(0.98)} &  \makecell{0.91\\(0.99)} &  \makecell{0.85\\(1.01)} &  \makecell{0.91\\(1.03)} \\
         & 600 &  \makecell{0.92\\(0.76)} &   \makecell{0.9\\(0.80)} &  \makecell{0.88\\(0.80)} &   \makecell{0.9\\(0.80)} &  \makecell{0.86\\(0.78)} &  \makecell{0.92\\(0.81)} &  \makecell{0.88\\(0.80)} &  \makecell{0.93\\(0.84)} &  \makecell{0.88\\(0.81)} &  \makecell{0.88\\(0.80)} &  \makecell{0.93\\(0.83)} &  \makecell{0.93\\(0.83)} \\
\bottomrule
\end{tabular}
\end{table}

%% file: tbl_rev_coverage_comparison.tex
\begin{landscape}

\vspace{10cm}
\begin{table}[ht!]
    \caption{Coverage comparison between the Hessian-free CI in \cref{eq:simplified variance estimator}, the naive CI, and Hessian-based CI in \cref{eq:plugin variance}. $\alpha$ = proportion of $\beta_i = 1$. In each cell we present the coverage rate and the average CI widths in parentheses. Nominal rate = 90\%. The number of repetitions in each cell = 500.}
    \label{tbl:rev_coverage_new_vs_naive_full}
    \centering
    \footnotesize
    \begin{tabular}{llllll}
    \toprule
    \
     & items & 100 & 200 & 400 & 800 \\
    buyers & $\alpha$ &  &  &  &  \\
    \midrule
    \multirow[t]{6}{*}{20} & 0.0 & \makecell{{0.98(0.02)}| 1.00(0.67)| 1.00(0.60)} & \makecell{{1.00(0.00)}| 1.00(0.49)| 1.00(0.44)} & \makecell{{1.00(0.00)}| 1.00(0.35)| 1.00(0.32)} & \makecell{{1.00(0.00)}| 1.00(0.25)| 1.00(0.23)} \\
     & 0.05 & \makecell{{0.89(0.93)}| 0.84(0.75)| 0.85(0.76)} & \makecell{{0.93(0.69)}| 0.83(0.54)| 0.83(0.55)} & \makecell{{0.93(0.50)}| 0.84(0.39)| 0.85(0.39)} & \makecell{{0.93(0.35)}| 0.85(0.27)| 0.85(0.28)} \\
     & 0.1 & \makecell{{0.92(1.05)}| 0.89(0.94)| 0.88(0.92)} & \makecell{{0.95(0.75)}| 0.90(0.67)| 0.90(0.66)} & \makecell{{0.93(0.53)}| 0.91(0.47)| 0.91(0.47)} & \makecell{{0.95(0.38)}| 0.92(0.33)| 0.92(0.33)} \\
     & 0.2 & \makecell{{0.93(0.98)}| 0.92(0.96)| 0.90(0.94)} & \makecell{{0.89(0.70)}| 0.89(0.68)| 0.89(0.68)} & \makecell{{0.93(0.50)}| 0.92(0.48)| 0.92(0.48)} & \makecell{{0.93(0.37)}| 0.91(0.34)| 0.91(0.34)} \\
     & 0.3 & \makecell{{0.92(1.01)}| 0.91(0.94)| 0.91(0.93)} & \makecell{{0.92(0.72)}| 0.89(0.67)| 0.89(0.67)} & \makecell{{0.91(0.50)}| 0.89(0.47)| 0.89(0.47)} & \makecell{{0.91(0.36)}| 0.89(0.34)| 0.89(0.34)} \\
     & 1.0 & \makecell{{0.90(1.29)}| 0.88(1.21)| 0.88(1.18)} & \makecell{{0.91(0.87)}| 0.90(0.85)| 0.90(0.84)} & \makecell{{0.92(0.61)}| 0.92(0.61)| 0.92(0.61)} & \makecell{{0.90(0.43)}| 0.90(0.43)| 0.90(0.43)} \\
    \cline{1-6}
    \multirow[t]{6}{*}{50} & 0.0 & \makecell{{0.94(0.03)}| 1.00(0.73)| 1.00(0.69)} & \makecell{{1.00(0.01)}| 1.00(0.53)| 1.00(0.50)} & \makecell{{1.00(0.00)}| 1.00(0.38)| 1.00(0.36)} & \makecell{{1.00(0.00)}| 1.00(0.27)| 1.00(0.26)} \\
     & 0.05 & \makecell{{0.82(1.20)}| 0.88(0.92)| 0.87(0.91)} & \makecell{{0.93(0.79)}| 0.89(0.66)| 0.89(0.66)} & \makecell{{0.95(0.60)}| 0.90(0.47)| 0.90(0.47)} & \makecell{{0.97(0.41)}| 0.92(0.33)| 0.92(0.33)} \\
     & 0.1 & \makecell{{0.92(1.03)}| 0.90(0.96)| 0.89(0.93)} & \makecell{{0.94(0.78)}| 0.89(0.67)| 0.88(0.67)} & \makecell{{0.93(0.54)}| 0.89(0.48)| 0.90(0.48)} & \makecell{{0.95(0.40)}| 0.90(0.34)| 0.90(0.34)} \\
     & 0.2 & \makecell{{0.93(1.08)}| 0.91(0.96)| 0.90(0.94)} & \makecell{{0.91(0.75)}| 0.89(0.68)| 0.89(0.67)} & \makecell{{0.93(0.53)}| 0.89(0.48)| 0.89(0.48)} & \makecell{{0.92(0.37)}| 0.90(0.34)| 0.90(0.34)} \\
     & 0.3 & \makecell{{0.92(1.27)}| 0.90(1.15)| 0.90(1.13)} & \makecell{{0.90(0.85)}| 0.89(0.82)| 0.88(0.81)} & \makecell{{0.93(0.60)}| 0.92(0.58)| 0.92(0.58)} & \makecell{{0.91(0.42)}| 0.90(0.41)| 0.90(0.41)} \\
     & 1.0 & \makecell{{0.90(1.28)}| 0.89(1.23)| 0.89(1.21)} & \makecell{{0.90(0.90)}| 0.90(0.88)| 0.89(0.88)} & \makecell{{0.87(0.63)}| 0.87(0.63)| 0.87(0.62)} & \makecell{{0.92(0.44)}| 0.92(0.44)| 0.92(0.44)} \\
    \cline{1-6}
    \multirow[t]{6}{*}{80} & 0.0 & \makecell{{0.88(0.04)}| 1.00(0.75)| 1.00(0.72)} & \makecell{{0.98(0.02)}| 1.00(0.54)| 1.00(0.52)} & \makecell{{1.00(0.00)}| 1.00(0.39)| 1.00(0.37)} & \makecell{{1.00(0.00)}| 1.00(0.28)| 1.00(0.27)} \\
     & 0.05 & \makecell{{0.93(1.05)}| 0.90(0.95)| 0.90(0.94)} & \makecell{{0.94(0.77)}| 0.90(0.68)| 0.90(0.67)} & \makecell{{0.95(0.58)}| 0.87(0.48)| 0.87(0.48)} & \makecell{{0.94(0.38)}| 0.90(0.34)| 0.90(0.34)} \\
     & 0.1 & \makecell{{0.95(1.15)}| 0.90(0.96)| 0.90(0.94)} & \makecell{{0.92(0.76)}| 0.89(0.68)| 0.89(0.68)} & \makecell{{0.96(0.59)}| 0.91(0.48)| 0.91(0.48)} & \makecell{{0.95(0.39)}| 0.90(0.34)| 0.90(0.34)} \\
     & 0.2 & \makecell{{0.92(1.28)}| 0.90(1.14)| 0.89(1.12)} & \makecell{{0.90(0.85)}| 0.89(0.81)| 0.89(0.81)} & \makecell{{0.91(0.59)}| 0.90(0.58)| 0.90(0.58)} & \makecell{{0.92(0.43)}| 0.91(0.41)| 0.91(0.41)} \\
     & 0.3 & \makecell{{0.90(1.27)}| 0.89(1.19)| 0.89(1.17)} & \makecell{{0.91(0.88)}| 0.91(0.85)| 0.91(0.85)} & \makecell{{0.90(0.61)}| 0.89(0.60)| 0.89(0.60)} & \makecell{{0.89(0.43)}| 0.88(0.43)| 0.88(0.43)} \\
     & 1.0 & \makecell{{0.90(1.33)}| 0.88(1.27)| 0.88(1.25)} & \makecell{{0.90(0.91)}| 0.90(0.90)| 0.90(0.90)} & \makecell{{0.90(0.64)}| 0.90(0.64)| 0.90(0.64)} & \makecell{{0.87(0.45)}| 0.87(0.45)| 0.87(0.45)} \\
    \cline{1-6}
    
    \end{tabular}
    \end{table}
    
\end{landscape}